\documentclass[12pt]{article}
\pdfoutput=1 % if your are submitting a pdflatex (i.e. if you have
\usepackage[margin=1in]{geometry} % Adjust margin size here
\usepackage[font=small]{caption}
\usepackage{amsmath}      % For \underset and \smash
\usepackage{mathtools}    % For \underbracket
\usepackage{amssymb}
\usepackage[dvipsnames]{xcolor} % optional, for named/HTML colors
\usepackage{graphicx,tikz} 
\usetikzlibrary{calc}

\usepackage{cite}
\usepackage{authblk}

\definecolor{CiteRed}{HTML}{C00000} % custom deep red
\definecolor{EqRefBlue}{HTML}{1F4B99}
\definecolor{CiteGreen}{HTML}{0B7A0B}

\usepackage[
  colorlinks=true,      % color the text of links instead of boxes
  linkcolor=EqRefBlue,  % \eqref, \ref, \autoref, etc.
  citecolor=CiteRed,  % \cite, \parencite, etc.
  urlcolor=MidnightBlue % optional
]{hyperref}

\usepackage{tabularx}   % flexible-width columns
\usepackage{booktabs}   % nicer rules
\usepackage{array}      % 
\newcolumntype{C}{>{\centering\arraybackslash}X}  % flexible width centered
\newcolumntype{P}[1]{>{\centering\arraybackslash}p{#1}} % fixed width centered
\newcolumntype{M}[1]{>{\centering\arraybackslash}m{#1}}

\newcommand{\otoc}{{\rm OTOC}}

\newcommand{\JW}[1]{\textcolor{teal}{[JW: #1]}}

\usepackage{tensor}
\usepackage{braket}
\usepackage{comment}
\usepackage{subcaption}
\usepackage{dsfont}
\usepackage{bbm}

\usepackage{upgreek}

\usepackage[normalem]{ulem}
\usepackage{amsmath}

\numberwithin{equation}{section}

\usepackage{enumerate}
\usepackage{epsfig}
\usepackage{mathbbol}

\def\Re{\mathop{\rm Re} }

\newcommand{\sU}{\mathcal{U}}

\newcommand{\OTOC}{{\rm OTOC}}

\newcommand{\be}{\begin{equation}}
\newcommand{\ee}{\end{equation}}
\newcommand{\bea}{\begin{eqnarray}}
\newcommand{\eea}{\end{eqnarray}}
\newcommand{\bega}{\begin{gather}}
\newcommand{\eega}{\end{gather}}
\newcommand{\nn}{\nonumber\\}
\newcommand{\bi}{\begin{itemize}}
\newcommand{\ei}{\end{itemize}}
\newcommand{\ben}{\begin{enumerate}}
\newcommand{\een}{\end{enumerate}}
\newcommand{\bca}{\begin{cases}}
\newcommand{\eca}{\end{cases}}
\newcommand{\bln}{\begin{align}}
\newcommand{\eln}{\end{align}}
\newcommand{\bst}{\begin{split}}
\newcommand{\est}{\end{split}}
\def\ie{\begin{equation}\begin{aligned}}
\def\fe{\end{aligned}\end{equation}}
\newcommand{\bma}{\le(\begin{matrix}}
\newcommand{\ema}{\end{matrix}\ri)}
\newcommand{\bwt}{\begin{widetext}}
\newcommand{\ewt}{\end{widetext}}

\newcommand\dd{{\rm d}}

\newcommand\vep{\varepsilon}

\def\le{\left}
\def\ri{\right}

\newcommand\sB{{\ensuremath{{\mathcal B}}}}
\newcommand\sC{{\ensuremath{{\mathcal C}}}}

\newcommand\sE{{\ensuremath{{\mathcal E}}}}

\newcommand\sI{{\ensuremath{{\mathcal I}}}}

\newcommand\sH{{\ensuremath{{\mathcal H}}}}

\newcommand\sM{{\ensuremath{{\mathcal M}}}}
\newcommand\sN{{\ensuremath{{\mathcal N}}}}

\newcommand\sR{{\mathcal R}}
\newcommand\sS{{\mathcal S}}

\newcommand\sW{{\mathcal W}}
\newcommand\sX{{\mathcal X}}
\newcommand\sY{{\mathcal Y}}
\newcommand\sZ{{\mathcal Z}}

\newcommand{\Tr}{\text{Tr}}

\usepackage{amsthm}
\usepackage{thmtools, thm-restate} % <—

\declaretheorem[name=Lemma, numberwithin=section]{lemma}
\newtheorem{theorem}{Theorem}[section]
\newtheorem{definition}{Definition}[section]

\newtheorem{remark}[definition]{Remark}

\newcommand{\rev}[1]{#1}

\usepackage[percent]{overpic}

\begin{document}

\title{\fontsize{17}{20}\selectfont  \textbf{Free mutual information and higher-point OTOCs}}

\author[1]{\normalsize Shreya Vardhan\thanks{svardhan@caltech.edu}}
\author[2]{\normalsize Jinzhao Wang\thanks{wang.jinzhao226@gmail.com}}
\affil[1]{\normalsize Institute for Quantum Information and Matter, California Institute of Technology, Pasadena, CA 91125}
\affil[2]{\normalsize Leinweber Institute for Theoretical Physics, Stanford University, Stanford, CA 94305}
\date{}

\maketitle
\begin{abstract}
We introduce a quantity called the free mutual information (FMI), adapted from concepts in free probability theory, as a new physical measure of quantum chaos. This quantity captures the spreading of a time-evolved operator in the space of all possible operators on the Hilbert space, which is doubly exponential in the number of degrees of freedom. It thus provides a finer notion of operator spreading than the well-understood phenomenon of operator growth in physical space. We derive two central results which apply in any physical system: first, an explicit ``Coulomb gas’’ formula for the FMI of two observables $A(t)$ and $B$ in terms of the eigenvalues of the product operator $A(t)B$; and second, a general relation expressing the FMI as a weighted sum of all higher-point out-of-time-ordered correlators (OTOCs). This second result provides a precise information-theoretic  interpretation for the higher-point OTOCs as collectively quantifying operator ergodicity and the approach to  freeness. This physical  interpretation is particularly useful in light of recent progress in experimentally measuring higher-point OTOCs. We identify universal  behaviours of the FMI and higher-point OTOCs across a variety of chaotic systems, including random unitary circuits and chaotic spin chains, which indicate that spreading in the doubly exponential  operator space is a generic feature of quantum many-body chaos. At the same time, the non-generic  behavior of the FMI in various non-chaotic systems, including certain unitary designs, shows that there are cases where an operator spreads in physical space but remains localized in operator space. The FMI is thus a sharper diagnostic of chaos than the standard 4-point OTOC.

\end{abstract}

\newpage

\tableofcontents

\section{Introduction and summary of results}

Ergodicity is a universal feature of classically chaotic systems. Consider an initially localized distribution in the classical phase space of a  particle. If the particle is governed by a chaotic evolution, this localized distribution spreads out over a large portion of the phase space at late times. More precisely, while the exact  phase space volume of the distribution remains constant, a coarse-grained version of the volume grows with time. In contrast, in integrable systems, the coarse-grained  distribution continues to be localized %occupy a small fraction of the total phase space volume 
even at late times. As a concrete example, one could have in mind the standard map~\cite{standard_map_1, standard_map_2} in its chaotic and integrable regimes.

The idea of spreading of phase space distributions from classical chaos does not generalize in a straightforward way to quantum many-body chaos. In particular, there is no canonical and universal way to define the analog of phase space distributions in general quantum many-body systems. A more useful way of characterizing quantum many-body chaos, which has developed over the last decade, is in terms of universal properties of the Heisenberg evolution of operators. 
One important universal property is the growth of {\it operator size} in chaotic systems with local and $k$-local interactions:  an initial operator supported on a small number of degrees of freedom generically gains support on a large number of degrees of freedom at late times. The decay of the four-point out-of-time-ordered correlator (OTOC)~\cite{larkin} in quantum many-body systems is equivalent to a  precise notion of growth of operator size (see for instance~\cite{size_syk, localized_shocks, swingle_tutorial}). OTOCs have been widely studied in a variety of chaotic quantum many-body systems, starting with initial studies in holographic CFTs and fermionic models in~\cite{butterfly,Shenker_2014,firewalls, kitaev2014fundamental,Maldacena_2016,Maldacena_2016_syk, stringy}. They show simple  behaviors that fall into a small number of universality classes across different chaotic systems~\cite{vdle}(also see~\cite{swingle_tutorial} for a review). 

Our goal in this paper is to introduce a distinct notion of operator spreading in quantum many-body systems, which can be seen as a new quantum many-body analog of ergodicity in classical chaos. Instead of considering the volume occupied by the time-evolved operator $A(t) = e^{-iHt} A e^{iHt}$ in the {\it physical space} of degrees of freedom, we will instead quantify the volume occupied by it in {\it operator space}, or the abstract space of all possible time-evolved operators $U A U^{\dagger}$ for any unitary $U$ that can act on the Hilbert space. 

Suppose we have a system with $n$ local degrees of freedom. While the volume of physical space is $O(n)$, the volume of the operator space parameterized by all possible $U$ acting on these $n$ degrees of freedom is {\it doubly exponential} in $n$. This distinction helps emphasize the fact that spreading in operator space is a different, and more fine-grained, notion than  spreading in physical  space. See Fig.~\ref{fig:physical_vs_op}. In particular, we will discuss examples below of certain systems where an operator is  fully delocalized in physical space while remaining confined to a small, highly structured corner of   operator space. Nevertheless, we will show that spreading in operator space, i.e., the growth of the \emph{volume fraction}
\be 
f_{A(t)} = \frac{\text{Vol}(e^{-iHt}A e^{iHt})}{\text{Vol}(UA U^{\dagger} \text{ for all possible } U)} \, , \label{fa_intro} 
\ee
is a universal property of generic chaotic quantum many-body systems.

\begin{figure}[t]
\centering
\includegraphics[width=\textwidth]{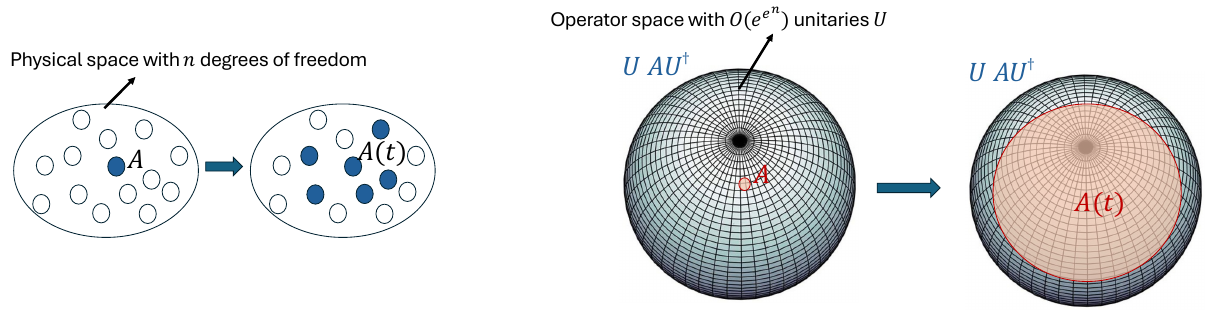}
\caption{We contrast spreading of the time-evolved operator $A(t) =e^{-iHt}A_0e^{iHt}$ in a chaotic quantum many-body system in the physical space of $n$ degrees of freedom, shown on the left, with spreading in the much larger space of all possible time-evolved operators $UA_0U^{\dagger}$, shown on the right. The sphere in the right figure should be seen as living in $O(e^n)$ dimensions, corresponding to the number of parameters in a general unitary acting on the Hilbert space, and a discretization of this sphere consists of $O(e^{e^n})$ operators.} 
\label{fig:physical_vs_op} 
\end{figure}

We have not yet defined precisely what we mean by the volumes in \eqref{fa_intro}. One ingredient needed in order to do so is a formal definition of volumes on the manifold of $d$-dimensional unitaries, which we will introduce in Sec.~\ref{sec:definition}. In addition to this, a further conceptual ingredient is needed in order to associate a non-trivial volume with the fixed operator $A(t)=e^{-iHt}A e^{e^{iHt}}$. Since $A(t)$ is a single operator, it can clearly only occupy a measure zero volume within the set of all possible $UAU^{\dagger}$, unless we coarse-grain it to allow for a larger set of operators. 

We use a coarse-graining prescription inspired by a quantity called the ``free entropy'' in the mathematical literature on noncommutative probability theory~\cite{voiculescu1993analogues,voiculescu1994analogues,voiculescu1996analogues,voiculescu1997analogues,voiculescu1998analogues,voiculescu1999analogues,Voiculescu_2002,hiai2005large,HIAI_2009,Collins_2014}. This coarse-graining procedure requires a choice of some simple reference operator $B$  at $t=0$. We define the coarse-grained set of operators associated with $A(t)$ as the set of operators $\tilde A$ which approximately match all ``joint moments'' between $A(t)$ and $B$, i.e. 
\be 
\Tr[\rho {\tilde A}^{m_1} B^{n_1} {\tilde A}^{m_2} B^{n_2}...  {\tilde A}^{m_r} B^{n_r}] \approx \Tr[\rho {A(t)}^{m_1} B^{n_1} A(t)^{m_2} B^{n_2}...  A(t)^{m_r} B^{n_r}]\label{all_moments_intro}
\ee
for all $m_i$, $n_i$, and some reference state $\rho$. We will typically take $\rho$ to be the maximally mixed state. $B$ can be seen as a reference point in the past  to keep track of the spreading of $A(t)$.  While specifying such joint moments with a complete basis of operators at $t=0$ would restrict $\tilde A$ to be precisely $A(t)$, the condition \eqref{all_moments} with a single operator $B$ generically allows for many choices of $\tilde A$. The volume in the numerator of \eqref{fa_intro} should be seen as the volume of this coarse-grained set of $\tilde A$. We should therefore  introduce a $B$ label in the volume fraction \eqref{fa_intro}, denoting it as $f_{A(t)|B}$. We will study the behaviour of a quantity called the free mutual information (FMI), defined as 
\be 
I_{\rm free}(A(t):B) \equiv - \log f_{A(t)|B}\ . \label{ifree_informal}
\ee
We will provide a more formal version of the definition of $I_{\rm free}$, including more details on the approximation used in \eqref{all_moments}, in Sec.~\ref{sec:definition}. We will also discuss a few different ways of viewing the physical interpretation of this quantity in more detail at the end of Sec.~\ref{sec:definition}.  

The quantity $I_{\rm free}$ defined above is similar to (and motivated by) the  free mutual information in the mathematical literature on free probability~\cite{voiculescu1999analogues,Voiculescu_2002,hiai2005large,HIAI_2009,Collins_2014}.~\footnote{\label{ft:FMI}In fact, there are two definitions of free mutual information in the math literature: one based on matricial microstates known as the orbital free entropy~\cite{HIAI_2009}, and the other based on the free Fisher information of a free liberation process~\cite{voiculescu1999analogues}. They are expected but have not yet been shown to coincide~\cite{Voiculescu_2002}. Here, we are referring to the microstate definition.} The mathematical quantity is defined in terms of abstract non-commuting variables instead of fixed operators $A(t)$ and $B$ living in some Hilbert space, and its definition involves certain asymptotic limits which are not appropriate in a physical setting with a fixed Hilbert space dimension. Since our physical quantity in \eqref{ifree_informal} has a similar form to the mathematical FMI apart from these subtleties about asymptotic limits, we will also refer to it as the FMI (see Sec.~\ref{sec:conclusions} for more discussion on the comparison with the mathematical quantity).

If the coarse-graining procedure defined in \eqref{all_moments_intro} is relevant for characterizing quantum many-body chaos, then we should find that in physical systems that we ordinarily consider highly chaotic, the volume fraction $f_{A(t),B}$ grows with time (or equivalently, that $I_{\rm free}(A(t):B)$ decays with time). If the initial operator $A$ commutes with $B$, then it is intuitive that the initial value of $I_{\rm free}(A(t):B)$ is infinite.  This can be understood from the fact $\tilde A$ is forced to be block-diagonal in the eigenbasis of $B$, which makes  $\tilde A$ a measure-zero  set within the full set of all possible $UAU^{\dagger}$.  On the other hand, if we take the time-evolution operator to be a Haar-random unitary $U_H$, then in the limit of large Hilbert space dimension $A(t)=U_H A U_H^{\dagger}$ and $B$ approach ``freely independent'' or ``free" noncommuting variables. \rev{For freely independent variables, the asymptotic FMI is zero.} This is well-known to be true for the standard definition of the FMI from~\cite{voiculescu1999analogues,Voiculescu_2002,hiai2005large,HIAI_2009,Collins_2014}.%~\footnote{This is well-known to be true for the standard definition of the FMI from~\cite{voiculescu1999analogues,Voiculescu_2002,hiai2005large,HIAI_2009,Collins_2014}. \rev{Our finite-dimensional analysis likewise gives a value that vanishes as $d\to\infty$, although at finite $d$ the residual value is polynomially suppressed in $d$ rather than exactly zero.}}
The basic reason is that  when $A(t)$ and $B$ are free, the joint moments in~\eqref{all_moments_intro} are fully determined by the   ``marginal moments'' $\Tr[A^n]$ and $\Tr[B^n]$~\cite{nica2006lectures}. \rev{Thus, in the asymptotic free limit,}~\eqref{all_moments_intro} does not impose any  non-trivial constraints, and the volume fraction \rev{approaches} $1$.

Hence, when two operators are freely independent, they can be viewed as maximally scrambled relative to each other. Based on this observation, the relevance of freeness to quantum chaos and scrambling has been emphasized in a number of recent works~\cite{Fava_2025,pap,pap_new,jonah,Cipolloni_2022,Chandrasekaran_2023,penington2025,camargo2025quantum,jahnke2025free,fritzsch2025,fritzsch2025free}. So far,  a useful  figure of merit to quantify the approach to freeness under quantum dynamics has been lacking. The free mutual information   provides precisely such a tool.   

%\JW{Therefore, free mutual information measures how close the two operators are from being free. Physically, free independence means that operators are fully scrambled with respect to each other. Its relevance to quantum chaos and quantum scrambling has been emphasized in many recent works~\cite{Fava_2025,pap,pap_new,jonah,Cipolloni_2022,Chandrasekaran_2023,penington2025,camargo2025quantum,jahnke2025free,fritzsch2025}. However, there lacks a useful figure of merit to quantify the approach to free independence in a quantum dynamic process. Our work exactly provides such a tool.}

Haar-random unitaries lead to instantaneous scrambling, and thus lack the physical constraints of continuous time-evolution and locality or $k$-locality of interactions. To confirm that $I_{\rm free}(A(t):B)$ decays with time in more realistic chaotic systems from its initial infinite value to a small late-time value,\footnote{Note that at large but finite Hilbert-space dimension, the equilibrium value of our physical FMI will not be exactly zero: as discussed in Sec.~\ref{sec:haar}, it is suppressed by an inverse power of $d$ and vanishes only in the large-$d$ limit.} we need an explicit formula for this quantity which can be computed in various physical systems.  We derive such an explicit formula for a Hilbert space with some large finite dimension $d$,  in the case where $A$ and $B$ both have the spectrum of Pauli matrices ($d/2$ eigenvalues equal to $-1$ and $d/2$ eigenvalues equal to 1):  
\begin{align}
I_{\rm free}(A(t):B) =&- \frac{1}{d^2} \sum_{\substack{1 \leq i, j \leq d\\
{\rm Re}(z_i) \neq {\rm Re}(z_j) }} \log  |{\rm Re}(z_i) - {\rm Re}(z_j)| ~ -\log 2 \, ,  \label{ifree_formula_intro}
\end{align}
where $\{z_i\}_{i=1}^d$ are the eigenvalues of $A(t)B$. We will sometimes refer to the RHS as the ``Coulomb gas formula.'' See Theorem~\ref{theorem_coulomb} for a more formal version of this statement, which clarifies its regime of validity.

%There are certain edge cases where the Coulomb gas formula cannot be identified with $I_{\rm free}$ from the definition \eqref{ifree_informal}, and we will clarify these Sec.~\ref{sec:explicit_formula}. 

The formula \eqref{ifree_formula_intro}  allows us not only to compute the time-evolution of $I_{\rm free}(A(t):B)$ in a variety of physical examples, but also to derive a fully general relation between the free mutual information and higher-point out-of-time-ordered correlators, defined as 
\be 
{\rm OTOC}_n(A(t):B) \equiv  \Tr[\rho (A(t) B)^n],\,  \quad n=1, 2, ...
\ee
in some state $\rho$, which we will again take to be the maximally mixed state in the rest of the paper. $\otoc_1$ is the two-point function rather than an out-of-time-ordered correlator, but we denote it as above to use a common notation for all $n$. $\otoc_2$ is the standard four-point OTOC discussed at the beginning of the introduction, which can be related to growth of operator size in the physical space of degrees of freedom. 

$\otoc_n$ for $n>2$ have previously been discussed in the literature as natural generalizations of $\OTOC_2$, and have intuitively been understood to capture more fine-grained aspects of chaos and scrambling~\cite{ cotler_hunterjones,Leone_2021, jonah, pap,Fava_2025,jahnke2025free,fritzsch2025, nandy, haehl1,Haehl_2018,saad, nicole1, nicole2}. In particular, some of these references have noted that there is a natural relation between $\otoc_n$ and ideas from free probability, as in the limit where 
$A(t)$ and $B$ are freely independent, all $\otoc_n(A(t):B)$ are precisely zero for traceless $A$ and $B$. Based on this observation, the decay of $\otoc_n$ was  previously proposed as a  measure of approach to freeness in chaotic systems and studied in a few different models in the above references. %, and a natural set of questions was raised about whether the decay of $\otoc_n$  for higher $n$ can be associated  increasingly later time scales~\cite{pap}. 

Even with the existing studies of $\otoc_n$ in various models, a few important questions have not been conclusively addressed: 

\vspace{-0.2cm}

\begin{enumerate} 
\item[(i)] What is the precise physical interpretation of $\otoc_n$ for $n>2$?

\vspace{-0.2cm}

\item[(ii)] How should the data from the different $\otoc_n$ be combined together to understand the approach to freeness? For example, are smaller $n$ more relevant than higher $n$ in defining an approximate notion of freeness?

\vspace{-0.2cm}

\item[(iii)] Are there universal patterns in the $n$-dependence of $\otoc_n(A(t):B)$ in generic chaotic quantum many-body systems? 
\end{enumerate}

One main result of this paper is that for $A$ and $B$ which have the spectrum of Pauli matrices, in any system, we have the following general relation:
\be 
I_{\rm free}(A(t):B) = \sum_{n=1}^{\infty} \frac{2}{n} \otoc_n(A(t):B)^2 \label{otoc_intro}
\ee
See Theorems~\ref{theorem:otoc_1} and \ref{theorem:otoc_2} for more precise versions of this statement in cases with and without ensemble averaging.
This general relation has a number of conceptual as well as practical uses, which we summarize below: 
\begin{enumerate}
\item Since the quantity $I_{\rm free}$ has a clear physical interpretation in terms of ergodicity in operator space, the relation~\eqref{otoc_intro} also gives a physical interpretation to higher $\otoc_n$ in the cumulative form on the RHS of \eqref{otoc_intro}. This helps address questions (i) and (ii) above.~\footnote{Cumulative partial sums of OTOCs such as the RHS of \eqref{424}, with various choices of $n$-dependent weights, were recently introduced  in~\cite{jahnke2025free} as  prescriptions  to combine the higher-point OTOCs into a single quantity. Our result shows that the cumulative OTOC with the harmonic weight has a clear physical interpretation as the FMI.}   
\item If an operator in a chaotic system spreads not only in physical space but also in the abstract operator space, then we expect that $I_{\rm free}(A(t):B)$ should become finite and small for late enough times. This in turn leads to the expectation that at late enough times in chaotic systems, the magnitude of $\otoc_n$ must be upper-bounded by a decreasing function of $n$, at least for large enough $n$, so that the sum in \eqref{otoc_intro} converges. This suggests a potential universal pattern of $n$-dependence that addresses question (iii) above.  
\item In practice, $\otoc_n$ are often easier to compute than the spectrum of $A(t)B$ with existing techniques in various physical models, and are likely to also be  easier to measure in experiments. Indeed, these $\otoc_n$ up to $n=4$ have already been measured in a recent experiment in random unitary circuits~\cite{google_paper}. See also~\cite{yao2016,Swingle_2016,Garttner_2017,Vermersch_2019} for earlier experimental progress in measuring $\otoc_2$. \rev{Further experimental developments include Refs.}~\cite{Li_2017_OTOC,Li_2026_Scramblons}. \rev{In the regime where \eqref{otoc_intro} applies, every term on its right-hand side is non-negative, so the finite partial sums
\[
I_{\rm free}(A(t):B)
\gtrsim  \sum_{n=1}^{n_{\rm max}}\frac{2}{n}\otoc_n(A(t):B)^2
\]
form an increasing sequence of lower approximations to the FMI. Thus one can estimate the FMI from below by accumulating experimentally measured higher-point OTOCs, with the approximation improving as $n_{\rm max}$ is increased.} 
\end{enumerate}

To understand the typical time-scales associated with convergence and decay of $I_{\rm free}(A(t):B)$, and the manner in which the sum in \eqref{otoc_intro} converges, we then turn to the task of evaluating $I_{\rm free}$ and $\OTOC_n$ in a number of examples of quantum many-body systems. We consider three examples of chaotic systems: (a) random GUE Hamiltonians, which have no local structure, (b) random unitary circuits, which have locality but no energy conservation, and (c) chaotic spin chain systems with both locality and energy conservation. 
In each of these examples, we find that the initial decay of $\otoc_n$ is  increasingly fast for higher $n$, so that the OTOC sum formula for the FMI in \eqref{otoc_intro} starts to converge and decay at some characteristic time  $t^{\ast}$. In the non-local GUE model, $t^{\ast}=0$, and in the local cases $t^{\ast}$ is proportional to the initial distance between $A$ and $B$. $t^{\ast}$ is close to, but somewhat earlier than, the time scale associated with decay of the standard four-point $\otoc_2$.  $I_{\rm free}(t)$ decays monotonically with time for $t>t^{\ast}$ and reaches a vanishing saturation value in the thermodynamic limit in each of these models. This  confirms our expectation that in general chaotic quantum many-body systems, any operator increasingly spreads out to fill the operator space volume. 

We further observe in each of these models
that while the {\it initial} decay of $\otoc_n$ is always faster for higher $n$, 
 $\otoc_n(t)$ show oscillations as a function of both $t$ and $n$ at intermediate times. For larger $n$, the magnitude of the oscillations with $t$ becomes larger and the time of their onset becomes earlier.  Such oscillations might naively suggest a departure from freeness at intermediate times in chaotic systems. However, we find that these oscillations do not lead to a lack of convergence of $I_{\rm free}$ or to oscillations in its value. This is a consequence of the way in which the $\otoc_n$ are organized in the sum \eqref{otoc_intro}. Hence, by using $I_{\rm free}$ as a figure of merit, we obtain the physically reasonable conclusion that generic chaotic systems monotonically approach freeness as a function of time. The partial sums of \eqref{otoc_intro}, cut off at some finite $n_{\rm max}$, also show monotonic decay with time.

It is natural to ask whether there are systems where an initial operator spreads in physical space, but not in operator space, in line with the intuition that spreading in operator space is a more fine-grained notion of chaos. We identify three examples of systems where $\otoc_2$ decays, indicating delocalization in physical space, but $I_{\rm free}$ remains infinite for all times: Clifford unitaries, PFC unitaries~\cite{Metger_2024}, and the transverse field Ising model (TFIM), which is a free fermion integrable system.   The mechanism underlying lack of decay of $I_{\rm free}$ is somewhat different in each of these models, and can be understood from various atypical features of the spectrum of $A(t)B$ which are not discernible from the value of $\otoc_2$ alone. 

In the free-fermion integrable TFIM, we find that not only do the individual $\otoc_n(t)$ show oscillations as a function of time, but the partial sums of \eqref{otoc_intro} also show non-monotonic behavior with time. This is in contrast to the monotonic decay with time of the OTOC partial sums in the chaotic spin chain. This observation suggests a  practical way of using the data of the higher $\otoc_n(t)$ in some physical system, obtained either from calculation or from measurements in the lab, to decide whether or not the system is chaotic. Instead of looking for monotonic decay with time in the individual $\otoc_n$, one should check whether the partial sums of \eqref{otoc_intro} decay monotonically with time. Unlike the free fermion integrable case, the case of the interacting integrable Heisenberg XXZ model appears to resemble that of chaotic spin chains.

%The case of the interacting integrable Heisenberg XXZ model appears to resemble that of chaotic spin chains.

%is particularly notable:  despite being integrable, this model  qualitatively resembles chaotic spin chains in terms of most dynamical phenomena (see for instance~\cite{interacting_integrable}), so that finding measures to distinguish its dynamical properties from chaotic systems has been an open problem.

%Clifford unitaries and PFC unitaries, but somewhat different in free fermion systems and the Heisenberg model. %These models can all be seen as integrable models in different ways. 

 It is worth noting that based on proposed relations between $\otoc_n$ and the formation of approximate unitary $n$-designs~\cite{roberts_yoshida}, it appears natural  that $\otoc_n$ for larger  $n$  may start to decay at later times  when the chaotic time-evolution is closer to forming an $n$-design. However, this turns out not to be the case for the generic chaotic models which we study in this paper, where  $\otoc_n$ for higher $n$ decay {\it faster}. In fact, the formation of approximate $n$-designs has more recently been understood not to be related to the decay of $\otoc_n$~\cite{schuster2025random,pap_new}. In Sec.~\ref{sec:designs}, we discuss an explicit example of an approximate $n$-design (the recently introduced PFC ensemble~\cite{Metger_2024}) where for certain operators, many of the $\otoc_n$
 do not decay at all.

The plan of the paper is as follows. In Sec.~\ref{sec:definition}, we present a formal version of the definition of $I_{\rm free}(A(t):B)$ for the setting of physical time-evolution. In Sec.~\ref{sec:explicit_formula}, we derive the Coulomb gas formula~\eqref{ifree_formula_intro}. In Sec.~\ref{sec:otocs}, we derive the OTOC sum formula~\eqref{otoc_intro}. Sections~\ref{sec:definition}-\ref{sec:otocs} should be seen as building the mathematical framework needed to study the free mutual information in any physical system. In Sec.~\ref{sec:physical_systems}, we study the behavior of $I_{\rm free}$ and $\otoc_n$ in a variety of concrete examples of chaotic  systems: the random GUE model, random unitary circuits, and the mixed-field Ising spin chain. On a technical level, we develop new tools for analyzing the behavior of $\otoc_n$ in both the random GUE model and random unitary circuits in Sec.~\ref{sec:physical_systems} and the associated appendices, which may also be useful in other contexts. We discuss various non-chaotic systems (Clifford unitaries, PFC unitaries, the transverse field Ising model, and the Heisenberg model) in Sec.~\ref{sec:nonchaotic}.  We comment on various future directions in Sec.~\ref{sec:conclusions}. Many technical details are presented in the Appendices.

\section{Free mutual information}

\label{sec:definition}

In this section, we will present a more formal version of the definition of the free mutual information which was summarized in the introduction. We will start by explicitly defining the set $\sS_A$, and volumes of subsets of $\sS_A$.

Consider an initial Hermitian operator $A$ in a quantum many-body system of Hilbert space dimension $d$, and its time-evolution $A(t)= U(t)^{\dagger} A U(t)$ under the evolution operator $U(t)$ of the system. For example, for an energy-conserving system, $U(t) = e^{-iHt}$ for some Hamiltonian $H$. We are interested in characterizing the spreading of $A(t)$ with time in the set $\sS_A$ of all possible operators that can be obtained from $A$ by any unitary evolution:
\be 
\sS_A = \{ U A U^{\dagger}~|~U \in \mathbf{U}(d) \}  \label{sadef}
\ee
where $\mathbf{U}(d)$ is the group of $d$-dimensional unitary matrices. $\sS_A$ can be viewed as a manifold  which is diffeomorphic to $\mathbf{U}(d)/{\rm Stab}_{\mathbf{U}(d)}(A)$, where $\text{Stab}_{\mathbf{U}(d)}(A)$ is the stabilizer of $A$, or  the subgroup of $\mathbf{U}(d)$ which commutes with $A$. Explicitly, we have the following bijection between the two manifolds: %Under a bijective diffeomorphism, 
\begin{align}
    \phi_A:\mathbf{U}(d)/{\rm Stab}_{\mathbf{U}(d)}(A)&\to \sS_A \nn
  U&\mapsto UAU^\dagger\ \, . 
\end{align}
%we have $\sS_A\simeq\mathbf{U}(d)/{\rm Stab}_{\mathbf{U}(d)}(A)$.  

%We will refer to the dimension of this manifold as $\gamma$. 

Let us first consider the dimension $\gamma_A$ of the manifold $\sS_A$. The manifold $\mathbf{U}(d)$  has dimension $d^2$. %Consider the coset $\mathbf{U}(d)/{\rm Stab}_{\mathbf{U}(d)}(A)$ obtained from quotient out $\text{Stab}_{\mathbf{U}(d)}(A)$, which is the subgroup of $\mathbf{U}(d)$ which commutes with $A$. 
The subgroup $\text{Stab}_{\mathbf{U}(d)}(A)$ depends on the degeneracies in the spectrum of $A$. For example, if the spectrum is non-degenerate, then the stabilizer group consists of diagonal unitary matrices in the eigenbasis of $A$, i.e.  $\text{Stab}_{\mathbf{U}(d)}(A) = \mathbf{U}(1)^d$.  Hence, for a non-degenerate $A$,  $\gamma_A$ is $d^2 - d$. More generally, if $A$ has $k$ distinct eigenvalues with multiplicities $g_1, ..., g_k$, then $\text{Stab}_{\mathbf{U}(d)}(A) = \mathbf{U}(g_1) \times \mathbf{U}(g_2) \times \cdots \times \mathbf{U}(g_k)$, and  $\gamma_A = d^2 - \sum_{i=1}^k g_i^2$. Note that in all cases, as long as $A$ is not the identity operator, $\gamma_A$ is $O(d^2)$.

To formalize the notion of spreading of operators in $\sS_A$, it will be useful to introduce a volume measure on $\sS_A$. The measure we will use throughout this paper is the one induced by the Haar measure $\dd \nu_{\mathbf{U}(d)}(U)$  on $\mathbf{U}(d)$, which we henceforth abbreviate as $\dd U$. The Haar measure is the unique (up to normalization) measure on $\mathbf{U}(d)$  which has the property of left- and right- unitary invariance: $\dd U = \dd(UV) = \dd(VU)$ for any $V \in \mathbf{U}(d)$. Since we will mostly be interested in ratios of volumes below,  the normalization of $\dd U$ can be chosen arbitrarily. To use the Haar measure to define volumes in $\sS_A$, we will use  the quotient map from $\mathbf{U}(d)$ onto the coset $\mathbf{U}(d)/ \text{Stab}_{\mathbf{U}(d)}(A)$, denoted by $q$. $q^{-1}(U)$ denotes the preimage of $U$ under this map. Then the induced volume measure $\nu_{\sS_A}$ is the pushforward of the Haar measure through $q$ and then through $\phi_A$:
\begin{equation}
    \nu_{\sS_A}\equiv%\nu_{\mathbf{U}(d)/ \text{Stab}_{\mathbf{U}(d)}(A)}\circ\phi^{-1}\equiv
    \nu_{\mathbf{U}(d)}\circ q^{-1}\circ\phi_A^{-1}
\end{equation}
The integral of a function $f(X)$ on $\sS_A$ is thus given by 
\be \label{volumedef}
 \int_{\sS_A} f(X) \dd\nu_{\sS_A} (X) %= \int_{\phi_A^{-1}(\sS)} f(U) \dd\nu_{\mathbf{U}(d)/ \text{Stab}_{\mathbf{U}(d)}(A)} (U)
 =\int_{\mathbf U(d)}f(U A U^\dagger)\,\dd U\ . %= \frac{ \int_{\mathbf{U}(d)} \dd  U}{\int_{\text{Stab}_{\mathbf{U}(d)}(A)} \dd  U} \,   
\ee
The total volume $\text{Vol}(\sS_A)$ is given by the above integral for $f(X)=1$, and the volume of some subset $\sS$ of $\sS_A$ is given by taking $f(X)$ to be the indicator function $\mathbf{1}_{\sS}(X)$, which evaluates to 1 for $X\in \sS$ and zero otherwise. 

A more explicit way to visualize the size of $\sS_A$ is to discretize it, by constructing a finite  ``epsilon-net'' of unitaries $\sS_{A, \epsilon}$ such that any element of $\sS_A$ is at most $\epsilon$-far from some element of $\sS_{A, \epsilon}$ according to some chosen metric.  The  number of elements in $\sS_{A, \epsilon}$ is 
\be 
n_{\sS_{A, \epsilon}} = c_1\left( \frac{c_2}{\epsilon}\right)^{\gamma_A} \label{nsa}
\ee 
for some $O(1)$ constants $c_1$, $c_2$. Recall that 
$\gamma_A$ is $O(d^2)$ -- hence, $n_{\sS_A, \epsilon}$ is {\it doubly exponential} in the number of degrees of freedom $n$ in a quantum many-body system. %This counting emphasizes the distinction between the spreading of $A(t)$ operator in the physical space of $n$ degrees of freedom, and the present discussion of the spreading of $A(t)$ in the set $\sS_A$. 

We would now like to formalize questions such as: what fraction of the volume in $\sS_A$ does $A(t)$ occupy, or what fraction of the elements of the $\epsilon$-net $\sS_{A, \epsilon}$ is $A(t)$ close to? As stated, the  answer to the question is trivial: since $A(t)$ is a single fixed operator, it occupies zero volume in $\sS_A$, and is $\epsilon$-close to an $O(1)$ number of elements of $\sS_{A, \epsilon}$. To come up with a meaningful version of the question, we need to provide a prescription for replacing $A(t)$ with some larger set of operators $\sS_{A(t), {\rm coarse}}\subset \sS_A$, whose volume can grow with time. 

As discussed in the introduction, we consider a coarse-graining procedure in terms of correlation functions of $A(t)$ with some simple reference operator $B$.  We define the  coarse-grained set  of $A(t)$ conditioned on $\rho$ and $B$, $S_{A(t)|B, \rho, \delta, N}$,  as the set of all operators $\tilde A \in \sS_A$ such that 
\be 
|\Tr[\rho {\tilde A}^{m_1} B^{n_1} {\tilde A}^{m_2} B^{n_2}...  {\tilde A}^{m_r} B^{n_r}] - \Tr[\rho {A(t)}^{m_1} B^{n_1} A(t)^{m_2} B^{n_2}...  A(t)^{m_r} B^{n_r}]|< \delta  \label{all_moments}
\ee
for all possible sequences of positive integers $m_i$, $n_i$ for all  lengths $r\leq N$ for some maximum length $N$. 
 In the rest of this paper, we will set $\rho$ to be the maximally mixed state or infinite temperature density matrix, $\rho= \frac{\mathbf{1}}{d}$, and drop the $\rho$ label, referring to the above set simply as $\sS_{A(t)|B, \delta, N}$. We restrict to this choice of $\rho$ for technical convenience, and it will allow us to probe the rich operator dynamics manifested in the infinite-temperature correlators in physical systems with large but finite Hilbert space dimensions.\footnote{ In more general physical setups, especially in continuum systems such as quantum field theories, considering other choices of $\rho$ such as the finite temperature Gibbs state $e^{-\beta H}/Z_{\beta}$ is important and should be developed further in future work.} We can also define a discretization $\sS_{A(t)|B,  \delta, N,  \epsilon}$ of $\sS_{A(t)|B, \delta, N}$, which is a subset of $\sS_{A, \epsilon}$ constrained by the same condition as in \eqref{all_moments}.

With the above definition of $\sS_{A(t)|B, \delta, N}$, we can now characterize the spreading of $A(t)$ in the space of all operators by the ratio 
\be 
f_{A(t)|B, \delta, N} \equiv \frac{\text{Vol}(\sS_{A(t)|B, \delta, N})}{\text{Vol}(\sS_{A})}\, . 
\ee
Equivalently, we can consider the following quantity, which we will refer to as the {\it free mutual information} of $A(t)$ and $B$: 
\begin{align}
I_{{\rm free}}(A(t):B)_{\delta, N} &\equiv - \frac{4}{d^2} \log f_{A(t)|B, \delta, N} \label{ifree1} \\ 
& = \frac{4}{d^2} \le(\log \text{Vol}(\sS_A) - \log \text{Vol}(\sS_{A(t)|B, \delta, N})\right) \label{ifree2} \\
& = \frac{4}{d^2} \lim_{\epsilon \to 0} (\chi(A(t);\epsilon) - \chi(A(t)|B;\delta, N,  \epsilon)) \label{ifree3}
\end{align} 
%In most of the discussion below, $d$ will be clear from the context, and $\delta$ will have a negligible effect as long as it is small in a sense we will make precise \SV{[again, see if $r_{\rm max}$ needs to be included in this discussion]}, so we will refer to the above quantity as simply $I_{\rm free}(A(t):B)$. 
In the final version of the above expression, we have written $I_{\rm free}$ in terms of 
\begin{equation}\label{free_entropy}
     \chi(A(t);\epsilon)\equiv\log n_{S_{A(t),\epsilon}},\quad \chi(A(t)|B;\delta, N,\epsilon)\equiv\log n_{S_{A(t)|B, \delta, N,\epsilon}}
\end{equation}
$\chi(A(t);\epsilon)$ can be intuitively interpreted as the amount of information needed to specify $A(t)$ up to $\epsilon$ precision with no knowledge except that it is obtained from $A$ by unitary evolution, and can be seen as a physical version of the ``free entropy'' of~\cite{voiculescu1993analogues,voiculescu1994analogues,voiculescu1996analogues,voiculescu1997analogues,voiculescu1998analogues,voiculescu1999analogues,Voiculescu_2002}. Similarly, $\chi(A(t)|B;\delta, N,\epsilon)$ is a version of the conditional free entropy of $A(t)$ conditioned on $B$, 
and captures the remaining uncertainty in our knowledge of $A(t)$ if we know its correlation functions with $B$. Definition~\eqref{ifree3} therefore justifies referring to $I_{\rm free}$ as a mutual information.

%Hence, their difference quantifies how much \SV{[}knowing $B$ in the past could help us describe $A(t)$ in the future \{] the knowlege of the relation .  %\eqref{ifree2} or \eqref{ifree3} justifies the interpretation of $I_{\rm free}(A(t):B)$ as a mutual information. The expression in \eqref{ifree3} makes the ``counting'' interpretation that one usually associates with entropies clearer. 

Note that in our motivation for the above definition, the operator $B$ was used used as a reference to quantify the spreading of $A(t)$. However, $I_{\rm free}$ is symmetric with respect to its arguments $A(t)$ and $B$, as will also be clear from the explicit formulas we will derive below.

Note that each free entropy term in~\eqref{ifree3} is divergent in the $\epsilon\to 0$ limit, but their difference  is finite from \eqref{ifree2} unless $\text{Vol}(\sS_{A(t)|B, \delta, N})$ is zero. The normalization factor of $\frac{4}{d^2}$ is included so that the quantity has a $d$-independent large $d$ limit (recall the leading $O(d^2)$ behavior of the exponent in \eqref{nsa}). %The specific choice to subtract $2d$ in the denominator will be motivated later, and for now should be seen as an arbitrary convention.

In the discussion so far, $I_{\rm free}$ defined in \eqref{ifree1}-\eqref{ifree3} depends on the parameters $\delta$ and $N$. To capture intrinsic properties of the physical system, we would like to avoid the dependence on these parameters. In the standard mathematical definition  from~\cite{voiculescu1999analogues,Voiculescu_2002,hiai2005large,HIAI_2009,Collins_2014}, one simultaneously takes the limits $d \to \infty, \delta \to 0, N\to \infty$. In the present  setup, where we want to consider a large but fixed value of $d$ to define a physically meaningful quantity, the limits  $\delta \to 0, N\to \infty$ turn out to be too strong, and always lead to an infinite value of $I_{\rm free}$. We will explain in the next section how we can choose $N$ and $\delta$ in such a way that $I_{\rm free}$ is finite, but at the same time the leading contribution to it does not depend on $N$ and $\delta$. This will justify our omission of the $\delta, N$ label for $I_{\rm free}$  in most of the paper.

Before further discussing the behavior of the free mutual information defined above, let us make a few comments on the motivation for the coarse-graining procedure based on the relations \eqref{all_moments}, and a few different ways of viewing the physical interpretation of this quantity:
\begin{enumerate}
\item Recall that the operator $A(t)$ can be fully specified by its correlation functions in the maximally mixed state with a complete basis of operators for $\sH_d$. 
Hence, a natural way of coarse-graining $A(t)$ is to only specify correlation functions with a smaller subset of simple operators, and we consider above the case with just one reference operator $B$. 
This idea is similar to the coarse-graining procedure used in the Jaynes entropy or ``simple entropy'' prescription for coarse-graining a quantum state $\rho$%, which is defined by considering a larger set of states $\sigma$ in which the correlation functions of simple operators are approximately equal to those in $\rho$~
\cite{jaynes_1, jaynes_2, engelhardt_wall, vsafranek2019quantum,vsafranek2021brief}.

\item The quantity $I_{\rm free}(A(t):B)$ can be viewed as a new attempt to characterize ``mutual information across time,'' which has been approached from a variety of perspectives in ~\cite{Hosur_2016,Glorioso_2024,wu2025quantum,milekhin2025observable}.~\footnote{See also time-like entanglement entropies~\cite{Doi_2023,Nakata_2021}.} Like in these previous approaches, one of our goals is to find an information-theoretic  quantity that captures the fact that chaotic dynamics lead to a loss of correlations over time.  One intrinsic difficulty in defining a mutual information of this kind is that one cannot define a joint density matrix  for multiple times  and its reduced density matrices for each time~\cite{Fitzsimons_2015,Horsman_2017,Fullwood_2022}. A joint density matrix is defined for two commuting subalgebras of observables, but we cannot have commuting subalgebras for different times.  
For two non-commuting variables, the joint distribution is defined by the set of joint moments appearing in \eqref{all_moments} without resorting to a joint density matrix. The natural entropy associated with two such variables is the free entropy of the joint distribution~\cite{voiculescu1999analogues,Voiculescu_2002,hiai2005large,HIAI_2009,Collins_2014}, which is not the von Neumann entropy of any (pseudo-)density matrix.

\item An alternative way to view $I_{\rm free}$ is in terms of quantum programming, where one asks for the minimal quantum description length~\cite{Berthiaume_2001,Gacs_2001,Yang_2020} of the future observable $A(t)$, conditioned on some observations $B$ made in the past. Suppose that $I_{\rm free}$ is non-zero: this tells us that the volume of the set $\tilde A$ determined by the conditions \eqref{all_moments} is smaller than the full volume of the set of all possible $UAU^{\dagger}$. The relations between $A(t)$ and~$B$, expressed through their joint moments, thus allow us to give a more compressed description of $A(t)$ than we would have without any knowledge of the relation between $A(t)$ and $B$ (we only need to specify $A(t)$ within the smaller set of $\tilde A$, instead of the full set of $UAU^{\dagger}$). A more precise version of this programming interpretation is discussed in~\cite{jw_talk}.
\end{enumerate}

\section{Explicit formulas for the free mutual information}
\label{sec:explicit_formula}

In this section, we first derive an explicit formula for the free mutual information in terms of the eigenvalues of the operator $U(t)AU(t)^{\dagger}B$ for a fixed time-evolution operator $U(t)$ in Sec.~\ref{sec:formula}. We then consider the case where $U(t)$ is drawn from an ensemble of time-evolution operators in Sec.~\ref{sec:ensemble}. We discuss how these formulas can be used to confirm that $I_{\rm free}$ vanishes for large $d$ when the time-evolution operator is a Haar-random unitary in Sec.~\ref{sec:haar}. 

\subsection{A Coulomb gas formula}
\label{sec:formula}

In this section, we will derive a closed expression for the free mutual information defined in \eqref{ifree2} in the case where $A$ and $B$ are both Hermitian and unitary operators with the spectrum of Pauli matrices on $\sH_d$, i.e. they have  $d/2$ eigenvalues equal to +1 and $d/2$ eigenvalues $-1$.  We will assume throughout that $d$ is even. Since this property of the spectrum is equivalent to the combination of the properties $\Tr[A]=\Tr[B]=0$ and  $A^2 = B^2 = \mathbf{1}$, we will sometimes refer to $A$ and $B$ as ``traceless involutions'' below. 
 In local systems made up of $n$ qubits, a natural choice is to consider $A$ and $B$ to be strings of Pauli matrices, i.e.  operators of the form $\otimes_{i=1}^n {\sigma^{(\alpha_i)}}_i$ where $\alpha = 0, ...,3$ labels the 4 possible Pauli matrices $I, X, Y, Z$ on each site. More generally, the formula below applies when $A=U P U^{\dagger}$ and $B=VQV^{\dagger}$ for any Pauli strings $P$ and $Q$ and any unitaries $U$ and $V$. 
 
Due to the involution property of $A$, $A(t)$ and $B$, if we take $\rho= \frac{\mathbf{1}}{d}$, then the only non-trivial conditions in the set \eqref{all_moments} are 
 \be 
 \frac{1}{d}|\Tr[(\tilde A B)^n] - \Tr[(A(t) B)^n]| < \delta  \text{ for all integers } 0< n \leq  N \, .  \label{pauli_condition}
 \ee
The remaining conditions in~\eqref{all_moments} are automatically satisfied by the fact that since $\tilde A \in \sS_A$, we have  $\tilde A = U A U^{\dagger}$ for some $U$, so that $\tilde A^2 = \mathbf{1}$. This simplification is the key reason for choosing $A$ and $B$ to be traceless involutions.\footnote{\rev{The definition of the FMI in Sec.~\ref{sec:definition} itself applies to operators with general spectra. We do not, however, expect the same simple Coulomb-gas formula derived in this section to apply to arbitrary $A$ and $B$. When $A$ and $B$ are not involutions, the full collection of alternating joint moments in \eqref{all_moments} does not reduce to the moments of the single product $A(t)B$. Consequently, there is no single density function whose moments capture all of the coarse-graining constraints for generic $A$ and $B$; a general expression would have to retain genuinely multivariable noncommutative data.}}

Note that for any pair of operators $O$ and $B$ that are both Hermitian and unitary, the eigenvalues of the product $OB$ come in pairs of complex conjugate phases. Let us define the following vector in $\vec{y}\in[-1,1]^{d/2}$ associated with $O$ (we will always take $B$ to be the fixed reference operator): 
\be 
\vec{y}(OB)_j \equiv \cos\phi_j=\Re e^{\pm i\phi_j}, \quad \{e^{\pm i\phi_j}\}_{j=1}^{d/2} \text{ are the eigenvalues of } OB \label{ydef}
\ee
\rev{We always list these components in nondecreasing order, so $\vec y(OB)$ lies in the ordered domain
\be
\sW_{d/2}\equiv\left\{\vec y\in[-1,1]^{d/2}:y_1\leq\cdots\leq y_{d/2}\right\}. \label{nondec}
\ee}
The moments $\frac{1}{d}\Tr[(OB)^n]$ depend on $O$ only through $\vec{y}(OB)$. We will use a special notation $\vec{x}$ for the case $O = A(t)$: 
\be 
\vec{x} \equiv \vec{y}\le(A(t)B\ri) \label{xvecdef}
\ee

\iffalse 
Note that since $A(t)$ and $B$ are both Hermitian and unitary, the eigenvalues of the product $A(t)B$ come in pairs of complex conjugate phases, and can be labelled $\{e^{\pm i\theta_j}\}_{j=1}^{d/2}$. Similarly, the eigenvalues of $\tilde A B$ for any $\tilde A \in \sS_A$ can be labelled $\{e^{\pm i\phi_j}\}_{j=1}^{d/2}$. The moments appearing in \eqref{pauli_condition} only depend on the real parts of these eigenvalues, i.e. on the $\frac{d}{2}$-dimensional vectors $\vec{y}(\tilde AB)\in [-1, 1]^{d/2}$,  $\vec{x}\in [-1, 1]^{d/2}$ defined by 
\be 
\alpha_i(\tilde A) \equiv \cos \phi_i, \quad \beta_i \equiv \cos \theta_i \, \, . \label{alphabetadef}
\ee
We will choose to label both $\alpha_i$ and $\beta_i$ such that their values are nondecreasing as a function of $i$. It will be convenient to write $\vec{y}$ explicitly as a function of $\tilde A$, as in the discussion below  $\tilde A$ will be allowed to vary as long as it satisfies \eqref{pauli_condition}. On the other hand, $\vec{x}$ is always defined as above in terms of the eigenvalues of the {\it fixed} operator $A(t)B$. 
\fi

Our goal is to obtain an explicit formula for $I_{\rm free}(A(t):B)_{\delta, N}$ defined in \eqref{ifree1} that does not depend  strongly on the parameters $\delta$ and $N$, and is instead determined by intrinsic properties of $A(t)$ and $B$. A first natural choice for this purpose is to consider the limit of $I_{\rm free}(A(t):B)_{\delta, N}$ where take $N \to \infty$ and $\delta \to 0$. In  this limit,  \eqref{pauli_condition} 
is equivalent to the condition 
\be
\vec y(\tilde AB)=\vec x\, . 
\ee
\rev{where the nondecreasing ordering \eqref{nondec} is assumed.}
Using the definitions~\eqref{ifree1} and~\eqref{volumedef}, we have\footnote{The limits of $\delta \to 0$ and $N \to \infty$ commute because they can be replaced by $\inf_{\delta>0,N}$ as $I_{\rm free}(A(t):B)_{\delta, N}$ monotonically decreases as $N$ grows and $\delta$ shrinks.}
\begin{align}
\lim_{\delta \to 0, N \to \infty}I_{\rm free}(A(t):B)_{\delta, N} &= -\frac{4}{d^2} \log\le(\frac{\text{Vol}(\{\tilde A \in \sS_A  \text{ such that } \vec{y}(\tilde AB) = \vec{x} \})}{\text{Vol}\le(\sS_A\ri)}\ri) \\
&= -\frac{4}{d^2} \log\frac{\int_{\sS_A} \mathbf1_{\vec{y}(\tilde AB) = \vec{x}}\,\dd\nu_{\sS_A}(\tilde A)}{\text{Vol}\le(\sS_A\ri)} \\
&= -\frac{4}{d^2} \log\frac{\int_{\mathbf{U}(d)} \mathbf1_{\vec{y}(U\tilde A U^\dagger B) = \vec{x}}\,\dd U }{\int_{\mathbf{U}(d)} \dd (U) }\\
&\equiv -\frac{4}{d^2} \log P_U\le(\vec{y}(UAU^{\dagger}B) = \vec{x}\ri) \label{ifree_prob}
\end{align}
where $\mathbf1_{\vec{y}(\tilde AB) = \vec{x}}$ is the indicator function that equals $1$ if the condition is met, and $0$ otherwise. In the last line, we define $P_U(X)$ as the probability of event $X$ under the distribution of  Haar-random unitaries $U$. %The interpretation of the volume ratio as a probability in  \eqref{ifree_prob} comes from the fact that we are normalizing by the volume for $\sS_A$ induced by the Haar measure, as discussed around \eqref{volumedef}. 

Now \eqref{ifree_prob} reveals an immediate problem with taking the 
limits $N \to \infty$ and $\delta\to 0$. To see this, note the following lemma: 
\begin{restatable}{lemma}{lemmazero}\label{lem:lemma_0} 
The probability under Haar-random $U$ that the \rev{ordered vector} $\vec{y}(UAU^{\dagger}B)$ lies in a measurable subset \rev{$S\subseteq\sW_{d/2}$} is
\be 
P_U(\vec y(UAU^{\dagger}B) \in S) = \int_S \,  \prod_{i=1}^{d/2} \dd y_i \,\,   p\le(\vec{y}\ri) \label{prob_S}
\ee
where the probability density $p(\vec{y})$ is given by 
 \be
p(\vec{y})   \equiv \sN  \prod_{1\leq i<j\leq d/2}(y_i-y_j)^2, \quad \quad \sN =\rev{\left(\frac d2\right)!} 2^{-\frac{d^2}{4}}
\prod_{k=0}^{\frac{d}{2}-1}
\frac{(\tfrac{d}{2}+k)!}{(k+1)!\,(k!)^2}  \, .   \label{prob}
\ee
\end{restatable}
\begin{proof}
We provide the proof in Appendix~\ref{app:prob}.
\end{proof}
Putting this result into \eqref{ifree_prob}, where the allowed volume of $\vec{y}$ around the fixed value $\vec{x}$ is zero, we see that 
\be 
\lim_{\delta \to 0, N \to \infty} I_{\rm free}(A(t):B) = - \frac{4}{d^2}(\,\log p(\vec{x}) + \sum_{i=1}^{d/2} \lim_{\dd y_i \to 0} \log \dd y_i \, ) = \infty\ . \label{ifree_infinite}
\ee
The second term involving the log of the zero volume around $\vec{x}$  dominates over the first term, giving an infinite result. 
This tells us that for a finite Hilbert space dimension, the conditions \eqref{pauli_condition} for $N=\infty$ and $\delta=0$ are too constraining, and do not allow $\sS_{A(t)|B, \delta, N}$ to have a non-zero volume fraction. 

We therefore need to regulate this infinite answer by considering some finite $N$ and $\delta$, while still aiming to get a result that is universal in the sense that it does not strongly depend on the regulators. In order to get a universal  formula, we will have to make an assumption for the minimum spacing among the $x_i$, which are the real parts of the eigenvalues of $A(t)B$,
\be 
\Delta_{\rm min} \equiv \min_{i\neq j}|x_i-x_j|\, . \label{deltamin_def}  
\ee
At $t=0$, we will typically take $A(t=0)$ and $B$ to be commuting operators. Since $A$ and $B$ both have have $d/2$ eigenvalues equal to  $+1$ and  $d/2$ eigenvalues $-1$, their product also has at least $d/2$ degenerate eigenvalues if they commute, so that we have $\Delta_{\rm min} =0$.  However, we will assume that in typical situations of interest, at all but very early times, all degeneracies among the~$x_i$  get sufficiently lifted so that we have 
\be
\Delta_{\rm min} =\exp(-d^{o(1)}) \label{327}
\ee
which explicitly means that $\Delta_{\rm min} \geq \exp(-\kappa \,  d^\beta)$ for all $0< \beta < 1$ and all $\kappa>0$ at a sufficiently large $d$.  (Note in particular that this assumption allows $\Delta_{\rm min}$ to have some arbitrary polynomial decay with $d$.) We justify this assumption with explicit computation in one model in Appendix~\ref{app:gue_checks}. With this assumption, if we take 
\be 
\delta = \Theta(\exp({- d^\alpha})),\quad N = C/\delta,\  
\label{deltaN_choice} 
\ee
for some $0< \alpha <1$
and for some constant $C> 9\pi$. Here the big-$\Theta$ notation means that $\delta$ is both upper and lower bounded by $\exp({- d^\alpha})$ up to a prefactor for large $d$. For the derivation in this section, it is necessary to distinguish it from the big $O$ notation that only strictly means an asymptotic upper bound. 

Then we will be able to obtain a formula for $I_{\rm free}(A(t):B)$ which is both {\it finite} and     {\it independent of $N$ and $\delta$ at leading order in $1/d$}. For example, it does not depend on the specific choice of  $\alpha$ or $C$ in \eqref{deltaN_choice}.  Intuitively, this regulated result for $I_{\rm free}$ should end up being the first term of \eqref{ifree_infinite}. The purpose of the rest of this section is to show that  this is indeed the case. Readers who are mainly interested in the result can skip to Theorem~\ref{theorem_coulomb} and the comments below it.

 It will be useful to first provide upper and lower bounds on the volume of  $\sS_{A(t)|B, \delta, N}$, the set of $\tilde A$ satisfying the condition \eqref{pauli_condition} for some given $\delta, N$.  For this purpose, we will first find upper and lower bounds of the volume  in $\sW_{d/2}$ of the vectors $\vec{y}(\tilde AB)$ for $\tilde A \in \sS_{A(t)|B, \delta, N}$, and then use the formula~\eqref{prob}.  The following two lemmas, both proved in Appendix~\ref{app:volume_lemmas}, will help provide such bounds. \rev{Here and below, an $\epsilon$-box around $\vec x$ means the intersection of the usual coordinate box ($\vec{x}'\in [-1, 1]^\frac{d}{2}$ such that $|x'_i -x_i|\leq r$ for each $i = 1, ..., d/2$) with the ordered domain $\sW_{d/2}$.}

%\JW{Perhaps we should state our main result more prominently (Coulomb gas formula) as a Theorem here and put the derivation inside a proof, so readers can immediately find it. (Now they can only see a couple of Lemmas.) What do you think?} \SV{Sounds good.}

\begin{restatable}{lemma}{L}\label{lem:lemma_1}
%\begin{lemma}
%\label{lemma_1}
Define $\sB_1\subset \rev{\sW_{d/2}}$ as an $\epsilon$-box around $\vec{x}$ of radius 
\be 
\epsilon_1 = c d\,(- \delta\log\delta) \, \label{e1def}   
\ee
for a constant $c$ (for concreteness, we can take $c=\frac{26}{\pi}$). 
 If $N$ and $\delta$ satisfy the condition  
\be 
N > \frac{9 \pi}{\delta} \, ,   \label{310}
\ee
then 
for all $\tilde A \in  \sS_{A(t)|B, \delta, N}$, the corresponding vectors $\vec{y}(\tilde AB)$ are contained within $\sB_1$. 
\end{restatable}

\begin{restatable}{lemma}{secondlemma}\label{lem:lemma_2}
Define $\sB_2\subset\rev{\sW_{d/2}}$ as an $\epsilon$-box around $\vec{x}$ of radius 
\be
\epsilon_2=\delta/N^2\, . 
\ee
$\sB_2$ is contained within the set of $\vec{y}(\tilde AB)$ corresponding to $\tilde A$ in $\sS_{A(t)|B, \delta, N}$. 
\end{restatable}

From Lemmas~\ref{lem:lemma_1} and~\ref{lem:lemma_2} (cf. Fig.~\ref{fig:boxes} for an illustration), we have the following upper and lower bounds on the volume fraction  $f_{A(t)|B, \delta, N}$:  
\be
f_2 \leq f_{A(t)|B, \delta, N} \leq f_1 \label{317} 
\ee
where 
\begin{align}
& f_{1,2} \equiv \frac{\text{Vol}(\tilde{A} \in \sS_A\ \text{s.t.} \ \vec{y}(\tilde AB) \subseteq\sB_{1,2})}{\text{Vol}(\sS_A)}  = P_U \le( \vec y(U A U^{\dagger}B)\subseteq \sB_{1,2} \ri) \, . 
\end{align} 
and the last equality follows from~\eqref{ifree_prob}. 
\begin{figure}
    \centering
    \includegraphics[width=0.45\linewidth]{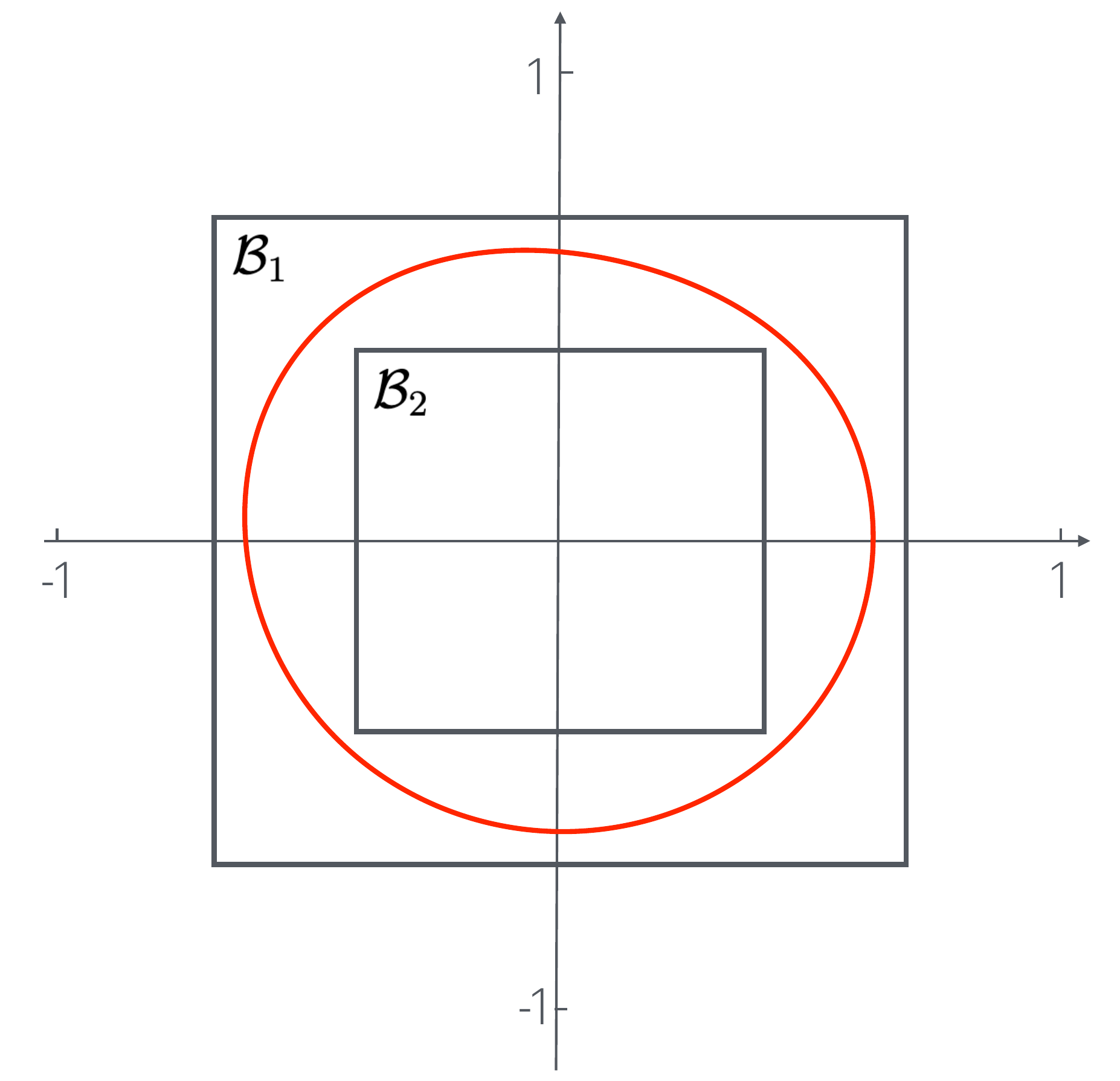}
    \caption{We use two boxes $\sB_1$ and $\sB_2$ to sandwich the set of $\vec{y}(\tilde AB)$ corresponding to $\tilde A$ in $\sS_{A(t)|B, \delta, N}$, depicted as the region contained within the red circle, so as to estimate its volume. This figure is an illustration for $d=4$.}
    \label{fig:boxes}
\end{figure}

Our strategy to derive a formula for $I_{\rm free}$ will be to show that $-\frac{4}{d^2}\log f_1$ and $-\frac{4}{d^2}\log f_2$ are close in the large $d$ limit, so that  from \eqref{317} either one gives a good approximation to $I_{\rm free}$. Let us use the probability density~\eqref{prob} to write $f_a$ ($a=1,2$) more explicitly: 
\begin{align}
    f_a &= \sN \rev{\int_{\sB_a}}  \prod_{1\leq i<j\leq d/2}(y_i-y_j)^2 \prod_{i=1}^{d/2} \dd y_i  \,.  \label{fi_def} 
\end{align} 
Recall that the eigenvalues $x_i$ are sorted in nondecreasing order. So if $i<j$, for $\{y_i\}$ appearing in the integrand of \eqref{fi_def} for $i=1$, 
\be 
|y_i -y_j| \leq (x_j - x_i + 2\epsilon_1)
\ee
and therefore, 
\be 
f_1 \leq \sN (2\epsilon_1)^{d/2} \prod_{1\leq i<j\leq d/2}(x_j-x_i+2\epsilon_1)^2   \ ,  
\ee
which implies that 
\begin{align}
I_{\rm free}(A(t):B)_{\delta, N} \geq & - \frac{4}{d^2}\log \sN -  \frac{8}{d^2} \sum_{1 \leq i< j \leq d/2} \log(x_j -x_i) \label{ifree_line1} \\
&-  \frac{8}{d^2} \sum_{1 \leq i < j \leq d/2} \log\le( 1 + \frac{2\epsilon_1}{x_j-x_i} \ri) -\frac{2}{d} \log (2 \epsilon_1) \ .  \label{ifree_lb}
\end{align}

Recall from \eqref{e1def} that $\epsilon_1=-c d \delta \log \delta$. So, to get a $\delta$-independent expression at leading order in $1/d$, \eqref{ifree_lb} should be subleading compared to \eqref{ifree_line1}. Both terms in \eqref{ifree_line1} have leading $\Theta(1)$ contributions. This gives us the following requirement for the allowed range of the parameter $\delta$: 
\begin{enumerate}
\item[(i)]  $\epsilon_1=-cd \delta \log \delta$ should be such that \eqref{ifree_lb}  goes to 0 as $d \to \infty$, so that \eqref{ifree_lb} is negligible compared to \eqref{ifree_line1} for large $d$. 
\end{enumerate}

Next, in order to give an upper bound for $I_{\rm free}$, we need a lower bound on $f_2$. 
To provide such a lower bound, we need the condition: 
\begin{align}
\epsilon_2 = \frac{\delta}{N^2} < \frac{\Delta_{\rm min}}{2} \, , \label{3242}
\end{align}
under which we have 
\begin{align}
    f_2  
    &\ge \sN  \prod_{1\leq i<j\leq d/2}(x_j-x_i-2\epsilon_2)^2  (2\epsilon_2)^{d/2}   \label{326}
    %&\ge \sN \, 2^{-d} \prod_{1\leq i<j\leq d/2}(\beta_i-\beta_j)^2  \epsilon_2^{d/2} \, . 
\end{align}
which leads to 
\begin{align}
I_{\rm free}(A(t):B)_{\delta, N} \leq & - \frac{4}{d^2}\log \sN -  \frac{8}{d^2} \sum_{1 \leq i< j \leq d/2} \log(x_j -x_i) \label{ifree2_line1} \\
&-  \frac{8}{d^2} \sum_{1 \leq i < j \leq d/2} \log\le( 1 - \frac{2\epsilon_2}{x_j-x_i} \ri) -\frac{2}{d} \log (2\epsilon_2)   \label{ifree2_lb}
\end{align}
$\epsilon_2$ depends on both $\delta$ and $N$, so again in order to obtain a regulator-independent expression (which coincides with the lower-bound from \eqref{ifree_line1}) we want to choose the parameters $\delta$ and $N$ within a range such that 
\begin{enumerate}
\item[(ii)]  $\epsilon_2=\delta/N^2$ should be such that \eqref{ifree2_lb}  goes to 0 as $d \to \infty$, so that \eqref{ifree2_lb} is negligible compared to \eqref{ifree2_line1} for large $d$. 
\end{enumerate}

Overall, to get a result that is both finite and regulator-independent in the large $d$ limit, we want to choose $N$ and $\delta$ in a range such that we can simultaneously satisfy (i) and (ii) above, as well as the inequalities \eqref{310} and \eqref{3242} which were needed for intermediate steps of the above discussion. In a case where $A(t)B$ has degeneracies, $\Delta_{\rm min}=0$. In this case, we cannot simultaneously satisfy all four requirements:  \eqref{3242}, (i) and (ii) can never be satisfied for non-zero $\delta$, whereas \eqref{310} requires non-zero $\delta$ for finite $N$. Such degeneracies in $A(t)B$ exist in particular at $t=0$. \rev{They can also persist at later times when they are protected by symmetries; this is not merely a logical possibility, but occurs in the Heisenberg-model example of Sec.~\ref{sec:heisenberg}.  For dynamics and operator choices without symmetry-enforced degeneracies, we assume that after a short initial time, $\Delta_{\rm min}$ is lower-bounded as in \eqref{327}}.~\footnote{In symmetry-constrained models, the free mutual information should instead be defined in terms of the volume ratio on the submanifolds compatible with the relevant symmetry class, as discussed in Sec.~\ref{sec:conclusions}.} With this nondegeneracy assumption, taking $\delta$ and $N$ as in \eqref{deltaN_choice} allows us to satisfy all four requirements (i), (ii), \eqref{310}, and \eqref{3242}. 

We have $\epsilon_1=\Theta(cd^{1+\alpha}e^{-d^\alpha})$ and $\epsilon_2=\Theta(e^{-3d^\alpha})$.  It is easy to check that \eqref{310} and \eqref{3242} are satisfied by these choices. To see that (i) is satisfied, note that the magnitude of the first line of \eqref{ifree_lb} can be upper-bounded as follows: 
\begin{align}
\frac{8}{d^2} \sum_{1 \leq i < j \leq d/2} \log\le( 1 + \frac{2\epsilon_1}{x_j-x_i} \ri) %\leq &\frac{8}{d^2} \sum_{1 \leq i < j \leq d/2} \le(\frac{2\epsilon_1}{\beta_j-\beta_i} \ri) \nn 
%\leq & \frac{8}{d^2} \sum_{1 \leq i < j \leq d/2} \le(\frac{2\epsilon_1}{\Delta_{\rm min}} \ri) 
%\nn 
  \leq  & 4 \frac{\epsilon_1}{\Delta_{\rm min}} \leq  4 \,c\, d^{1+\alpha}e^{-d^\alpha/2} 
\end{align}
where we have used $\log(1+x)\leq x$ and the definition of $\Delta_{\rm min}$ in the first inequality, and the fact that $\Delta_{\rm min} \geq e^{-d^\alpha/2}$ from the assumption \eqref{327} in the second. For the magnitude of the second term of \eqref{ifree_lb}, note that  
\be 
-\frac{2}{d} \log (2\epsilon_1) =- \frac{2}{d} (\log 2c  + (1+\alpha)\log d -d^\alpha)  = O(d^{\alpha-1})\ .
\ee
Similarly, to see that (ii) is satisfied, note that magnitude of the first term of \eqref{ifree2_lb} has magnitude can be upper-bounded by $4 \epsilon_2/\Delta_{\rm min}$, which is $O(e^{-3d^{\alpha}}$), using the fact that $-\log(1-x)\leq \frac{x}{1-x}$. The second term of \eqref{ifree2_lb} is  again $O(d^{\alpha-1})$.

Now, combining the inequalities \eqref{ifree_line1} and \eqref{ifree2_line1}, and using the fact that (shown in Appendix~\ref{app:Stirling})
\be 
-\frac{4}{d^2} \log \sN = -\log 2 + O(\log d/d)\ , \label{eq:vol_constant} 
\ee
we have 
\be 
I_{\rm free}(A(t):B) = -\frac{4}{d^2}\sum_{1\leq i\neq j\leq d/2}\log|x_i-x_j|-\log 2  +O(d^{\alpha-1})\ . \label{ifree_dz}
\ee
The normalization $\frac{4}{d^2}$ we previously chose in the definition~\eqref{ifree1} is approximately equal to the inverse of the number of terms in the sum, $\frac{d^2-2d}{4}$, at large $d$, so that the resulting quantity is $\Theta(1)$ as long as the differences $|x_i-x_j|$ are $\Theta(1)$ for most of the pairs $(i, j)$. 

We summarize the result as the following theorem for the free mutual information $I_{\rm free}(A(t):B)_{\delta, N}$
defined in \eqref{ifree2}:
\begin{theorem}
\label{theorem_coulomb}
Let  $A$ and $B$ be two traceless involutions, and  $\{x_i\}$ be the real parts of the eigenvalues of the unitary $A(t)B$ lying on or above the real axis. Let us choose the parameters $\delta=\Theta(\exp(-d^\alpha))$ for some $0<\alpha<1$, and $N=C/\delta$ for $C>9\pi$.  When $\Delta_{\rm min} \equiv {\rm min}_{i\neq j}|x_i-x_j| =\exp(-d^{o(1)})$, we have 
\be 
I_{\rm free}(A(t):B)_{\delta, N} = -\frac{4}{d^2} \sum_{1\leq i\neq j \leq d/2} \log|x_i-x_j| - \log 2 + O(d^{-z})\, , \label{fmiformula}
\ee  
which is independent of the details of $\delta$ and $N$ at the leading order in $1/d$. Here $z$ is some positive real number. 
\end{theorem} 

It is clear from the above statement that there are certain restrictions on when $I_{\rm free}$ can be identified with the formula on the RHS. We will therefore sometimes refer to the RHS as the ``Coulomb gas formula'' or $\sC(A(t):B)$ in parts of the later text where we want to distinguish it from $I_{\rm free}$.  
Let us summarize two  important points about the regime of validity of this formula: 
\begin{enumerate}
\item In cases where $\Delta_{\rm min}$ does not satisfy condition \eqref{327}, some of the above arguments relating $I_{\rm free}$ to the Coulomb gas formula do not go through. In all such cases, we will formally set $I_{\rm free}(A(t):B)=\infty$. In particular, this means that in the case where $A(t)$ and $B$ commute, we set $I_{\rm free}=\infty$. 
This choice is consistent with the intuition mentioned in the introduction that when $A(t)$ and $B$ commute, $\tilde A$ must be block-diagonal in the eigenbasis of $B$, so that $\sS_{A(t)|B}$ is a lower-dimensional sub-manifold of $\sS_A$ and the volume fraction $f_{A(t)|B}$ is zero.~\footnote{It would require more work to formalize this intuition for finite $\delta$ and $N$, where $\tilde A$ is not forced to be exactly block-diagonal in the eigenbasis of $B$ by the conditions~\eqref{pauli_condition}. Instead of trying to formalize this edge case, we will simply choose to set $I_{\rm free}=\infty$ as a convention as discussed above.}

%Instead of trying to mathematically justify  this for $f_{A(t)|B,\delta, N}$, we will simply set the quantity to be infinite in this case.    
\item If the value of the Coulomb gas formula is proportional to some power of $1/d$, then the terms we have neglected in \eqref{ifree_dz} can give competing contributions to $I_{\rm free}$. We will formally set $I_{\rm free}$ equal to zero in such cases. 
\end{enumerate}

\iffalse 
The formula appears to diverge in the case where the spectrum of $A(t)B$ is degenerate. However, recall from the above analysis that  this formula only applies for $\Delta_{\rm min}$ satisfying \eqref{327}, and in particular does not apply in cases with exact degeneracies. Hence, within the regime of validity where we can relate it to $I_{\rm free}$ up to negligible corrections, the RHS of \eqref{fmiformula} is finite.  However, it is useful to understand when this formula can grow with $d$,  indicating a divergence in the thermodynamic limit. From \eqref{fmiformula}, we see that this can happen if $O(d^2)$ pairs of $i, j$ satisfy $|\beta_i-\beta_j| = O(1/d^m)$ for some power $m>1$. In this scenario, $I_{\rm free} \sim \log d$. Recall $A(t=0)B$ has $|\beta_i-\beta_j|=0$ for either $d(d-1)$ or $d/2(d/2-1)$ pairs $i, j$.  Hence, assuming that the spectrum of $A(t)B$ evolves in a continuous way with $t$, there is generically some range of times that is late enough such that   $\Delta_{\rm min}$ satisfies \eqref{327} (so that \eqref{fmiformula} is valid), and early enough such that $|\beta_i-\beta_j|=1/d^m$ for $O(d^2)$ pairs of $i, j$. \SV{See Appendix~\ref{app:prob_2} for a discussion of the relevant time scales in one physical model.}  
 Note also that \eqref{fmiformula} is no longer valid in a regime where it evaluates to a value that is suppressed as some power of $d$, as in this regime the $O(1/d^z)$ contributions that we have ignored from the above discussion compete with \eqref{fmiformula}. 
\fi

\subsection{An integral formula for ensembles of unitary evolutions}
\label{sec:ensemble}

In some of the examples we will consider later in the paper, the time-evolution operator $U(t)$ will be drawn from an ensemble of time-evolutions instead of a single fixed evolution. For example, we will consider random GUE Hamiltonians in Sec.~\ref{sec:gue}, and random unitary circuits in Sec.~\ref{sec:random_circuit}. Such ensembles are often used as toy models for chaotic systems. While it may be possible to numerically probe single realizations from the random ensemble in such cases, for analytic calculations it is often convenient to average the quantity of interest over the ensemble. In this section, we will present a useful formula for the ensemble average of the free mutual information in such cases, and clarify the assumptions and approximations under which it is valid.

Let the real parts of the eigenvalues of $U(t)A U(t)^{\dagger}B$ be $(x_1,\ldots,x_{d/2})$. We define the corresponding empirical spectral density as
\be 
\rho(x) \equiv \frac{2}{d} \sum_{i=1}^{d/2} \delta(x- x_i )\ .
\ee
Say $U(t)$ is distributed according to some ensemble $\nu$. Consider the ensemble averages
\be 
\overline{\rho}(x) \equiv \mathbb{E}_{U(t)\sim \nu}\, \rho(x)\ ,\quad \overline{\rho(x)\rho(y)} \equiv \mathbb{E}_{U(t)\sim \nu}\, \rho(x)\rho(y)\ ,
\ee 
In general, $\overline{\rho(x)\rho(y)}\neq\overline{\rho}(x)\overline{\rho}(y)$. However, we will assume that for the systems and time scales we are interested in, for large $d$ we can assume factorization 
$\overline{\rho(x)\rho(y)}\approx\overline{\rho}(x)\overline{\rho}(y)$ while evaluating $I_{\rm free}$.  We expect, based on evidence from the random GUE model which will later be discussed in Sec.~\ref{sec:gue}, that this factorization  generally gives a good approximation in a regime where $I_{\rm free}$ is 
not proportional to a power of $1/d$.~\footnote{Haar-random unitaries provide an interesting example where $I_{\rm free}$ is $O(1/d)$ and the factorization does not give a good approximation, as we will discuss in the next subsection.} Recall that in cases where $I_{\rm free}$ is proportional to a power of $1/d$, the Coulomb gas formula for $I_{\rm free}$  does not apply, and we simply choose to set the quantity to zero. 

Note that while $\rho(x)$ has $\delta$-function peaks, as long as the measure of the  ensemble $\nu$ is ``sufficiently smooth'', the averaged density $\bar \rho(x)$ does not have such peaks. The absence of $\delta$ function peaks can be formally stated as   %the property that the spectral density $\bar \rho$ is \SV{integrable}, i.e.,
\begin{equation}\label{eq:no_atoms_mt}
    \int\dd x\,\overline{\rho}(x)\mathbf{1}_{x=x^{\ast}} =0, \quad \text{for all }\,  x^{\ast} \in [-1, 1]\ .
\end{equation}
We will explain the smoothness condition under which \eqref{eq:no_atoms_mt} holds in Appendix~\ref{sec:smoothness}, and show that it is satisfied for the non-local GUE Hamiltonian evolution for any $t>0$, and  for local random unitary circuits after a time scale determined by the initial separation of $A$ and $B$.

Under the above assumptions and conditions, the average of the free mutual information \eqref{fmiformula} can be approximated as follows:

\begin{align}
    \mathbb E_{U(t)\sim\nu}I_{\rm free}(A(t):B)=&\mathbb -\, \mathbb E_{U(t)\sim\nu}\frac{4}{d^2}\sum_{1\leq i\neq j\leq d/2}\log|x_i-y_j| -\log 2 \label{stepzero}\\
    =&-\int_{-1}^1\int_{-1}^1\log|x-y|\bar\rho(x)\bar\rho(y)\,\dd x\,\dd y \nn 
&+\int_{-1}^1\int_{-1}^1\log|x-y|\bar\rho(x)\bar\rho(y)\mathbf 1_{x=y}\,\dd x\,\dd y 
     -\log 2\label{step1}\\
    =& \, \,  \sI(A(t):B)  \label{step2}
\end{align}
where we have defined 
\be 
\sI(A(t):B) \equiv -\int_{-1}^1\int_{-1}^1\log|x-y|\bar\rho(x)\bar\rho(y)\,\dd x\,\dd y -\log 2 \, . \label{eq:integralFMI}
\ee
In \eqref{step1}, we
subtract the diagonal contribution (enforced by the indicator function $\mathbf{1}_{x=y}(x,y)=1$ if $x=y$ and $0$ otherwise) from the double integral to to match the exact definition of the off-diagonal sum in \eqref{stepzero}. In going from \eqref{step1} to \eqref{step2}, we notice that this diagonal term vanishes because the indicator function $\mathbf{1}_{x=y}(x,y)$ is supported on the line $x=y$ that has Lebesgue measure zero.\footnote{One might worry that $\log|x-y|$ diverges on the line $x=y$, so the function $g(x,y)=\log|x-y|1_{x=y}(x,y)\bar \rho(x)\bar \rho(y)$ is ill-defined on the line. However, note that $g$ is zero almost everywhere and the line is a set of measure zero. We can define its value on the diagonal arbitrarily (e.g., set it to $0$ there) without affecting the integral. Here we are also using the property \eqref{eq:no_atoms_mt} of $\bar\rho$ that it contains no delta functions.} Note that the condition \eqref{eq:no_atoms_mt} implies that $\sI(A(t):B)$ is always finite. We will sometimes refer to $\sI$ as the ``ensemble FMI'' below.

\subsection{Haar-random unitaries} \label{sec:haar}

For $A(t) \equiv U_H B U_H^{\dagger}$ where $U_H$ is a Haar-random unitary, $A(t)$ and $B$ are ``freely independent'' non-commuting variables in the large $d$ limit. The standard mathematical definition of the free mutual information from~\cite{voiculescu1999analogues,Voiculescu_2002,hiai2005large,HIAI_2009,Collins_2014} vanishes in this case. This observation provides part of our motivation for studying the free mutual information in the context of more general chaotic time-evolutions. Since our definition for the free mutual information is somewhat different from that in the mathematical references, we should explicitly  check how our  formulas behave in this case. By using an analytic formula for  $\overline{\rho(x)\rho(y)}$ for this case (which can be obtained by integrating out all but two of the $y_i$ in  \eqref{prob_S}), we find that the average of the Coulomb gas formula 
\be 
\overline{\sC}=-\log 2 -\int_{-1}^{1}\log|x-y|\overline{\rho(x)\rho(y)}\dd x\dd y
\ee
has a negative value with magnitude $\sim d^{-0.8}$.  
Hence, based on the discussion at the end of Sec.~\ref{sec:formula}, the precise value of $\overline{\sC}$ cannot be  identified with $\overline{I_{\rm free}}$ in this case. However, $|\overline{\sC}|=O(1/d^2)$ tells us that $\overline{I_{\rm free}}$ is $O(1/d^z)$ for some $z>0$. This is sufficient for confirming that 
$\overline{I_{\rm free}}$ vanishes for large $d$, as expected from the free independence of the two random variables in the large $d$ limit~\cite{voiculescu1999analogues,Voiculescu_2002,hiai2005large,HIAI_2009,Collins_2014}. 

We can also evaluate the ensemble FMI $\sI$ defined in \eqref{eq:integralFMI} in this case. We find that it has a positive $O(1/d^2)$ value, and thus does  not agree with that of $\overline{\sC}$. Hence, in this case, $\overline{\rho(x)\rho(y)}$ cannot be approximated as $\bar\rho(x) \bar\rho(y)$ while evaluating $\overline{\sC}$. Based on our study of other examples later in the paper, we believe this lack of factorization is a peculiar feature of cases where the values of $\overline{\sC}$ and $\sI$ are $O(1/d)$. Since the precise values of $\overline{\sC}$ and $\sI$ in this case do not have a clear physical meaning, we will not go into more detail on them in this paper, but will discuss them in a companion paper on mathematical aspects of the FMI~\cite{vw2025}.

\begin{comment}
Under the assumption that $\overline{\rho(x)\rho(y)}\approx\overline{\rho}(x)\overline{\rho}(y)$ at large $d$, which shall be verified in examples, we have
\begin{align}
    \mathbb E_{U\sim\nu}I_{\rm free}(A(t):B)&\approx \mathcal{I}(A(t):B)\,\\
   \mathcal{I}(A(t):B)&\equiv -\int_{-1}^1\int_{-1}^1\log|x-y|\overline{\rho}(x)\overline{\rho}(y)\,\dd x\,\dd y -\log 2\ .\label{eq:integralFMI}
\end{align}

The formula~\eqref{eq:integralFMI} is actually the exact formula for free mutual information $\mathcal I(A(t):B)$ of two Paulis according to the mathematical definition of free mutual information used in the free probability literature. We will provide a proper derivation of this formula in a companion paper.
\end{comment} 

\section{Free mutual information and OTOCs} 
\label{sec:otocs}

In this section, we will relate the Coulomb gas  formula \eqref{fmiformula} for the free mutual information from the previous subsection (up to an equivalent reorganization of the sum), 
\begin{align}
\sC(A(t):B) \equiv&- \frac{1}{d^2} \sum_{\substack{1 \leq i, j \leq d\\
{\rm Re}(z_i) \neq {\rm Re}(z_j) }} \log  |{\rm Re}(z_i) - {\rm Re}(z_j)| ~ -\log 2, \nn
&\text{where } \{z_i\}_{i=1}^d\in\mathbb T \text{ are the eigenvalues of } A(t)B \, ,  \label{ifree_formula}
\end{align}
to a sum over the infinite temperature higher-point OTOCs,  
\be 
{\rm OTOC}_n(t) \equiv \frac{1}{d} \Tr[(A(t) B)^n]\, . 
\ee

We summarize our main results in the two theorems below. 

\begin{theorem}\label{theorem:otoc_1}
   Let $A$ and $B$ be two traceless involutions and $A(t)=U(t) A U(t)^\dagger$ for some unitary time-evolution operator $U(t)$. Then Coulomb gas formula $\sC(A(t):B)$, defined as the RHS of \eqref{fmiformula}, can be equated with a series over $OTOC_n$: 
\be 
\sC(A(t):B) = \sum_{n=1}^{\infty} \frac2n \le(\otoc_n(t)^2-\frac{1}{d}\ri)  - \frac{2\log 2}{d} - \frac{1}{d}\sum_{n=1}^{\infty} \frac2n \otoc_{2n} \ .  \label{coulomb_series}
\ee
\end{theorem}
This will in turn allow us to relate $I_{\rm free}$ to a sum over higher-point OTOCs in the regime where its value is $O(1)$, i.e., not as small as some power of $1/d$. 

\begin{remark}\label{rem:otoc}
Suppose we are interested in evaluating $I_{\rm free}(A(t):B)$ up to some $O(1)$ error $\epsilon$.  We can then drop the $O(1/d)$  tail in \eqref{coulomb_series},  truncate the infinite series at some $n_{\max}\sim O(\vep^{-1})$, and get the approximation~\footnote{We assume that we are considering systems at times such that ${\rm OTOC}_n(t)$ is $O(1)$ for $n$ of $O(1)$. If this were not the case, then the value of $I_{\rm free}$ itself would be $O(1/d)$, and the Coulomb gas formula~\eqref{fmiformula} would not be valid. }
\begin{equation}
I_{\rm free}(A(t):B) = \sum_{n=1}^{n_{\rm max}} \frac2n {\rm OTOC}_n(t)^2 + O(\vep)\ . \label{424}
\end{equation}
\end{remark}

Similarly, for dynamics described by a unitary ensemble, the ensemble FMI~\eqref{eq:integralFMI} also admits a sum-of-OTOCs formula.
\begin{theorem}\label{theorem:otoc_2}
    Let $A$ and $B$ be two traceless involutions and $A(t)=U(t) A U(t)^\dagger$, where the unitary dynamics $U(t)$ is drawn from a unitary ensemble $(\mathcal U,\nu)$. Let $\overline{\rho}$ be the average spectral density of $A(t)B$, and $\overline{\otoc_n(t)}\equiv d^{-1}\mathbb E_{U(t)\sim\nu}\mathrm{Tr}[(A(t)B)^n]$. The ensemble FMI~\eqref{eq:integralFMI} admits the following OTOC series
\begin{equation}\label{425}
    \mathcal I(A(t):B)= \sum_{n=1}^\infty\frac2n \overline{\otoc_n(t)}^2\ .
\end{equation}
\end{theorem}
This formula~\eqref{425} can be a useful tool for calculating $\sI$ in cases where we only have access to the 
$\overline{\otoc}_n$ instead of $\overline \rho(x)$. 

The proofs of the two theorems are provided in the next two subsections.

\subsection{Derivation of the OTOC sum formula}

Let us first write $\sC$ as follows\footnote{\rev{The decomposition into $\sX+\sY$ used below is a formal bookkeeping device. The two terms are separately divergent, and only their combination after cancellation is assigned a meaning. We therefore do not regard the intermediate manipulations of $\sX$ and $\sY$ as a stand-alone rigorous derivation. The resulting finite OTOC identity is independently and quantitatively corroborated by the numerical comparison in Appendix~\ref{app:gue_checks_otoc}.}}:
\begin{align} 
&\sC(A(t):B) = \mathcal{X} + \mathcal{Y} \label{xy} - \log 2\ , \\
&\mathcal{X} = -\frac{1}{d^2}\sum_{1 \leq i , j \leq d} \log  |{\rm Re}(z_i) - {\rm Re}(z_j)| \label{45}\ ,\\ 
&\mathcal{Y} =
\frac{2}{d^2}\sum_{1 \leq i \leq d} \log  |{\rm Re}(z_i) - {\rm Re}(z_i)|\ . \label{46}
\end{align}

Let us first focus on the term $\sX$. Instead of working with the density of states $\rho:[-1,1]\to \mathbb{R}$ introduced in the previous section, in the discussion below it will be convenient to work with an analog  of the density of states for complex eigenvalues. Recall that the eigenvalues $\{z_i\}_{i=1}^d$ of $A(t)B$ lie on the unit circle $\mathbb{T}$ in the complex plane. We define a distribution $\mu$ on $\mathbb{T}$, which has the property that for any function $f:\mathbb{T} \to \mathbb{C}$,  
\be 
\oint_{\mathbb{T}} \dd z\, \mu(z) f(z)  = \frac{1}{d}\sum_{i=1}^d f(z_i)\, .  \label{48}
\ee
Using the Sokhotski–Plemelj formula in complex analysis,  $\mu(z)$ with the property \eqref{48} is given by 
\begin{align}
\mu(z) = \frac{1}{2\pi i} \le( \lim_{\zeta \to z^+} R(\zeta) - \lim_{\zeta \to z^-} R(\zeta)  \ri), \label{mudef}  \\
\quad \sR(\zeta) \equiv \frac{1}{d}\sum_{j=1}^{d} \frac{1}{\zeta - z_i} = \frac{1}{d}\Tr\le[\frac{1}{\zeta- A(t)B}\ri] \,. \label{49}
\end{align}
where $\zeta \to z^{+}$ ($\zeta \to z^{-}$) refers to approching $z$ from outside (inside) the unit circle $\mathbb{T}$. Here we have introduced the resolvent function $\sR(\zeta)$ for $A(t)B$, which has the same form as the more standard resolvent for a Hermitian operator. Such resolvents are also used in the study of OTOCs in~\cite{google_paper}.  

We can then express $\sX$ in terms of $\mu(z)$ as follows: 
\begin{align}
\sX =-\oint_{\mathbb{T}} \dd z_1 \oint_{\mathbb{T}} \dd z_2 \mu(z_1) \mu(z_2) \log |{\rm Re}(z_1) - {\rm Re}(z_2)| \, . \label{411_0} 
\end{align}

We will now relate $\mu(z)$ to the ${\rm OTOC}_n$  by noting that  $\sR(\zeta)$ can be expanded in terms of OTOCs in two different ways depending on whether $\zeta$ lies inside or outside $\mathbb{T}$:\footnote{We thank Douglas Stanford for suggesting this method.}
\begin{align}
&\sR(\zeta) =  \frac{1}{\zeta} \sum_{n=0}^{\infty} \frac{1}{\zeta^n} \frac{1}{d}\Tr[(A(t)B)^n] =\frac{1}{\zeta} \sum_{n=0}^{\infty} \frac{1}{\zeta^n} \otoc_n\, , &  |\zeta|>1 \label{411} \\
&\sR(\zeta) =
-\sum_{n=0}^{\infty} \zeta^n \frac{1}{d}\Tr[(A(t)B)^{-n-1}]=-\sum_{n=0}^{\infty} \zeta^n \otoc_{n+1} \ , & |\zeta|<1 \label{412}
\end{align}
where we have used $d^{-1}\Tr[(A(t)B)^{-n-1}] =d^{-1}\Tr[ (A(t)B)^{n+1}]=\otoc_{n+1}$ in the second equality. Note that we have dropped the $A(t):B$ argument of $\otoc_n$ to simplify notation. Putting \eqref{411} and \eqref{412} into \eqref{49}, we get 
\be 
\mu(z) = \frac{1}{2\pi i}\le(\sum_{n=0}^{\infty} {\rm OTOC}_n(z^{-n-1} + z^{n-1}) + \frac{1}{z}\ri)\ ,
\ee
and further 
\be 
\mu(z) \dd z|_{z=e^{i\theta}} = \dd\theta \le(\frac{1}{\pi} \sum_{n=0}^{\infty} {\rm OTOC}_n \cos (n\theta) + \frac{1}{2\pi}\ri)\ .
\ee
Putting this expression for $\mu(z)\dd z$ into \eqref{411_0}, we find 
\be
\begin{aligned}
    \sX=&-\frac{1}{\pi^2}\sum_{n=0}^\infty\sum_{m=0}^\infty\int_0^{2\pi}\int_0^{2\pi}\dd\theta\,\dd\phi\log|\cos\theta-\cos\phi|\otoc_n\otoc_m\cos(n\theta)\cos(m\phi)\\\
    &+\frac{2}{\pi^2}\sum_{n=0}^\infty\int_0^{2\pi}\int_0^{2\pi}\dd\theta\,\dd\phi \log|\cos\theta-\cos\phi|\otoc_n\cos(n\theta)- \log 2\ .
\end{aligned}
\ee
The above expression can be further simplified using the following identity:
\begin{equation}
    \log |\cos\theta - \cos\phi| = \log 2+ \log|\sin\frac{\theta+\phi}{2}|+ \log|\sin\frac{\theta-\phi}{2}|\ .
\end{equation}
Using the standard Fourier series for $\log|\sin \omega|=-\log 2-\sum_{n=1}^\infty\cos(2n\omega)/n$, we obtain a Fourier series expansion for the logarithmic potential,
\begin{equation}
    \log |\cos\theta - \cos\phi| = -\log 2 - \sum_{p=1}^\infty\frac2p\cos(p\theta)\cos(p\phi)\ ,
\end{equation}
which is valid for the whole domain $\theta,\phi\in[0,2\pi]$. Using the Fourier series expansion and the orthogonality relation
\begin{equation}
\int_0^{2\pi}\dd\theta\cos(p\theta)\cos(n\theta)\,=\pi\delta_{pn}\ ,
\end{equation}
we obtain
\be 
\sX = \log 2 + \sum_{n=1}^\infty \frac2n\otoc_n^2\ . \label{420}
\ee
Returning to \eqref{xy}-\eqref{46}, recall that $\sX$ contains a divergent contribution from cases where ${\rm Re}(z_i) = {\rm Re}(z_j)$, and $\sY$ subtracts these divergences. How can we see the divergent contribution in \eqref{420}? It is useful to note that in terms of the full set of eigenvalues $\{e^{i\theta}\}_{i=1}^d$ of $A(t)B$,  
\begin{equation}
\otoc^2_n=\frac1{d^2}\sum_{i,j=1}^d e^{in(\theta_i+\theta_j)}=\frac1{d^2}\sum_{i,j=1}^d e^{in(\theta_i-\theta_j)}= \frac1d+\frac{1}{d^2}\sum_{1\leq i\neq j\leq d} e^{in(\theta_i-\theta_j)}\ . \label{421}
\end{equation}

Since the $1/d$ term is independent of $n$, it gives a divergent contribution to \eqref{420}, while the remaining term should generically give a convergent contribution as its sign varies with $n$. This is true  unless  there are degeneracies among the $\theta_i$, which would lead to additional divergent contributions, consistent with \eqref{fmiformula}.

We therefore expect the $\sY$ term to cancel the above divergence. By similar steps to \eqref{411} to \eqref{420}, on expanding $\sY$ in terms of the OTOCs, we find 
\be 
\sC(A(t):B) = \sum_{n=1}^{\infty} \frac2n \le(\otoc_n^2-\frac{1}{d}\ri)  - \frac{2\log 2}{d} - \frac{1}{d}\sum_{n=1}^{\infty} \frac2n \otoc_{2n} \ . \label{full_otoc}
\ee
The details leading to \eqref{full_otoc} are discussed in Appendix~\ref{app:yterm}. Note that \eqref{full_otoc} holds as an exact identity for the Coulomb gas formula $\sC$, even in the regime where it cannot be identified as~$I_{\rm free}$. 

This concludes the proof of Theorem~\ref{theorem:otoc_1}. Now we discuss Remark~\ref{rem:otoc}.

In \eqref{full_otoc}, the divergent contribution to $\sX$ expected from \eqref{421} thus precisely gets cancelled on including the contribution from $\sY$. The $\sY$ term also gives rise to two additional contributions. We expect that that the series in the last term in \eqref{full_otoc} is likely to be convergent quite generally as it involves both positive and negative terms with an overall suppression in magnitude of $O(1/n)$.

Suppose we are in a regime where $\sC(A(t):B)$
is not proportional to some power of $1/d$, so that it can be identified with $I_{\rm free}$ according to Theorem~\ref{theorem_coulomb}. The $O(1/d)$ contributions to $\sC$ are not relevant for $I_{\rm free}$, so we can ignore the second and third terms of \eqref{full_otoc}. In principle, we should keep the $1/d$ term in the parenthesis first term to control a potential divergence in the infinite sum: 
\be 
I_{\rm free}(A(t):B) = \sum_{n=1}^{\infty} \frac2n\le({\rm OTOC}_n(t)^2 - \frac{1}{d}\ri) \label{423}
\ee
up to finite corrections of $O(1/d)$. In cases where $I_{\rm free}(A(t):B)$ has a finite value, the sum must converge and we should get a good approximation to $I_{\rm free}$ by keeping some $O(1)$ number of terms $n_{\rm max}$,
\begin{align}
I_{\rm free}(A(t):B) &\approx  \sum_{n=1}^{n_{\rm max}} \frac2n\le({\rm OTOC}_n(t)^2 - \frac{1}{d}\ri) \nn 
&\approx \sum_{n=1}^{n_{\rm max}} \frac2n {\rm OTOC}_n(t)^2 \, .  \label{nmax} 
\end{align}
In the last line we have assumed that we are considering systems at times such that ${\rm OTOC}_n(t)$ is $O(1)$ for $n$ of $O(1)$, such that we are justified in truncating to $n_{\rm max} \ll d$. If this were not the case, then the value of $I_{\rm free}$ itself would be $O(1/d)$, and the Coulomb gas formula would not be valid. 

We compare the partial sums of both \eqref{full_otoc} and   \eqref{nmax} with the Coulomb gas formula in one physical model in Appendix~\ref{app:gue_checks_otoc}, and confirm that the difference between \eqref{full_otoc} and   \eqref{nmax} is only important at late times when $\sC$ is an inverse power of $d$.

%Cumulative partial sum of OTOCs such as the RHS of \eqref{424}, with various choices of $n$-dependent weights, were recently studied in~\cite{jahnke2025free} to analyze the relationship between operator statistics and higher-point OTOCs. Here, our result shows that the cumulative OTOCs is particularly meaningful with the harmonic weight, and it is equal to an information-theoretic quantity that measures \SV{how much $B$ knows about $A(t)$ after evolution time $t$}.

\subsection{OTOC sum formula for ensemble averages}
\label{sec:otoc_ensemble}

Let us now consider the setup from Sec.~\ref{sec:ensemble}, where the time evolution operator $U(t)$ is drawn from some ensemble $\nu$. Recall that if we can approximate $\overline{\rho(x)\rho(y)} \approx \bar \rho(x) \bar \rho(y)$  and the measure $\nu$  obeys certain smoothness conditions, the ensemble average of $I_{\rm free}$ can be approximated as $\sI(A(t):B)$ in~\eqref{eq:integralFMI}. Equivalently in terms of $z_1, z_2 \in \mathbb{T}$, 
\be 
\sI(A(t):B) = - \oint dz_1 \oint dz_2 \bar \mu(z_1) \bar \mu (z_2) \log|\Re (z_1)- \Re(z_2)|
\ee
where $\bar\mu$ is the ensemble average of $\mu$ defined in \eqref{mudef}. The analysis of $\sI(A(t):B)$ is exactly parallel to that of $\sX$ in the previous section, with the replacement 
\be
\mu(z) \to \bar \mu(z), \quad  \frac1d\Tr[(A(t)B)^n] \to \oint_{\mathbb{T}}\dd z\, \bar \mu(z) z^n \,,
\ee
so that we find 
\be 
\sI(A(t):B) = \sum_{n=1}^{\infty} \frac{2}{n} \, \overline{\otoc_n(t)}^2, 
\ee
where $\overline{\otoc_n}$ is the ensemble-averaged OTOC: 
\be 
\overline{\otoc_n(t)} \equiv \oint_{\mathbb T} dz \bar \mu(z) z^n  = \frac{1}{d}\overline{\Tr[(A(t) B)^n]}\, .
\ee
There is no analog of the $\sY$ term of the previous section in this case. This concludes the proof of Theorem~\ref{theorem:otoc_2}.

\subsection{Physical implications}\label{sec:implications}

Let us now note two physical implications of  \eqref{424} and \eqref{425}: 
\begin{enumerate}
\item These results immediately imply the following inequality for any (up to the approximation in \eqref{424}, which requires neglecting $O(1/d)$ contributions):
\begin{equation}
    I_{\rm free}(A(t):B)\ge \frac2n\otoc_n(t)^2\ . \label{ifree_bound}
\end{equation}
 We can also put any partial sum over some subset of OTOCs as a stronger lower bound. The equality~\eqref{424} is clearly a much stronger statement, but the bounds can be useful if we do not have access to higher-point OTOCs beyond some $n$. Moreover, the bound has a similar structure to  
the following well-known bound for the quantum mutual information $I(X:Y)_\rho$ between two regions $X$ and $Y$~\cite{Wolf_2008} in some state $\rho$,
\begin{equation}\label{QMI>2pt}
    I(X:Y)_\rho \ge \frac{1}{2}\langle O_X O_Y\rangle_{\rho, c}^2\ .
\end{equation}
where $O_X$ and $O_Y$ are any operators in $X$ and $Y$ with unit operator norm, and $\langle O_X O_Y\rangle_{\rho, c}$ is their connected two-point function in $\rho$. The $n=1$ case of \eqref{ifree_bound} is most directly analogous to \eqref{QMI>2pt} as it involves the two-point function $\frac{1}{d}\Tr[AB(t)]$. 
This analogy suggests that $I_{\rm free}$ is a natural candidate for an information-theoretic measure of correlations over time. For various timelike mutual information measures proposed in the literature from other perspectives, see~\cite{Hosur_2016,Glorioso_2024,wu2025quantum,milekhin2025observable,Doi_2023,Nakata_2021}. 

\item Recall from Sec.~\ref{sec:ensemble} that under a condition of sufficient smoothness of the ensemble that is satisfied by many examples of interest (including the GUE and random unitary circuit ensembles discussed later in this paper), $\bar \rho(x)$ contains no delta functions after some characteristic time scale, and as a result $\sI$ is finite. In such cases, the infinite sum over $\overline{\otoc_n}^2$ in \eqref{425} must converge. 

Note that since the series \eqref{425} already involves a harmonic weight, it is right on the border of divergence. This implies that even a mild decay of the $\otoc_n$ with $n$, such as $\sim 1/n^{\delta}$ for small $\delta>0$, may be sufficient for a finite FMI. This leaves room for a variety of different decay behaviours of $\otoc_n$ with $n$. We shall examine the form of $\otoc_n$ decay and its associated timescales in a variety of physical models in the next two sections.
\end{enumerate}

%but none is fully satisfactory: they either suffer from a limited domain of applicability or a more complicated relation with temporal correlators. We believe this approach suffers from the intrinsic difficulty in defining a density matrix over time~\cite{Fitzsimons_2015,Horsman_2017,Fullwood_2022}, without which quantum entropies cannot be defined. Our approach evades this quandary, as we make use of entropy measure $I_{\rm free}$ that is not based on density matrices but rather directly on the observables and their distributions.

%\subsection{Remarks on OTOC series convergence and spectral degeneracies}\label{sec:smoothness}

\section{Free mutual information in chaotic systems} 
\label{sec:physical_systems}

Now that we have derived general formulas to study the free mutual information in  physical systems, let us turn to some concrete examples. In this section, we will consider three increasingly realistic models of chaotic time-evolution. In each of these models, our main goal is to evaluate $I_{\rm free}(A(t):B)$  and understand its time-dependence. Another important goal is to understand the role played by different $\otoc_n$ in the partial sums of \eqref{424} or \eqref{425}. In the course of trying to  understand these partial sums, we will also be able to extract certain universal features of the $n$-dependence of $\otoc_n$ across these different models. The techniques used and the results we are able to derive in each model are somewhat different, so we summarize the main results below. In each case, we take $A$ and $B$ to be initially commuting, so that $I_{\rm free}(A(t):B)=\infty$ at $t=0$.

\begin{enumerate}
\item {\it Random Hamiltonians from the GUE ensemble.} In Sec.~\ref{sec:gue}, we use the non-local GUE model as a first  example. In this model, we  numerically compute both $I_{\rm free}$ and $\otoc_n$ for  single realizations of the ensemble, and derive analytic formulas for the ensemble average of $\overline{\otoc_n}$. This allows for detailed comparison between ensemble averages and individual realizations, and between the the two sides of \eqref{424} for various $n_{\rm max}$. We find numerically that $I_{\rm free}(t)$ becomes finite for any $t>0$ and decays monotonically with time. All $\otoc_n$
start to decay for $t>0$, and higher $n$ show faster initial decay with $t$.  Due to this behavior of $\otoc_n$, at sufficiently late times we get a good approximation to $I_{\rm free}(t)$ even with $n_{\rm max}=2$ in \eqref{424}. 

At intermediate times, $\otoc_n(t)$ show oscillations as functions of both $t$ and $n$, which do not lead to oscillations in $I_{\rm free}(t)$ or the partial sums of \eqref{424}, and do not prevent the sums from converging.
\item {\it Random unitary circuits.} In Sec.~\ref{sec:random_circuit}, we consider brickwork Haar-random unitary circuits, which were   introduced as a model for chaotic dynamics in~\cite{operator_spreading_adam, operator_spreading_tibor}. This model is closer to realistic chaotic systems than the GUE model due to its local structure, although it lacks energy conservation.  Due to the local structure, the decay of $\otoc_n$ in this model for two operators initially separated by distance $x$ begins at a time $t^{\ast}$ proportional to $x$.  %$t>t^{\ast}= x/v^{\ast}$ for a characteristic velocity $v^{\ast}$.
We analytically study the ensemble average of $\otoc_n$ in this model, and conjecture a result for $t> t^{\ast}$ that is similar in structure to $\overline{\otoc_n}(t)$ in the random GUE model. In particular, all $\overline{\otoc_n}(t)$ start to decay simultaneously at $t^{\ast}$, and the result is consistent with a faster  initial decay for higher $n$.  

Based on this GUE-like structure, it is likely that the OTOC partial sums of \eqref{424} and $\overline{I_{\rm free}}$ for $t>t^{\ast}$ show a similar behavior to the GUE case, although we do not have direct access to these quantities from our analytic calculation. %We are not able to explicitly confirm this as we have less analytical and numerical control over this model.   

\item {\it Chaotic spin chain.} To take into account the effects of energy conservation as well as locality, we next consider the Ising model with both transverse and longitudinal fields. This model is widely studied in the literature as a representative example of chaotic spin chains. We numerically evaluate $I_{\rm free}$ and $\otoc_n$ in this model. We find that the qualitative features of both $I_{\rm free}(t)$ and $\otoc_n(t)$ at times  $t> t^{\ast}=x/v^{\ast}$ (for some characteristic $v^{\ast}$ for this model) are again similar to those of the random GUE model. 
\end{enumerate}

\subsection{Random GUE Hamiltonian}
\label{sec:gue}

As our first model for a highly chaotic time-evolution operator, let us consider $U_{\rm GUE}(t) = e^{-iHt}$, where $H$ is a random $d \times d$ Hermitian matrix from the probability distribution $p_{\rm GUE}(H)$  of GUE Hermitian matrices. More explicitly,  
\be
p_{\rm GUE}(H) = \frac{1}{\sN} e^{-\frac{d}{2}\Tr[H^2]} \, . \label{pgue}
\ee
The prefactor $d/2$ correspons to the choice that  the off-diagonal elements of $H$ are independent complex random variables with zero mean  and variance $\overline{|H_{ij}|^{2}}= \frac{1}{d}$, and the diagonal components are independent real complex variables with zero mean and variance $ \overline{H_{ii}^{2}}= \frac{1}{d}$. With this scaling of the variance, the average or typical  density of eigenvalues of $H$ lies between $-2$ and $2$ independent of $d$. Note that  the ensemble is invariant under conjugation by unitaries,  
$p_{\rm GUE}(H) = p_{\rm GUE}(UHU^{\dagger})$. 
This unitary invariance implies in particular that the eigenvectors of $H$ are Haar-random with respect to any fixed reference basis. As a result,  a typical realization of the ensemble does not obey any notion of locality with respect to some product basis. Nevertheless, unlike the Haar-random time-evolution operator discussed in Sec.~\ref{sec:haar},  $\sU_{\rm GUE}(t)$ does lead to a non-trivial time-evolution of $A(t)$, so that we can ask about the time-dependence of $I_{\rm free}$ and $\otoc_n$. 

We will first discuss  numerical observations from single realizations of this ensemble in Sec.~\ref{sec:numerical_section}, and then discuss the ensemble average in Sec.~\ref{sec:average_section}.

\begin{figure}[!h]
    \centering
    \begin{subfigure}{0.49\textwidth}
    \centering
    \includegraphics[width=\linewidth]{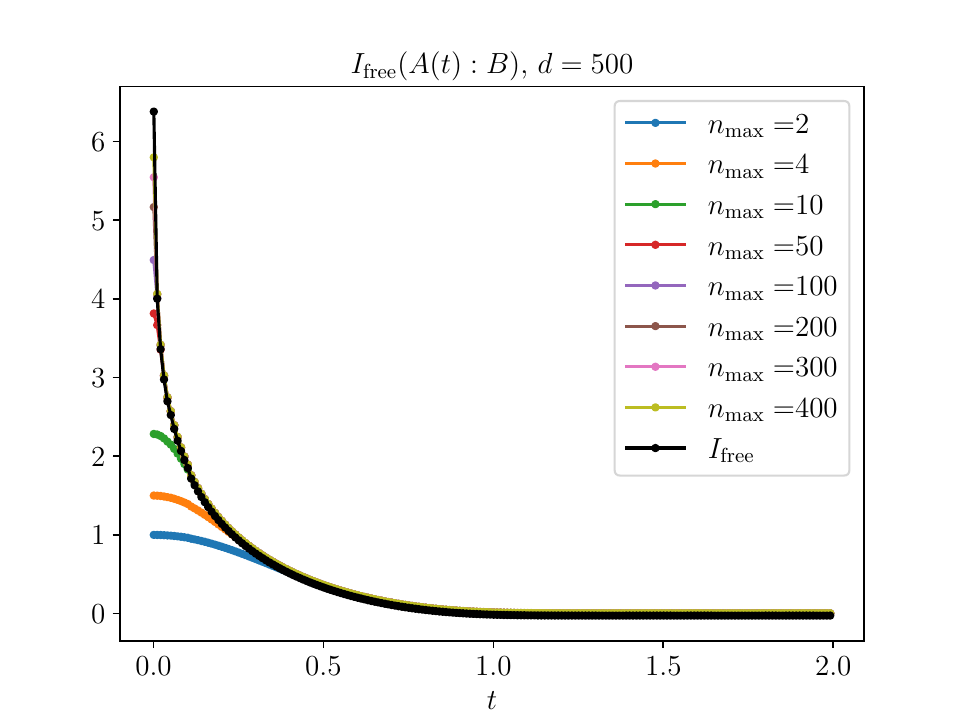}
    \end{subfigure}
    \begin{subfigure}{0.49\textwidth}      
    \centering
    \includegraphics[width=\linewidth]{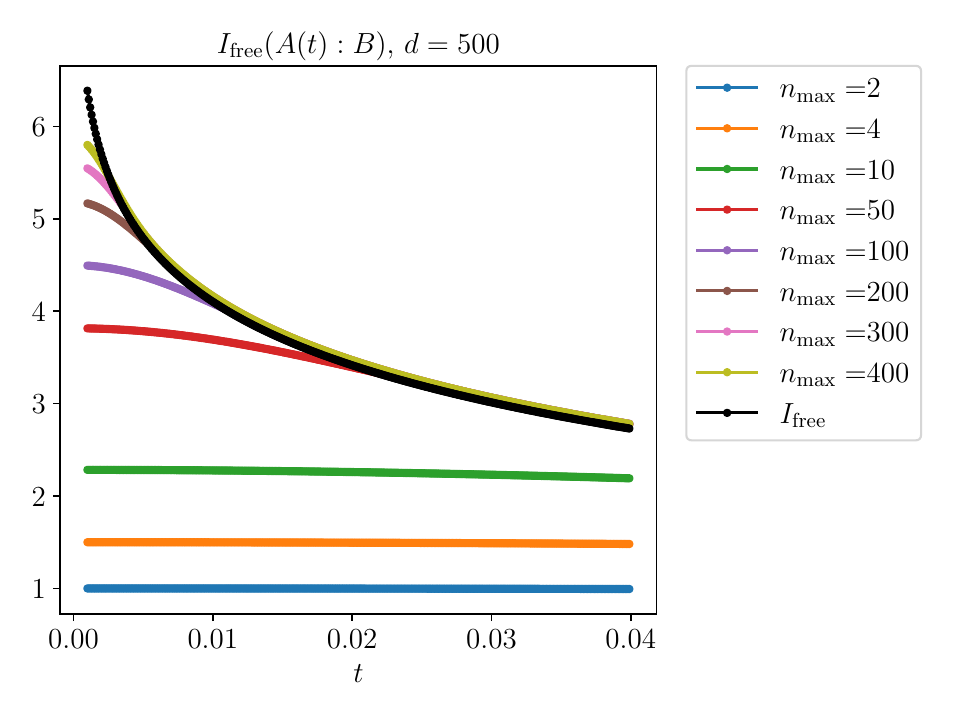}
    \end{subfigure}
    \caption{On the left, we show the behaviour of $I_{\rm free}$ for a range of times from $t=0.001$ to $t=2$ in a single instance of a random GUE Hamiltonian  for $d=500$. We compare the Coulomb gas formula in the black curve to the OTOC partial sums up to some $n_{\rm max}$. On the right, we show a zoomed-in version of the same plot at early times. The behaviour after times around $t=1$ is subtle, and a zoomed-in version is shown in Fig.~\ref{fig:d_dependence} and \ref{fig:otoc_negative} in Appendix~\ref{app:gue_checks}. While we explicitly show the case $A \neq B$, the plots are similar for the case $A=B$.}
    \label{fig:ifree_gue}
\end{figure}

\subsubsection{Single realization: numerical results}
\label{sec:numerical_section}

Let us first numerically probe the behaviour of $I_{\rm free}$, $\OTOC_n$, and the partial sums in \eqref{424} for a single realization for a random GUE matrix, with $d=500$. We observe that: 
\begin{enumerate}
\item $I_{\rm free}$ from the Coulomb gas formula~\eqref{fmiformula}, shown in the black curves in Fig.~\ref{fig:ifree_gue}, is finite for any finite time, and decays monotonically with time. 

We have explicitly shown the case where $A \neq B$, $\Tr[AB]=0$, and $(AB)^2=\mathbf{1}$. Other cases, such as $A=B$, or $AB= i\mathbf{1}$, show qualitatively similar behaviour of $I_{\rm free}$.

\item The partial sums of $\otoc_n$ in \eqref{424}, also shown in Fig.~\ref{fig:ifree_gue}, converge rapidly with $n_{\rm max}$.  The left plot of Fig.~\ref{fig:ifree_gue} shows that at sufficiently late times, just the leading $n=2$ term is sufficient to capture the full behaviour of $I_{\rm free}(A(t):B)$. The right plot shows that at earlier times, we need  increasingly large $n_{\rm max}$.  Each of the partial sums also shows a monotonic decay with $t$.  

%, due to the fact that $\OTOC_n(t)$ decays faster with time for larger $n$. 

\item The behavior of individual $\OTOC_n$ as a function of time is shown in Fig.~\ref{fig:otoc_gue_decay}. For $A=B$, $\otoc_n$ for all $n$ are non-zero, while for $A$ orthogonal to $B$, only even $\otoc_n$ are non-zero. We note that: 
\begin{enumerate}
\item[(i)] The initial decay of $\otoc_n$ is monotonically faster for higher $n$.
\item[(ii)] At later times, $\otoc_n(t)$ oscillates with $t$, leading to non-monotonic behavior as a function of both $t$ and $n$. The oscillations begin at an earlier time scale for larger $n$.
\end{enumerate}
We will observe similar  patterns in the local chaotic systems that we will study in the next two subsections. Due to point (ii), the convergence of the OTOC partial sums and their monotonic behaviour as a function of $t$, which we noted in the previous point, is not  obvious simply from observing the behavior of individual $\otoc_n$.

%It is worth noting that the oscillatory behaviour of the higher OTOCs with time, which can be seen from Fig.~\ref{fig:otoc_gue_decay}, is not reflected in the partial sums of \eqref{424} shown in Fig.~\ref{fig:ifree_gue}. 

 %In Appendix~\ref{app:gue_checks_otoc}, we compare the partial sums in~\eqref{424} to the partial sums of the exact expression \eqref{full_otoc}, and confirm that the difference between them can be ignored from early to intermediate times, where the Coulomb gas formula can be identified with~$I_{\rm free}$. 

\item The late-time saturation value of the Coulomb gas formula is $\sim d^{-0.84}$ and {\it negative}. See Fig.~\ref{fig:d_dependence} in Appendix~\ref{app:gue_checks}. By definition, the value of $I_{\rm free}$ cannot be negative; but recall from the discussion at the end of Sec.~\ref{sec:formula} that the Coulomb gas formula and $I_{\rm free}$ can be identified with each other only up to corrections of $O(1/d^z)$ for $z>0$. Hence, the only conclusion we can draw is that the late-time value of $I_{\rm free}$ is $O(1/d^z)$ for some $z>0$.

The condition~\eqref{327} for $\Delta_{\rm min}$ is satisfied for the full range of times shown in the plot, so that the Coulomb gas formula can be identified with $I_{\rm free}$ from early to intermediate times. See Appendix~\ref{app:gue_checks} for details.

\begin{figure}[!h]
    \centering
    \begin{subfigure}{0.49\textwidth}
    \centering
    \includegraphics[width=\linewidth]{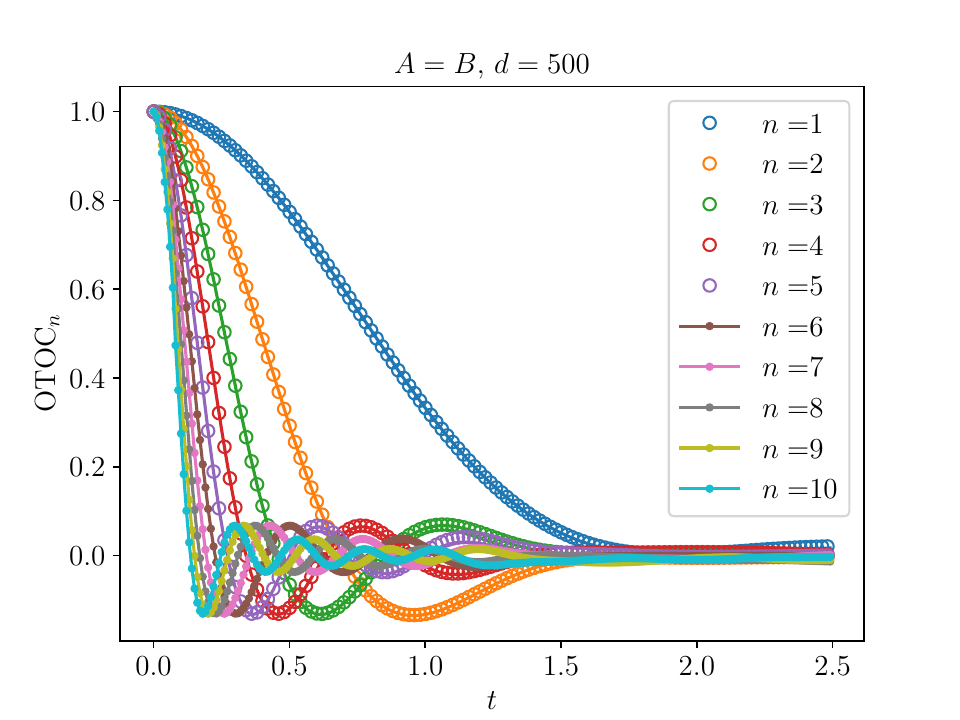}
    \end{subfigure}
    \begin{subfigure}{0.49\textwidth}      
    \centering
    \includegraphics[width=\linewidth]{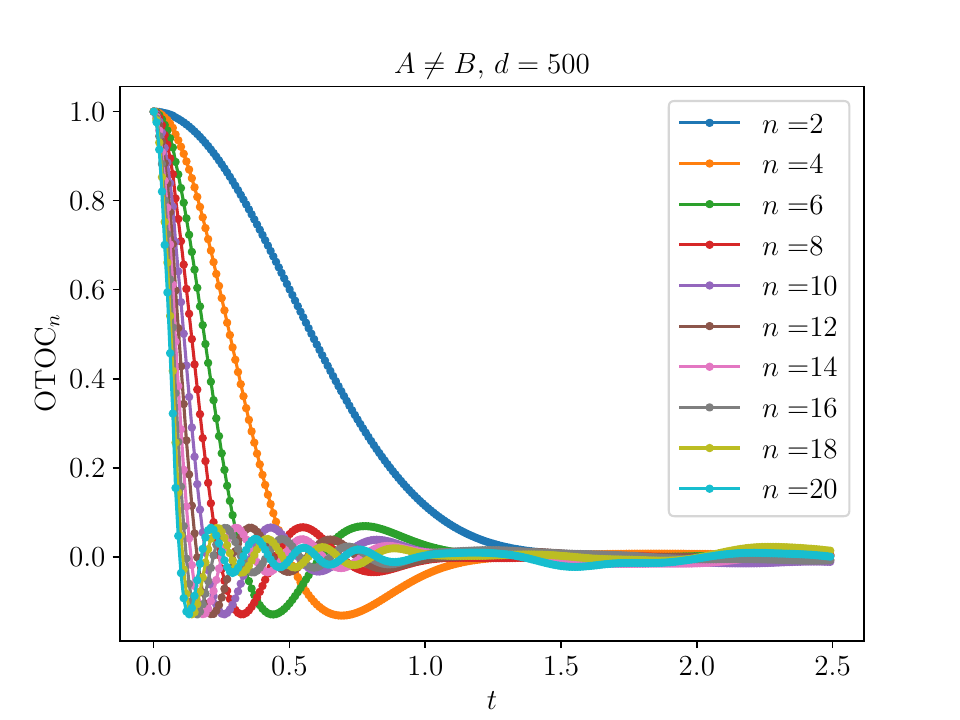}
    \end{subfigure}
    \caption{We show the behaviour of ${\OTOC_n}(t)$ for the case $A=B$ on the left for a single realization of the GUE ensemble, and for the case $A \neq B$, $\Tr[AB]=0$, $(AB)^2 = \mathbf{1}$ on the right. The data in the right plot provides the input for the partial sums shown in Fig.~\ref{fig:ifree_gue}. Note that we have not shown the odd $n$ OTOCs in this $A\neq B$ case, as they show small fluctuations around 0 for all times and give a negligible contribution to $I_{\rm free}$. In the left plot, for $n=1$ to $n=5$, the continuous lines are not interpolations of the data points, but instead correspond to the analytic expressions from the large $d$ limit of the ensemble average obtained in Table~\ref{gue_table}.}
    \label{fig:otoc_gue_decay}
\end{figure}

\end{enumerate}

\subsubsection{Ensemble average} \label{sec:average_section}

To go beyond the numerical analysis of the previous section, we need to consider ensemble averages of $\otoc_n$ and $I_{\rm free}$ over the GUE ensemble. First note that from our argument in Sec.~\ref{sec:ensemble} and Appendix~\ref{sec:smoothness}, due to the smoothness of the measure~\eqref{pgue},  for any $t>0$, the averaged spectral density $\bar \rho(x)$ of the real part of $A(t)B$ contains no  $\delta$ functions. As a result, the averaged free mutual information $\sI$ is finite for any $t>0$ (details on the smoothness condition and how it is satisfied for the ensemble \eqref{pgue} are discussed in Appendix~\ref{sec:smoothness}). 

While directly finding the averaged spectral $\bar \rho(x)$ as a  function of $t$ to compute the ensemble FMI $\sI$ from the formula \eqref{eq:integralFMI} is challenging, the average $\overline{\otoc_n(A(t):B)}$ is more tractable. 
In this section, we will use a technique from free probability to obtain a general analytic formula for $\overline{\otoc_n}$ under GUE evolution.~\footnote{Note that~\cite{cotler_hunterjones} previously proposed a much simpler formula for $\frac{1}{d}\Tr[A_1B_1(t)...A_nB_n(t)]$ in this model by identifying a leading contribution in the average over $(U \otimes U^{\ast})^{2n}$ for the Haar-random $U$ appearing in \eqref{udef}. However, their result applies only in the case where all $A_i, B_i$ are distinct.  In the present case where all $A_i=A$ and $B_i=B$, the analysis from Weingarten functions in the random unitary average becomes intractable for general $n$, with many non-trivial leading contributions coming from terms proportional to $\rm Wg(\sigma, \tau)$
for $\sigma\neq \tau$. One goal of this subsection is to illustrate how free probability techniques can give a practical tool to simplify such calculations.} While the formula turns out to be somewhat complicated, we write it in an explicit form up to $n=5$, and find excellent agreement with the numerical result from $d=500$ from the previous section. See the caption of Fig.~\ref{fig:otoc_gue_decay}.~\footnote{The fact that the variance of $\otoc_n$ is small for large $d$  confirms that we can approximate $\bar \rho(x) \bar \rho(y) \approx \overline{\rho(x)\rho(y)}$  for large $d$ in this model.} In principle, these formulas for $\overline{\otoc_n}$ can be used to obtain an analytic form 
for $\sI$ through \eqref{425}. We are not able to explicitly perform this sum due to certain difficulties mentioned below, and leave this task to future work.  

It is easiest to illustrate the technique for the case $A=B$, where our task is to compute 
\be 
\overline{\OTOC_n(t)} = \Tr[\le(A U_{\rm GUE}(t)^{\dagger}A U_{\rm GUE}(t)\ri)^n] \, . 
\ee
 The key observation is that since 
 \be
 U_{\rm GUE}(t) = U e^{-iH^0_{\rm GUE}t} U^{\dagger} \label{udef}
 \ee
 for a Haar-random unitary $U$ (where $H^0_{\rm GUE}$ is some fixed matrix),  $A$ and  $U_{\rm GUE}(t)$ approach freely independent non-commuting variables $a$ and $u$ in the large $d$ limit. As explained in more detail in Appendix~\ref{app:gue}, this leads to the following formula~\cite{nica2006lectures}:
\begin{equation}
\lim_{d\to \infty}\overline{\OTOC_n(t)} = \sum_{\pi\in {\rm NC}(2n)} \prod_{V\in\pi}\kappa_{|V|}(a)\prod_{W\in \pi^c}m_{|W|}(u)\ , \label{otoc_gen}
\end{equation}
where ${\rm NC}(2n)$ denotes the non-crossing (NC) partitions of $2n$ elements, $\pi^c$ is the Kreweras complement of $\pi$ in ${\rm NC}(n)$, and $V$ denotes a block in $\pi$. See Fig.~\ref{fig:NC} for explanations and illustrations of NC partition and Kreweras complements.  $m_n(u)$ refers to the $n$-th moment of $u$,
  \be 
m_n(u) \equiv \lim_{d\to\infty} \frac{1}{d}\overline{\Tr[e^{-inH_{\rm GUE}t}]} = \frac{J_1(2nt)}{nt}
  \ee
 where $J_1$ is the Bessel function of the first kind.  This expression comes from the fact that the averaged  density of eigenvalues of $H_{\rm GUE}$ is a semicircle distribution. For subsequent expressions, it will be convenient to introduce the notation 
  \be 
\eta(t)\equiv  \frac{J_1(2t)}{t} \, . 
\ee
Note that $|\eta(t)|\leq 1$ for all $t$. 
$\kappa_n(a)$ refers to the ``free cumulants'' of $a$, which in this case are given by
  \be 
\kappa_n(a) =\begin{cases}
(-1)^{\frac{n}2-1}C_{\frac{n}2-1} \ ,\quad n\ \text{is even}\\
0\ ,\ \ \quad\quad\quad\quad\quad\quad n\ \text{is odd}
\end{cases}
  \ee
where $C_k=\frac{1}{k+1}\binom{2k}{k}$ denotes the Catalan number. See Appendix~\ref{app:gue} for more details including the definition of free cumulants.  

\begin{figure}[t]
\centering
% ---------- LEFT: NC vs crossing ----------
\begin{subfigure}[t]{0.48\textwidth}
\centering
\newcommand{\gap}{1.25em} 
\newcommand{\sep}{1.25em} 

\begin{tikzpicture}[baseline=-2pt]
  % ---- (12)(34) non-crossing (left) ----
  \begin{scope}[xshift=0em]
    \node[inner sep=0.5pt] (a1) at (0,0) {};
    \node[inner sep=0.5pt] (a2) at (\gap,0) {};
    \node[inner sep=0.5pt] (a3) at (2*\gap,0) {};
    \node[inner sep=0.5pt] (a4) at (3*\gap,0) {};
    \draw (a1) --++(0,-1.2em) -| (a2);
    \draw (a3) --++(0,-1.2em) -| (a4);
    % digits (no extra text)
    \node[fill=white, inner sep=0.6pt] at (a1) {$1$};
    \node[fill=white, inner sep=0.6pt] at (a2) {$2$};
    \node[fill=white, inner sep=0.6pt] at (a3) {$3$};
    \node[fill=white, inner sep=0.6pt] at (a4) {$4$};
  \end{scope}

\hspace{1cm}
  % ---- (13)(24) crossing (right) ----
  \begin{scope}[xshift={3*\gap + \sep}]
    \node[inner sep=0.5pt] (b1) at (0,0) {};
    \node[inner sep=0.5pt] (b2) at (\gap,0) {};
    \node[inner sep=0.5pt] (b3) at (2*\gap,0) {};
    \node[inner sep=0.5pt] (b4) at (3*\gap,0) {};
    \draw (b1) --++(0,-1.2em) -| (b3);
    \draw (b2) --++(0,-2.0em) -| (b4);
    % digits (no extra text)
    \node[fill=white, inner sep=0.6pt] at (b1) {$1$};
    \node[fill=white, inner sep=0.6pt] at (b2) {$2$};
    \node[fill=white, inner sep=0.6pt] at (b3) {$3$};
    \node[fill=white, inner sep=0.6pt] at (b4) {$4$};
  \end{scope}

  % tighten bounding box
  \useasboundingbox (-0.25em,-2.4em) rectangle (3*\gap + \sep + 3*\gap + 0.5em, 1.0em);
\end{tikzpicture}

\subcaption{Non-crossing vs. crossing partitions.}
\end{subfigure}
\hfill
% ---------- RIGHT: Kreweras complement diagram ----------
\begin{subfigure}[t]{0.48\textwidth}
\centering
\begin{tikzpicture}[baseline=-2pt]
  \node[inner sep=1pt] (n1) at (0em,0) {};
  \node[inner sep=1pt] (n2) at (2em,0) {};
  \node[inner sep=1pt] (n3) at (4em,0) {};
  \node[inner sep=1pt] (n4) at (6em,0) {};
  \node[inner sep=1pt] (n5) at (8em,0) {};
  \node[inner sep=1pt] (n6) at (10em,0) {};
  \node[inner sep=1pt] (n7) at (12em,0) {};
  \node[inner sep=1pt] (n8) at (14em,0) {};
  \node[inner sep=1pt] (m1) at (1em,0) {};
  \node[inner sep=1pt] (m2) at (3em,0) {};
  \node[inner sep=1pt] (m3) at (5em,0) {};
  \node[inner sep=1pt] (m4) at (7em,0) {};
  \node[inner sep=1pt] (m5) at (9em,0) {};
  \node[inner sep=1pt] (m6) at (11em,0) {};
  \node[inner sep=1pt] (m7) at (13em,0) {};
  \node[inner sep=1pt] (m8) at (15em,0) {};
  \draw[densely dashed] (n1) --++(0,-1.5em) -| (n2);
  \draw (m1) --++(0,-1em);
  \draw (m4) --++(0,-1em);
  \draw (m7) --++(0,-1em);
  \draw[densely dashed] (n4) --++(0,-1.5em) -| (n5);      
  \draw[densely dashed] (n3) --++(0,-2.5em) -| (n6);
  \draw (m3) --++(0,-2em) -| (m5);      
  \draw[densely dashed] (n7) --++(0,-1.5em) -| (n8);
  \draw (m2) --++(0,-3em) -| (m6);
  \draw ($(m2)+(0,-3em)$)  -| (m8);
  \node [fill=white, right, inner sep=1pt] at ($(n1)+(-0.375em,0em)$) {$1\hspace{0.5em}\overline{1}\hspace{0.5em}2\hspace{0.5em}\overline{2}\hspace{0.5em}3\hspace{0.5em}\overline{3}\hspace{0.475em}4\hspace{0.5em}\overline{4}\hspace{0.5em}5\hspace{0.5em}\overline{5}\hspace{0.5em}6\hspace{0.5em}\overline{6}\hspace{0.5em}7\hspace{0.5em}\overline{7}\hspace{0.5em}8\hspace{0.5em}\overline{8}$};  
  \useasboundingbox (-0.25em,-3.125em) rectangle (15.25em, 0.75em);
\end{tikzpicture}
\subcaption{Kreweras complement.}
\end{subfigure}
\caption{\textbf{NC partition and Kreweras complements.} A noncrossing partition of a finite set is a partition in which no two blocks cross each other. The left panel illustrates non-crossing vs crossing partitions on a set of four elements: $(12)(34)$ is NC and $(13)(24)$ is not.  The Kreweras complement of $\pi$ in ${\rm NC}(n)$ is defined as the maximal NC partition of a set of $n$ elements interlaced with the original set of $n$ elements on which $\pi$ is defined, such that $\pi\cup\pi^c$ is an NC partition of $2n$ elements. Here, maximality is defined with respect to a partial order: $\pi<\sigma$ if each block of $\pi$ is completely contained in one of the blocks of $\sigma$.  The right panel illustrates the Kreweras complement for $\pi=(12)(36)(45)(78)$ (taken from~\cite{speicher2019lecture}), $\pi^c=(1)(268)(35)(4)(7)$. }
\label{fig:NC}
\end{figure}
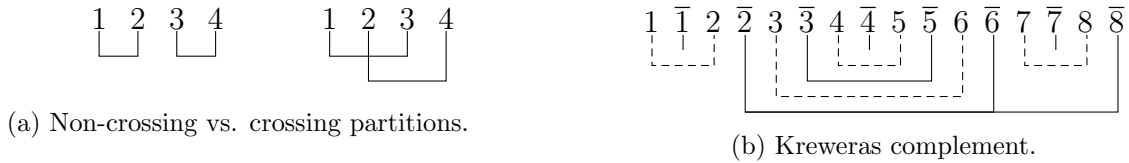
Putting the expressions for $m_n$ and $\kappa_n$ into \eqref{otoc_gen}, we obtain the following formula:  
\begin{equation}\label{eq:otocGUE}
    \overline{\OTOC_n(t)} =  \sum_{\substack{\pi\in {\rm NC}(2n),\\ |V|\,\textrm{even},\, \forall V\in\pi}} \prod_{V\in\pi}(-1)^{\frac{|V|}{2}-1}C_{\frac{|V|}{2}-1}\prod_{W\in \pi^c}\eta(|W|t)
\end{equation}
where the sum is over those NC partitions of $2n$ whose blocks all have even lengths, and we shall simply call them ``even NC partitions.'' For instance, we have
\begin{align}
&\overline{\OTOC_1(t)}=\eta(t)^2\\
&\overline{\OTOC_2(t)}=2\eta(t)^2\eta(2t)-\eta(t)^4 \, .
\end{align}
We list the resulting formulas up to $n=5$ in Table~\ref{gue_table} in Appendix~\ref{app:gue}, and compare them to the numerical data from a single sample for $d=500$ in Fig.~\ref{fig:otoc_gue_decay}.  

The expressions become increasingly complicated for larger $n$, but one general pattern is worth noting: the smallest number of total factors of the function $\eta$ that appear in any term in $\otoc_n$ is always $\geq n+1$.  This is a consequence of the combinatorial lemma \ref{lemma_noncross_1}. We will see a similar combinatorial pattern in the random circuit calculation of the next section.  Moreover, all terms with $n+1$ factors of $\eta$ come from cases where $\pi$ consists  only of 2-cycles (or is a ``non-crossing perfect matching''), and as a result the sum of the coefficients of these terms is the Catalan number $C_{n}$. Since $|\eta(t)|\leq 1$ and its envelope decays with time,
\begin{equation}
    \eta(t) = t^{-3/2}\pi^{-\frac12}\left(\cos(2t-3\pi/4)+O(t^{-1})\right), \quad t\gtrsim 1 \label{eta_approx}
\end{equation}
 the set of terms with only $n+1$ powers of $\eta$ should dominate in the full expression for $\overline{\otoc_n}$ at late enough times (which we also find in practice to be $t \gtrsim 1$), giving an upper bound of the form 
\be 
\otoc_n \leq C_{n} t^{-3 (n+1)/2} \pi^{-(n+1)/2} \label{511}
\ee
up to the cosine factors. In principle, \eqref{511} could be used to argue for convergence of the OTOC sum formula at late enough times (as the Catlan numbers grow exponentially rather than factorially).  However, at the times when the approximations  in \eqref{eta_approx} and \eqref{511} hold, we are already in the regime where the Coulomb gas formula becomes negative and its precise $O(1/d^z)$ value does not have a good physical interpretation. To analytically explain the numerically observed convergence of the OTOC sum formula in Fig.~\ref{fig:ifree_gue} at physically relevant time scales, we would  therefore need to take into account the full set of terms like those  in Table~\ref{gue_table}. For now, we will treat the numerical observation as sufficient evidence of convergence.

\subsection{Random unitary circuits}
\label{sec:random_circuit}

Let us now consider random unitary circuits with the brickwork pattern shown in Fig.~\ref{fig:randomcircuit_diagram}, which provide a useful  toy model for quantum chaotic systems with local interactions. Each site has a finite local Hilbert space dimension $q$. Each unitary in the circuit such as $U_{x', x'+1}^{(t')}$ is independently drawn from the ensemble of Haar-random matrices $\mathbf{U}(q^2)$.  We take two operators $A$ and $B$ in the circuit with initial separation $x$, and consider $\otoc_n$ and $I_{\rm free}$ as functions of $t$ and $v=x/t$.

Due to the sharp light-cone of this circuit model, for $v>1$
all $\otoc_n(t,v)$ are equal to 1 for even $n$ and zero for odd $n$. The decay of $\otoc_2(t,v)$ with increase in $t$ for fixed $v$, or with decrease in $v$ for fixed $t$,  captures the growth of operator size in the space of degrees of freedom. Previously, $\otoc_2$ in this model has been studied in detail in~\cite{operator_spreading_adam, operator_spreading_tibor}, where in particular it was found that the initial decay is of the form 
\be 
\overline{\otoc_2}(t, v) \approx   e^{- \frac{(v -v_B)^2}{2D} t} \, ,  
\ee
indicating a ballistic spreading of the operator front in space at some characteristic velocity $v_B$, as well as a broadening of the front with some diffusion constant $D$. 

\begin{figure}[!h]
    \centering
    \includegraphics[width=0.6\linewidth]{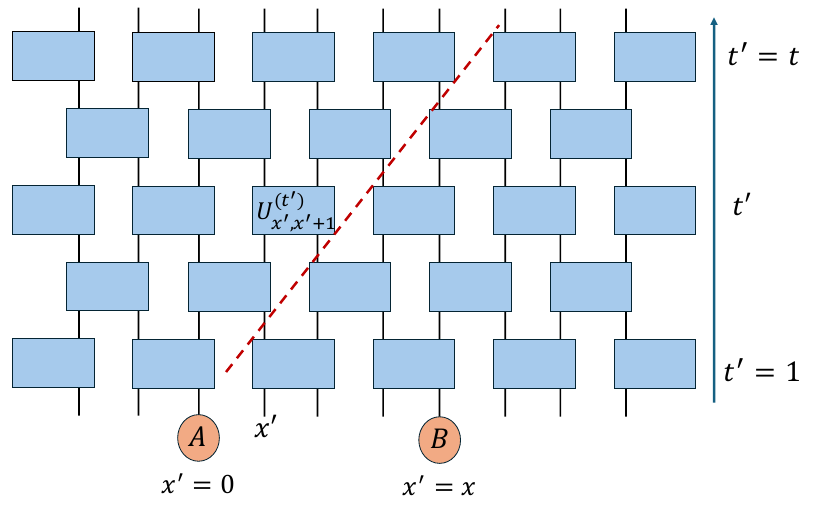}
    \caption{Brickwork random unitary circuit model.}
    \label{fig:randomcircuit_diagram}
\end{figure}

We would now like to understand our more fine-grained notion of spreading in operator space in this model, quantified by $I_{\rm free}$ and  $\otoc_n$.  Like in the ensemble average of the GUE model, the average of $\overline{\otoc_n}$
 is more tractable than that of $I_{\rm free}$. $\overline{\otoc_n}$ in this model can be evaluated by mapping to a partition function calculation on a triangular lattice with a Hilbert space of dimension $n!$ at each site, similar to previous calculations for R\'enyi entropies in~\cite{zhou_nahum}. We  discuss this calculation  in a mostly self-contained way in   Appendix~\ref{app:random_circuits}.  
 
As explained in Appendix~\ref{app:random_circuits}, it is relatively  straightforward to see that
\be 
\overline{\OTOC_n}(t,v) = 0 \,, \quad n \text{ odd} \,  
\ee
for all $v$ and $t$. 
For even $n$,
 the $n>2$ case is significantly more technically challenging than  the $n=2$ case that has previously been studied, and requires the use of a number of assumptions. In particular, we assume that the dynamics of domain walls between permutations in the partition function for $\overline{\otoc_n}$ can be captured by the ``membrane picture'' for random unitary circuits, developed in the context of R\'enyi entropies in~\cite{zhou_nahum}. A major complication in the calculation of 
 $\overline{\otoc_n}$ compared to that of the R\'enyi entropies is that there are now many different types of domain walls, associated with all possible permutations in $\sS_n$, that can potentially contribute to the partition function at leading order. In Appendix~\ref{app:random_circuits}, we explain how to combine inputs from the membrane picture with certain  combinatorial facts in order to rule out contributions from a subset of permutations. Based on this structure, we conjecture a result of the following form: 
\be 
\overline{\otoc_n}(t, v) = \begin{cases}
1 & v > v_B\\
\sum_{k=n/2}^{n-1} c^{(n)}_k q^{-k(\sE_{k+1}(v)-v)t}  & v < v_B \label{otoc_formula_randomcircuit_mt}
\end{cases}, \quad n \text{ even }\,. 
\ee
In both lines, we have ignored exponentially suppressed corrections in $t$ (except in the case $v>1$ in the first line, where the result is exact from the light-cone structure). 
 Here $\sE_n(v)$ is a function called the \emph{membrane tension function}, whose relevant properties are summarized in Appendix~\ref{sec:membrane_assumptions}, and $c^{(n)}_k$ are certain coefficients with either sign such that $\sum_{k=n/2}^{n-1} c^{(n)}_k=1$. The argument for~\eqref{otoc_formula_randomcircuit_mt} is presented in detail in Appendix~\ref{app:random_circuits}. \eqref{otoc_formula_randomcircuit_mt} is a conjecture due to several assumptions stated in the appendix, which we expect are correct but are not able to prove rigorously. The case $n=4$ was recently studied in \cite{google_paper} with a combination of numerical and analytical techniques. While the final result in \cite{google_paper} is obtained numerically and is not phrased in terms of the membrane tension like in \eqref{otoc_formula_randomcircuit_mt}, the combinatorial structure that they find is consistent with the $n=4$ case of our formula.

As explained in more detail in the appendix around \eqref{otoc_final_result}-\eqref{final_rc}, for $v\lesssim v_B$, for sufficiently large $t$, we can approximate 
\be 
\overline{\otoc_n}(t,v) \approx c^{(n)}_{n/2} \,  e^{-n/2 \frac{(v -v_B)^2}{2D} t}\ , \quad n \text{ even} . \label{very_near_front}
\ee
Then if $c^{(n)}_{n/2}$ is independent of $n$, from summing the series \eqref{425} we get the following expression for the averaged FMI: 
\be 
\sI(t,v) \approx c\sum_{n=1}^{\infty}\frac{1}{n}\le(e^{-\frac{(v-v_B)^2}{D}t}\ri)^{n}  = -c\log (1 - e^{-\frac{(v-v_B)^2}{D}t})\, , \quad v< v_B \,  . \label{toy_model}
\ee
where $c = (c_{n/2}^{(n)})^2$. 
This expression is unlikely to be quantitatively accurate
as there is no particular reason to expect that $c_{n/2}^{(n)}$ is independent of $n$; however, it can be seen as a toy model that captures the qualitative behavior of $I_{\rm free}$ under local chaotic dynamics in the regime $v \lesssim v_B$ close to the operator front at late times. Note that the crucial ingredient that makes the convergence of $I_{\rm free}$ plausible for $t>x/v_B$ is that the smallest $k$ appearing in the exponent in \eqref{otoc_formula_randomcircuit_mt} grows with $n$. This growth of the smallest $k$ with $n$ is a consequence of the arguments based on the membrane picture and the combinatorial structure discussed in Appendix~\ref{app:random_circuits}. Note also that the structure of the sum~\eqref{otoc_formula_randomcircuit_mt} is similar to \eqref{eq:otocGUE} in the GUE case, where the smallest number of powers of $\eta$ appearing in \eqref{eq:otocGUE} grows with $n$. Away from the front of the operator, the mechanism for the convergence of the OTOC sum formula is likely more complicated, with all terms in \eqref{otoc_formula_randomcircuit_mt} playing an important role.  

Considering the implications of the arguments of Appendix~\ref{sec:smoothness} for the convergence of $\sI$ in this case reveals a further subtlety. Due to the smoothness of the ensemble of random unitary circuits according to the criterion discussed in Appendix~\ref{sec:smoothness}, we conclude from this general argument that $\bar \rho(x)$ has no delta functions for any $v<1$. This implies that $\sI$ must be finite for any $v<1$, and hence the OTOC sum must converge in this regime. For $v<v_B$, this is consistent with~\eqref{otoc_formula_randomcircuit_mt} as discussed above. For $v_B<v<1$, this statement appears to be in contradiction with \eqref{otoc_formula_randomcircuit_mt} if all $\otoc_n(t,v)$ can be approximated as 1. Note that for $v_B<v<1$, the first line of \eqref{otoc_formula_randomcircuit_mt} has corrections of the form $O(g_1(n,v)e^{-g_2(n,v)t})$ which are beyond the scope of our current calculation. Here $g_1$ and $g_2$ are some functions of $n$. If $g_2(n)$ in these corrections grows rapidly for large $n$, this could explain convergence of the OTOC sum for $v_B<v<1$. Further scrutiny of this regime is beyond the scope of our current analysis, but presents an interesting puzzle for future work.

%This is likely because the decaying tail of the OTOC series is beyond the scope of our analysis. It is plausible that when $n$ is exponentially large in system size, various combinatorial factors become too significant to be ignored and consequently $\otoc_n$ is no longer approximated by~\eqref{otoc_formula_randomcircuit_mt}. 

\subsection{Chaotic spin chain} \label{sec:chaotic_spin}

Next, let us numerically probe the behaviour of $\otoc_n$ and $I_{\rm free}$ in the chaotic mixed field Ising model, 
\be 
H = \sum_{i=1}^L( Z_i Z_{i+1} + g X_i + h Z_i\,) . \label{mixed_field} 
\ee
with $g= -1.05$, $h=0.5$, and periodic boundary conditions. We take the system size $L=10$, $A = Z_{x=0}$, and $B = Z_{x=5}$. The result is shown in Fig.~\ref{fig:chaotic_spin}. Other choices of Pauli operators $A$ and $B$ lead to similar results. 

\begin{figure}[!h]
    \centering
    \begin{subfigure}{0.49\textwidth}
        \centering
    \includegraphics[height=6cm]{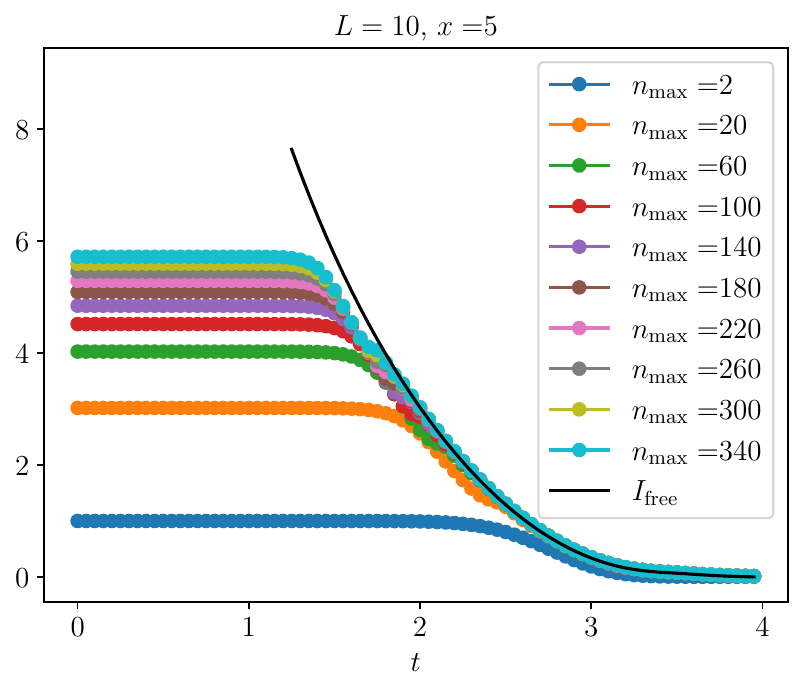}
    \end{subfigure}
    \begin{subfigure}{0.49\textwidth}
    \centering
    \includegraphics[height=6cm]{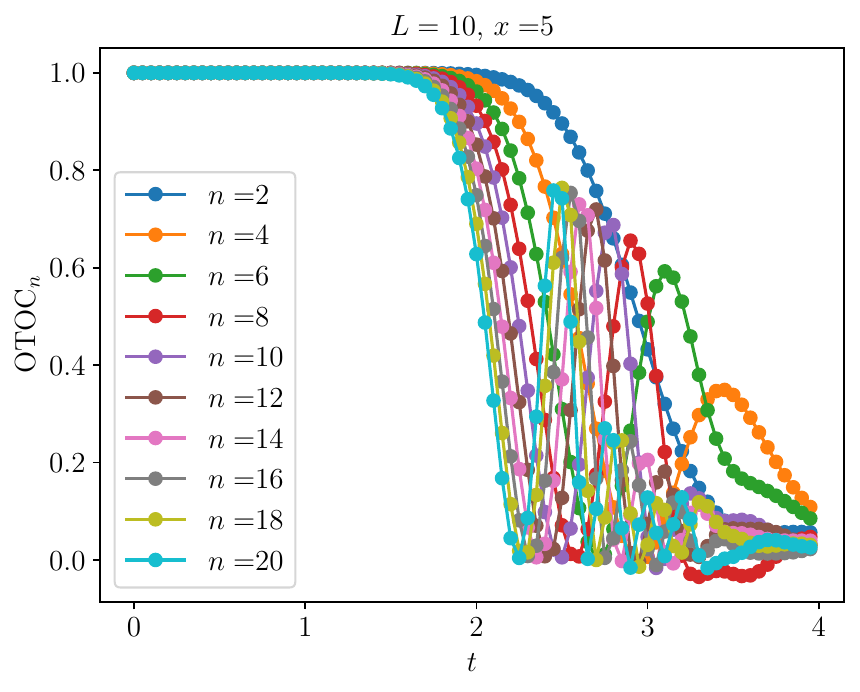}
    \end{subfigure}
        
    %\begin{subfigure}{0.49\textwidth}
    %\centering
    %\includegraphics[height=6cm]{otoc_n_dep_pbc_5.pdf}
    %\end{subfigure}
        
    \caption{On the left, we show $I_{\rm free}$ from the Coulomb gas formula (black curve) and the partial sums of \eqref{424} in the chaotic spin chain model. On the right, 
     we show the time-dependence of $\otoc_n$ for even $n$. Odd $n$ give a negligible contribution for all times, as expected from the random circuit calculation.}
    %We show $\Delta_{\rm min}$ as a function of time in the left bottom figure. The value for $t\lesssim 1.5$ is $\sim 10^{-16}$, leading to a divergent value of $I_{\rm free}$ at the numerical precision we are working with.} %The bottom right figure shows the $n$-dependence of $\otoc_n$ for a few times close to the operator front.} 
    \label{fig:chaotic_spin}
\end{figure}

We observe that:   
\begin{enumerate}
\item The value of $I_{\rm free}$ from the Coulomb gas formula, shown on the left in Fig.~\ref{fig:chaotic_spin}, starts to decay at a time scale $t^{\ast}$ close to $x/v_B$. Note that $v_B \approx 2$ for this model~\cite{mezei_stanford,jonay}, but its precise value in Hamiltonian systems depends on a choice of convention about how to define the operator front. Relatedly, due to the lack of a sharp lightcone, the precise time at which $I_{\rm free}(t)$ becomes finite is ill-defined, and depends on the numerical precision with which we evaluate the eigenvalues of $A(t)B$. For the case shown in the figure, we work at a precision of $10^{-16}$. $\Delta_{\rm min}$ becomes non-zero at this precision for $t > t^{\ast} \approx 1.25$, leading to a finite value of the FMI after such times. $I_{\rm free}(t)$ shows monotonic decay with time. 
\item The left panel of Fig.~\ref{fig:chaotic_spin} also shows the pattern of convergence of the OTOC partial sums for various $n_{\rm max}$ to $I_{\rm free}$, which is similar to that in the GUE model. Each of the OTOC partial sums shows a monotonic decay with time. 
\item The behavior of individual $\otoc_n$, shown on the right in Fig.~\ref{fig:chaotic_spin}, shows an initial decay  that is monotonically faster for higher $n$. In particular, the $n$-dependence for fixed $t$ shown in Fig.~\ref{fig:n_dependence} suggests that  if we first take the $t\to t^{\ast}$ limit and then $n_{\rm max}\to \infty$, the $\otoc_n$ decay monotonically with $n$.
This is the regime where a formula such as \eqref{very_near_front} from the random circuit analysis could apply in the chaotic spin chain for large enough $x$.

For any $n$, we see oscillations of $\otoc_n$ with time at intermediate times, which start earlier and have larger magnitude for higher $n$.  Similar oscillations could arise in the random circuit calculation from summing over all terms with different signs in \eqref{otoc_formula_randomcircuit_mt}. This also entails oscillations of $\otoc_n(t)$ for a fixed $t$ as a function of $n$, as we see in Fig.~\ref{fig:n_dependence}. In this intermediate-time regime, the convergence of the OTOC sum formula is explained by the decaying envelope of $\otoc_n$ instead of a strict monotonic decay with $n$.

\end{enumerate} 

\begin{figure}[!h]
\centering
\includegraphics[width=0.6\textwidth]{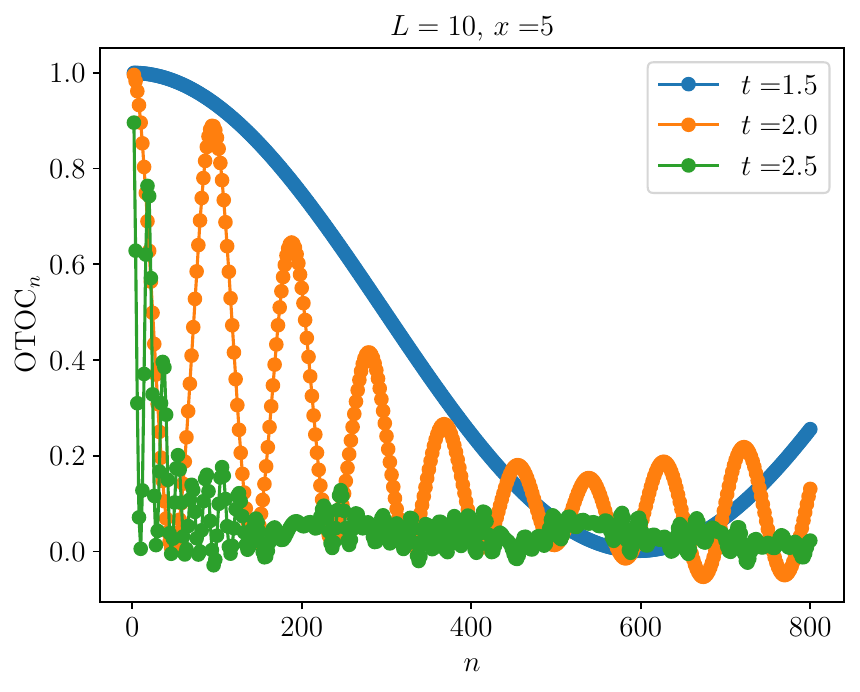}
\caption{$n$-dependence of $\otoc_n$ in the chaotic spin chain model at three different times. As we approach $t^{\ast}\approx1.25$, $\otoc_n$ decay monotonically with $n$ for an increasing range of $n$.}
\label{fig:n_dependence}
\end{figure}

\iffalse 
We therefore see that the initial decay of $\otoc_n$ and $I_{\rm free}$ (i.e., the regime where $B$ is close to the front of $A(t)$) appears to be qualitatively similar between chaotic spin chain systems and random unitary circuits. The differences between $\otoc_2$ in cases with and without conservation laws such as energy conservation become significant in cases where $B$ deep inside the lightcone of $A(t)$~\cite{vedika_conservation, tibor_conservation} and in the $L$-dependence of the late-time saturation value~\cite{brandao}. We expect similar effects of energy conservation on $\otoc_n$. In particular, by an extension of the argument of~\cite{brandao}, the late-time value of all $\otoc_n$ in cases with conservation laws decays polynomially rather than exponentially with $L$ \SV{(check)}.\JW{\href{https://chatgpt.com/share/68c35f1a-265c-800e-8133-e0b59556788a}{ChatGPT} suggested that the general result should be like $L^{-n}$. We can decide if we want to add an Appendix for it.} This implies that the saturation value of $I_{\rm free}$ also decays polynomially with $L$. \JW{ What's the late time value of the OTOCs? Do we observe a plateau in Fig.7?} \SV{With my existing code we can't reliably extract the plateau because the Krylov dimension is too small. I will run larger Krylov dimension over the weekend as it's more time-consuming.}\JW{Ok.}
\fi

\section{Free mutual information in non-chaotic systems} \label{sec:nonchaotic}

In this section, we discuss the behaviour of the Coulomb gas formula and the free mutual information in a few examples  which cannot be seen as generic chaotic models. In three of these cases -- Clifford unitaries, PFC unitaries, and free fermions-- we find degeneracies in the spectrum of $A(t)B$ even at late times, so that the FMI is infinite. In each of these models, the behavior of $\otoc_2$ reflects delocalization in physical space, even though the divergence of the FMI indicates a lack of spreading in operator space.  The precise mechanism underlying the degeneracy of the spectrum of $A(t)B$ is different in each case. In the interacting integrable Heisenberg model, for operators that do not commute with the symmetries of the system, we find behavior similar to the generic chaotic spin chains of Sec.~\ref{sec:chaotic_spin}.  

%Free mutual information is also useful in quantifying the lack of chaos in non-chaotic systems. We study two classes of non-chaotic systems: Integrable spin chains and computation-theoretic constructions of unitary designs. In these models, we find large free mutual information values that persist until late time, which can either plateau at some finite value or remain infinite. Free mutual information quantifies the lack of chaos by how far operators $A(t)$ and $B$ are from being fully scrambled.

\subsection{Unitary designs}
\label{sec:designs}

Recall from the discussion of Sec.~\ref{sec:haar} that under a single step of evolution by a Haar-random unitary, the free mutual information becomes zero in the limit of large Hilbert space dimension. It is then natural to ask whether ensembles known as unitary designs, which replicate certain features of Haar-random unitaries, also cause the free mutual information to vanish or become small after a single step. See for instance~\cite{Metger_2024} for the precise definition of unitary $n$-designs. 

We consider two examples of unitary designs below: 
\begin{enumerate}
\item {\it Random Clifford unitaries}, which form exact 3-designs~\cite{Webb_2016}. This implies in particular that the value of $\otoc_2$ agrees on average between random Cliffords and Haar-random unitaries.  Random Cliffords thus lead to delocalization of operators in physical space, and the related information-theoretic phenomenon of decoupling~\cite{dupuis2010decoupling, Hayden_2007, Hayden_2008,  Lashkari_2013,Brown_2015,Onorati_2017,swingle_tutorial}. This ensemble fails to form an exact or approximate 4-design~\cite{Webb_2016,zhu2016clifford}.   
\item {\it Random PFC ensembles}, which were recently introduced in~\cite{Metger_2024} and shown to form approximate unitary $n$-designs for any $n$ that does not grow with the Hilbert space dimension $d$. As one component of a PFC unitary is a random Clifford, this ensemble also forms an exact 3-design.  
\end{enumerate}

From our OTOC sum formula for the free mutual information, it is clear that if all $\otoc_n$ are equal to the Haar-random values for some unitary dynamics, then the free mutual information must be small. However, under both random Clifford unitaries and random PFC ensembles, for all $n$ that are multiples of 4, $\otoc_n(UAU^{\dagger}:B)=1$. For random Cliffords, this is true for any Pauli operators $A$ and $B$.~\footnote{Note that  we refer to any tensor product operator involving one of the four Pauli matrices $I$, $X$, $Y$, $Z$ at each site as a Pauli.} For the random PFC ensemble, it holds for the case where $A$ is any Pauli and $B$ consists only of Pauli $Z$ and $I$ at each site. This is true for each individual realization from both ensembles, and hence also for the ensemble average. For Cliffords, this behavior is consistent with the fact that they do not form exact  4-designs; for the PFC ensemble, it is consistent because the {\it approximate} $n$-design property, unlike the {\it exact} $n$-design property, does not imply that $\otoc_n$ are close to the Haar values (see~\cite{schuster2025random, pap_new}).

For both random Cliffords and PFC unitaries, the fact that $\otoc_n=1$ for all $n$ that are multiples of 4 comes from the fact that each of the eigenvalues of $A(t)B$ takes only one out of four possible values: +1, $-1$, $i$, or $-i$. This also immediately implies that the spectrum has a large amount of degeneracy, and the Coulomb gas formula diverges. Let us see how this feature of the spectrum comes about in both cases.  

\subsubsection{Cliffords}

The result in the Clifford case immediately follows from the fact that  $UAU^{\dagger}$ is a Pauli matrix if $A$ is a Pauli matrix. Then since $B$ is also a Pauli, there are three cases for the spectrum of $UAU^{\dagger}B$: (i) If $U A U^{\dagger}=B$, then $U A U^{\dagger}B=I$ and all eigenvalues of $U A U^{\dagger}B$ are 1. (ii) If $U A U^{\dagger}\neq B$ but commutes with $B$, then $U A U^{\dagger}B$ is $\pm P$ for some nontrivial Pauli $P$, so that $d/2$ eigenvalues are +1 and $d/2$ are $-1$. (iii) If $U A U^{\dagger}$ anticommutes with $B$, then $U A U^{\dagger}B$ is $\pm i P$ for some Pauli $P$. Half of the eigenvalues are $+i$, and half are $-i$.

%This result also highlights the distinction between scrambling characterized by the information spreading in physical space against the spreading in operator space. Quantum information scrambling was first used by Hayden-Preskill in understanding how an old black hole returns information in the Hawking radiation~\cite{Hayden_2007}. The black hole evolution needs to be sufficiently scrambling to delocalize information in the physical space over the whole system.  The technical core of their reasoning is the quantum decoupling theorem~\cite{dupuis2010decoupling}, which also underpins many of the optimal codes used in quantum information theory~\cite{Hayden_2008}. This motivates further studies of fast scrambling in terms of delocalizing quantum information~\cite{Sekino_2008,Lashkari_2013,Brown_2015,Onorati_2017}. This notion of scrambling can be closely related to the standard $\otoc_2$~\cite{Hosur_2016,Xu_2024}. However, this relation does not hold for higher OTOCs. We just saw that a random Clifford unitary, which is a 3-design, is sufficient for decoupling and rendering $\otoc_2$ small. Nevertheless, the higher-point OTOCs remain large, and thus Clifford unitaries fail badly in scrambling a Pauli observable in the operator space.

%footnote{See also~\cite{Harrow_2021} on the distinction between OTOC scrambling and entanglement scrambling.} 

Since random Clifford unitaries send Pauli matrices to other Paulis, the lack of chaotic behavior in this ensemble can be detected by a number of other measures, including operator entanglement~\cite{prosen_operator, jonay} and the stabilizer R\'enyi entropy~\cite{Leone_2022}. It is not too surprising, then, that the free mutual information shows non-generic behavior in this case. Next, let us turn to the more non-trivial example of the PFC ensemble~\cite{Metger_2024}, which does not send a Pauli matrix to another Pauli matrix. 

\subsubsection{PFC unitaries}

A PFC unitary augments a random Clifford (C) by further multiplying each computational state by a random binary phase (F) and applying a random permutation (P) to the computational basis. More precisely, a unitary in the PFC ensemble is given by the product 
\begin{equation}
     U=P \,  F\,  C \, , \quad \quad   C\in\mathrm{Clifford}, \quad F:\ket{x}\mapsto(-1)^{f(x)}\ket{x}, \quad  P:\ket{x}\mapsto\ket{\pi(x)}
\end{equation}
where $f(x)$ is a binary function and $\pi\in S_d$ is a permutation among computational basis states. The probability measure on this ensemble is induced from a uniformly random Clifford~$C$, a uniformly random binary function~$f(x)$, and a uniformly random permutation~$\pi$.

\iffalse 
Metger el al proved that $(\mathcal E_{\rm PFC},\nu_{\rm PFC})$ forms an approximate unitary $t$-design with an error $O(t/\sqrt{d})$~\cite{Metger_2024}:
\begin{equation}
    ||\mathbb{E}_{U\sim\nu_{\rm PFC}}[U^{\otimes t}\ \cdot \ {U^\dagger}^{\otimes t}] - \mathbb{E}_{U\sim\nu_{\rm Haar}}[U^{\otimes t}\ \cdot \ {U^\dagger}^{\otimes t}]||_\diamond\leq O(t/\sqrt{d})\ .
\end{equation}

We claim that the PFC ensemble, despite being a good $t$-design, fails to scramble some arrangements of Paulis observables. Now we analyze the OTOCs through the empirical spectrum density of $UAU^\dagger B$ with $U\sim\nu_{\rm PFC}$. We let $A$ to be any Pauli string and let $B$ consist of $Z$'s only. We show that, in these instances, their higher-point OTOCs never decay with the order index, and the free mutual information is divergent. 

\fi 

Consider some initial Pauli $A$, and its ``time evolution'' under a single step of the PFC unitary, 
\be 
A' \equiv PFCAC^\dagger F^\dagger P^\dagger.
\ee
We first argue that $A'$ always maps a computational basis state to another computational basis state (up to a phase), i.e., it never maps $\ket{x}$ to a superposition of computational basis states. This is because each factor in $A'=P\cdot F\cdot (CAC^\dagger)\cdot F^\dagger\cdot P^\dagger$ does so. This is obvious for $P$ and $F$ and their Hermitian conjugates. The middle factor $CAC^\dagger$ is some Pauli operator, and a Pauli acting on a computational basis yields a phase flip, a bit flip, or both, so that  again $CAC^\dagger\ket{x}$ is a computational basis state. 

Given this fact, let us denote the permutation action of $A'$ on the computational basis states by $\tau$. Since $A'^2=I$, $A'$ decomposes the computational basis into disjoint orbits that are either 1-cycles (fixed points) $\tau(x)=x$, or 2-cycles of form $\{x,\tau(x)\}$ with $\tau(\tau(x))=x$. On the subspace spanned by the states that belong to the 1-cycles, $A'$ acts as an identity. For a 2-cycle, we have
\begin{equation}
    A'\ket{x}=a_x\ket{\tau(x)},\quad A'\ket{\tau(x)}=a_{\tau(x)}\ket{x},\quad a_xa_{\tau(x)}=1\ .
\end{equation}
Since $B$ only consists of $Z$'s, we have 
\be
B\ket{x}=b_x\ket{x},\,b_x\in\{\pm1\} \, .
\ee
Hence, the restriction of $A'B$ to the span of an arbitrary 2-cycle $\{\ket{x},\ket{\tau(x)}\}$ has the matrix form
\begin{equation}
A'B\Big|_{\rm{span}\{\ket{x},\ket{\tau(x)}\}}=\begin{pmatrix}
    0  &a_xb_{\tau(x)}\\
    a_{\tau(x)}b_x &0
\end{pmatrix}\ .
\end{equation}
Its eigenvalues satisfy 
\be
\lambda_1+\lambda_2=0,\, \quad \lambda_1\lambda_2=-a_{\tau(x)}b_xa_xb_{\tau(x)}=-b_xb_{\tau(x)}\in\{\pm 1\} \, .
\ee
Hence, the only possible eigenvalues of each such $2\times 2$ block, and hence of  of $A'B$, are $\{\pm 1,\pm i\}$. %It immediately follows that the eigenvalues of $A'$ on the full space consisting of various subspaces of 2-cycles could only be $\{\pm 1,\pm i\}$.

We therefore see that even with the addition of the phase flips and permutations, the PFC ensemble does not lead to spreading in phase space and decay of higher-point OTOCs. 
It is nevertheless possible that other variants of the PFC ensemble, such as CPFC~\cite{ma2024construct}, PFCFP, or stacking up multiple layers of PFCs~\cite{chen2024incompressibility} could show behaviors more similar to the generic chaotic systems of Sec.~\ref{sec:physical_systems}.

%\subsection{Integrable spin chains}\label{sec:integrable}
%An integrable system has sufficiently many conserved charges, such that its motion is confined to a submanifold of much smaller dimensionality than that of its phase space. In terms of the operator space spreading quantified by the volume of the joint-moment subspace~\eqref{all_moments}, $\sS_{A(t)|B,\delta,N}$, we expect it to be much smaller than a chaotic counterpart. Hence, the free mutual information should be large.

\subsection{Free fermion integrable systems}\label{sec:tfim}

As a next example, let us consider the transverse-field Ising model (TFIM) on $L$ spins, which corresponds to a special case of \eqref{mixed_field} with $h=0$, i.e., 
\begin{equation}
    H= \sum_{i=1}^L  (Z_iZ_{i+1} +  g X_i)\ . \label{tfim}
\end{equation}
Through a Jordan-Wigner transformation, the Hamiltonian describes $2L$ free Majorana fermions (cf. for instance~\cite{sachdev1999quantum}). 

We take $A_{0}$ and $B_{x}$ to be various choices of initial Pauli matrices in this model. The evolution of $e^{-iHt}Ae^{iHt}$ has some non-generic features compared to chaotic systems, but is not as simple as the Clifford case of the previous section. In particular, both the operator entanglement~\cite{prosen_operator} and the Stabilizer R\'enyi entropy~\cite{leone_tfim} grow under time-evolution in this model, although the growth is slower or the late-time value is smaller than in chaotic systems. Similarly, $\otoc_2$ initially grows in this model, but later oscillates and decays (see for instance~\cite{swingle_xu}). 

The FMI shows a dramatically different behavior in this model compared to the chaotic spin chain for all times.  We can use the mapping of this model to free fermions to see that the real parts of the eigenvalues of $e^{-iHt} A e^{iHt}B$ always take exactly two distinct values, $\lambda_1$ and $\lambda_2$. We explain the details in Appendix~\ref{app:tfim}. We observe numerically that $\lambda_1$ and $\lambda_2$ both increasingly move inwards on the interval $[-1, 1]$ from their initial values of $\pm 1$. Hence, while the precise form of the spectrum of $A(t)B$ is somewhat different from the Clifford case, where the the real parts of the  eigenvalues are simply $\pm 1$ or 0, the consequence for $I_{\rm free}(A(t):B)$ is the same: it remains infinite for all times.  

We show the behavior of $\otoc_n$ and the OTOC partial sums of \eqref{424} for this model in Fig.~\ref{fig:fermion}. While the oscillations of the $\otoc_n$ are somewhat greater in magnitude than those of the chaotic case of Fig.~\ref{fig:chaotic_spin}, it is not entirely obvious from observing the individual $\otoc_n$ in each model whether there is a sharp qualitative difference in the dynamics. On the other hand, comparing the non-monotonic behaviour of the OTOC partial sums in Fig.~\ref{fig:fermion}  as a function of time to the monotonic behavior in Fig.~\ref{fig:chaotic_spin} provides a much sharper diagnostic of the difference between the two models.

%To see this divergence from the OTOC sum formula, we show $\otoc_n$ \SV{and their partial sums} for this model in Fig.~\ref{fig:fermion}. The large oscillations of $\otoc_2$ in this case were previously noted in \cite{swingle_tutorial}  \SV{[recall other refs]}; the oscillations of $\otoc_n$ for higher $n$ are even more dramatic, leading to a divergence of the OTOC sum formula, consistent with the divergence of the Coulomb gas formula. 

\begin{figure}
\centering\includegraphics[height=6cm]{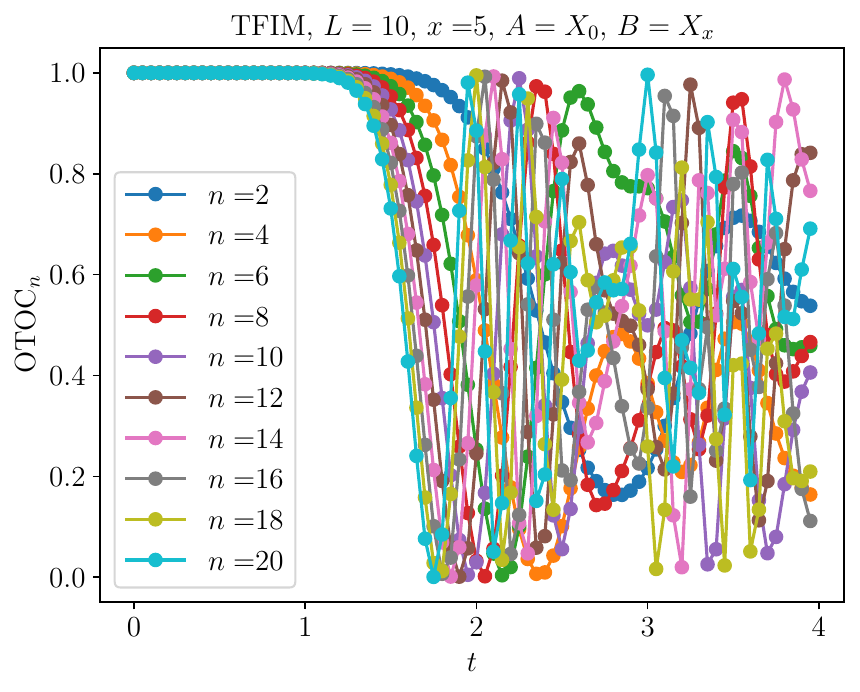}~~~~~~~\includegraphics[height=6cm]{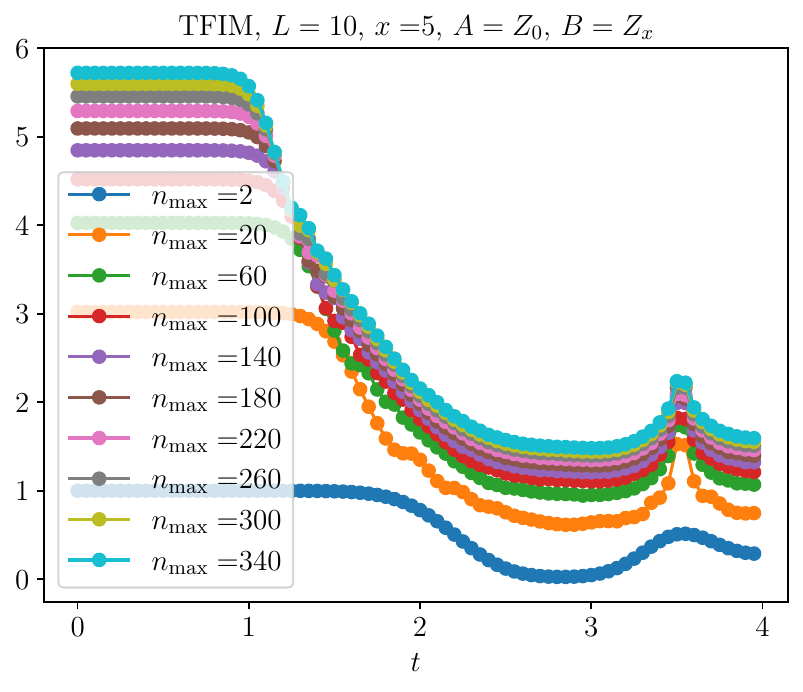}
\caption{We show $\otoc_n$ for the TFIM on the left, and the partial sums of \eqref{424} from this data on the right.} 
    \label{fig:fermion}
\end{figure}

\subsection{Interacting integrable systems}\label{sec:heisenberg}

\begin{figure}[!h]
    \centering
    \begin{subfigure}{7cm}
        \centering
    \includegraphics[width=\linewidth]{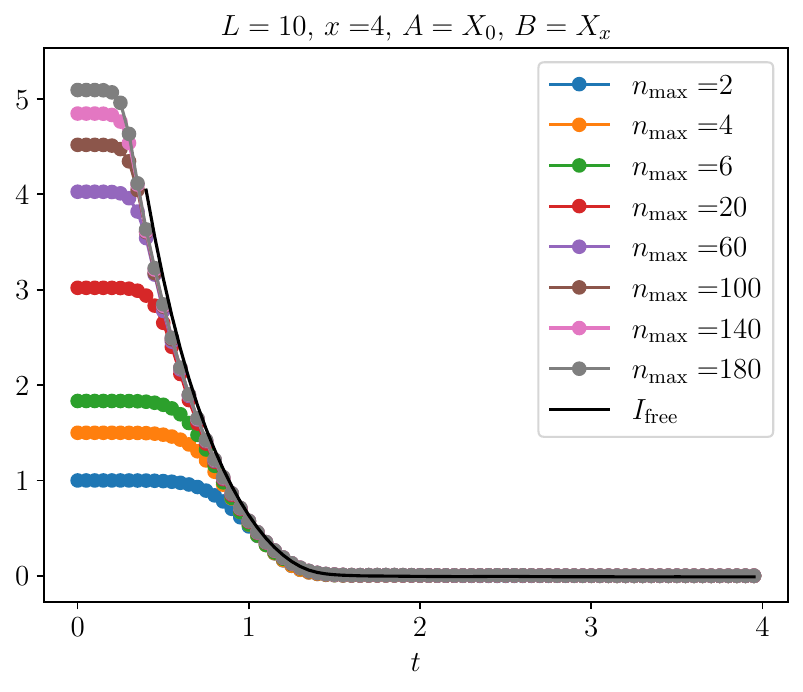}
    \end{subfigure}
      \begin{subfigure}{7cm}
    \centering
    \includegraphics[width=\linewidth]{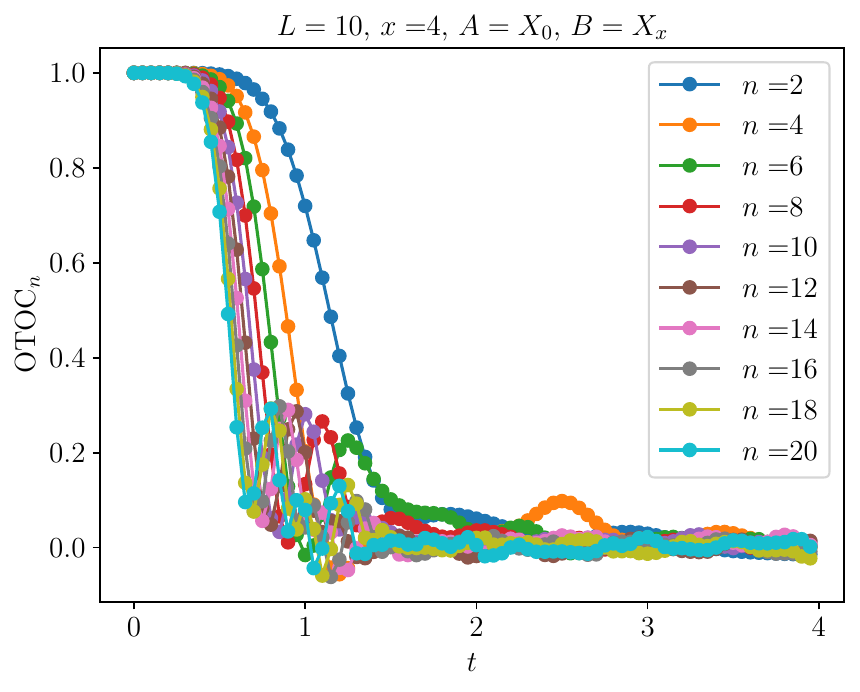}
    \end{subfigure}
    \begin{subfigure}{7cm}
        \centering
    \includegraphics[width=\linewidth]{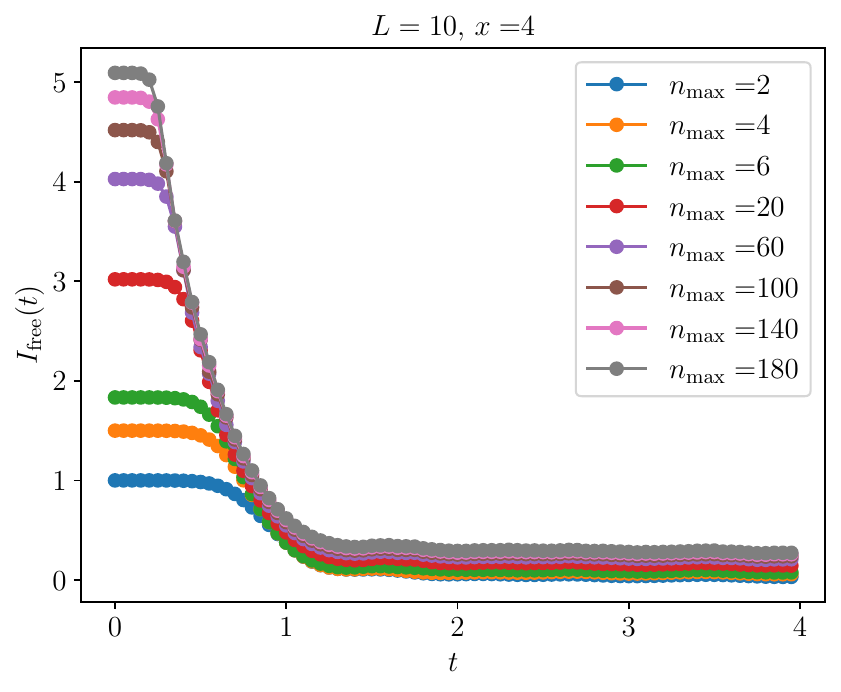}
    \end{subfigure}
    \begin{subfigure}{7cm}
    \centering
    \includegraphics[width=\linewidth]{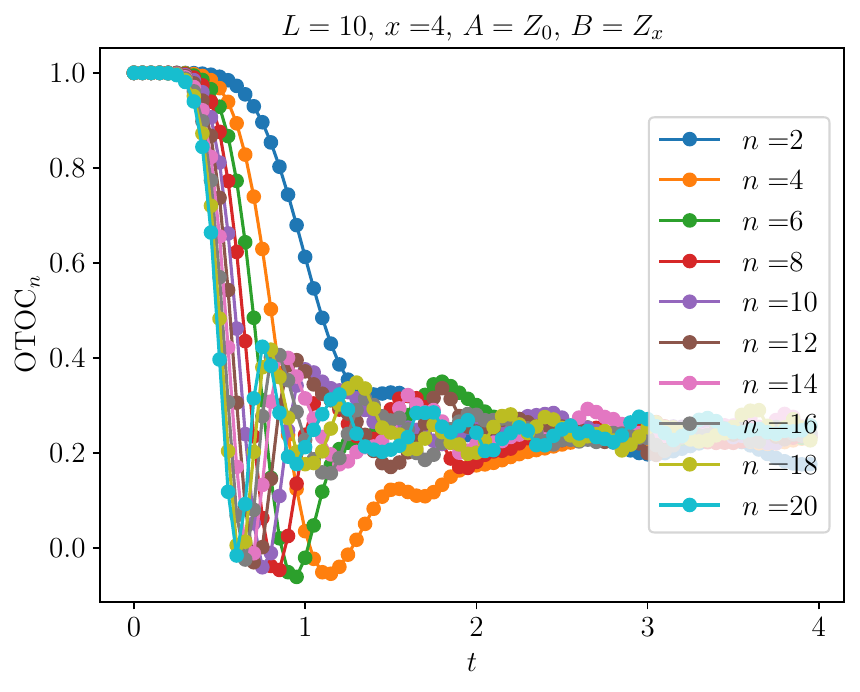}
    \end{subfigure} 
    \caption{Behaviour of $\otoc_n$, OTOC partial sums, free mutual information, and $\Delta_{\rm min}$ in the Heisenberg XXZ model. The top panel shows the case $A=X_0, B=X_x$, and the bottom panel shows the case $A=Z_0, B=Z_x$.}
    \label{fig:Heisenberg}
\end{figure}

As a final example, let us consider an interacting integrable system: the Heisenberg XXZ model, which has the Hamiltonian 
\begin{equation}
    H=\sum_{i=1}^{L}\left(X_iX_{i+1} +Y_iY_{i+1}+ \Delta Z_iZ_{i+1}\right)
\end{equation}
We set $\Delta=0.5$, and consider $L=10$ and periodic boundary conditions. 
Note that in addition to being integrable, the model has a $U(1)$ symmetry: the total $Z$ spin $Z_{\rm tot}\equiv\sum_{i=1}^L Z_i$ is conserved. 

In Fig.~\ref{fig:Heisenberg}, we show the behaviour of $\otoc_n$, $I_{\rm free}$, and $\Delta_{\rm min}$ in this model for two cases:  $A=X_0, B = X_x$, which do not commute with the conserved charge, and $A=Z_0, B = Z_x$, which do commute with the conserved charge.
In the former case, we see behavior similar to the chaotic case in Fig.~\ref{fig:chaotic_spin}. In the latter case, we see non-generic behavior, where $I_{\rm free}$ diverges even at late times due to degeneracies in the spectrum, and the partial sums do not converge. This  behavior should be understood as a consequence of the symmetry, rather than of the integrability of the model: we checked that on adding an integrability breaking next-nearest neighbor term $\sum_{i=1}^L Z_i Z_{i+2}$ which still commutes with $Z_{\rm tot}$, we find qualitatively the same behavior in all cases as in Fig.~\ref{fig:Heisenberg}. Hence, the integrable Heisenberg model does not appear to be distinguishable from chaotic models by $I_{\rm free}$. The  indistinguishability of this model from chaotic models using $\otoc_2$ was previously discussed in~\cite{interacting_integrable}. This example does show the need to formulate a different version of the free mutual information in the presence of symmetries, which we comment on further in the discussion.

\section{Discussion and outlook}
\label{sec:conclusions}

In this paper, we identified a new universal property of the Heisenberg evolution of operators in a chaotic quantum many-body system. We found that a certain coarse-grained volume of the time-evolved operator generically spreads out to fill the abstract space of all operators, which is doubly exponential in the number of degrees of freedom.  This notion of ergodicity in operator space goes beyond the more familiar notion of operator growth in the space of degrees of freedom. We introduced a quantitative measure of spreading in operator space called the free mutual information (FMI). We derived an explicit formula for the FMI for a broad class of operators, and showed a precise and general relation between the FMI and a sum over all higher-point out-of-time-ordered correlators ($\OTOC_n$). This relation provides an understanding of the precise role played by higher-point OTOCs in the approach to maximal scrambling and asymptotic freeness, which had previously been lacking. 

We then studied the behaviour of the FMI and the higher-point OTOCs in a variety of representative examples of  chaotic as well as integrable systems. In all examples of chaotic systems, we found that  while higher-point OTOCs  show large oscillations with time, the FMI decays monotonically, showing  that it captures a more physically meaningful measure of chaos than individual $\otoc_n$. Although the FMI starts to decay at a time scale close to that associated with decay of $\otoc_2$ in {\it chaotic} systems, we confirmed by looking at a variety of integrable models that it is a much finer probe of chaos than $\otoc_2$. In each of these integrable models, we found cases where $\otoc_2$ decays, indicating spreading in physical space, while $I_{\rm free}$ remains infinite, indicating localization in operator space.

While our study of various examples of chaotic quantum many-body systems has allowed us to extract a number of general patterns, various  results in these models  should be better understood in future work. In particular, it is important to address the challenge of performing the OTOC sum in \eqref{425} or directly evaluating the Coulomb gas formula analytically in some concrete physical example, to obtain an analytic form of $I_{\rm free}$ as a function of time. While we proposed a toy model for this functional form in local chaotic models in \eqref{toy_model}, the true result should be of the form
\be 
I_{\rm free}(t) = -c\log (1-f(t))  
\ee
where $f(t)$ is a more non-trivial function which is equal to 1 for $t<t^{\ast}$ and decays for $t>t^{\ast}$. We leave a study of its precise functional form to more detailed studies in future work. It is also important to better understand the general physical reasons why $\otoc_n$ show faster initial  decay for higher $n$, and why later oscillations in $\otoc_n(t)$ do not lead to corresponding oscillations in $I_{\rm free}(t)$. Further, it would be interesting to find an example of a system which one ordinarily thinks of as chaotic from the perspective of $\otoc_2$, which turns out not to be chaotic when probed through $I_{\rm free}$, or where $I_{\rm free}$ starts to decay at a later time scale than $\otoc_2$. 

We comment on a number of broader generalizations and future directions below: 

\vspace{-0.2cm}

\begin{itemize}
    \item {\bf FMI with more reference observables. } In this paper, we have discussed the simplest way to coarse-grain $A(t)$ with a single reference observable $B$. In principle, we could generalize the definition to a collection of multiple reference observables $\{B_i\}_{i=1}^k$, which would give a more refined FMI $I_{\rm free}(A(t):B_1\ldots B_k)$. It is clear that adding more reference observables imposes more constraints, and hence leads to a smaller volume fraction of the coarse-grained set of $\tilde A$ and a larger mutual information: $I_{\rm free}(A(t):B_1\ldots B_k)\ge I_{\rm free}(A(t):B_1)$. This is the monotonicity of free mutual information, which can also be seen as the strong subadditivity of free entropy. It would be  interesting to study when this  inequality is saturated in physical contexts, and to understand the relation of $I_{\rm free}(A(t):B_1\ldots B_k)$ to more general higher-point OTOCs. 
    \item {\bf FMI at finite temperature.} In this paper, we have always considered the case where the reference state $\rho$ in \eqref{all_moments} is the maximally mixed state. In continuum systems such as quantum field theories, infinite temperature correlation functions are not well-defined, and it is necessary to consider a finite temperature state, i.e. take $\rho$ to be $e^{-\beta H}$ for some $\beta>0$. This choice of $\rho$ is also more natural in the presence of energy constraints; indeed, the most interesting behaviour of the OTOC in models like SYK is seen at low temperature~\cite{kitaev2014fundamental, Maldacena_2016_syk}, and interesting physical constraints on OTOCs such as the chaos bound apply only for finite temperature~\cite{Maldacena_2016}.

    Mathematically, the generalization to finite temperature  poses the challenge that we are no longer dealing with tracial moments, so that our techniques in Sec.~\ref{sec:explicit_formula} need a revision. Nevertheless, we believe that it may still be possible to find an explicit formula in this case. If we take $A$ and $B$ to be traceless involutions (Paulis or rotated Paulis), the set of all joint moments in this case can be organized as $m_n\equiv d^{-1}\Tr\rho [A(t)B]^n$, $m_n^*\equiv d^{-1}\Tr\rho [BA(t)]^n$, $l_n\equiv d^{-1}\Tr\rho B[A(t)B]^n$ and $l_n^*\equiv d^{-1}\Tr\rho A(t)[BA(t)]^n$. %These four sets can be organised into two pairs of conjugate moments $(m_n,m_n^*)$ and $(l_n,l_n^*)$.
    The resulting free mutual information might not admit a simple Coulomb gas formula that depends on a single eigenvalue density $\rho(x)$ determined by one set of moments, but should somehow combine two densities determined by $(m_n,m_n^*)$ and $(l_n,l_n^*)$ respectively.

\rev{\item {\bf Symmetry-constrained free mutual information.}
    Our example in Sec.~\ref{sec:heisenberg} shows explicitly that symmetries can enforce degeneracies in the spectrum of $A(t)B$ and hence obstruct the unconstrained Coulomb-gas formula. The FMI definition can be adapted to the symmetry class of a given physical model by replacing the full orbit $\sS_A$ and its Haar-induced volume in \eqref{sadef}--\eqref{volumedef} with the orbit and volume measure on the symmetry-constrained submanifold. For example, when a charge $Q$ is conserved, one may restrict to unitaries commuting with $Q$, or equivalently formulate the volume ratio sector by sector within the corresponding charge blocks. The resulting symmetry-constrained volume ratio compares $A(t)$ only with operators accessible within the same physical symmetry class. We leave a systematic study of these model-dependent constructions to future work.}  
    
    \item {\bf Better understanding the OTOC sum  formula.} There are a few subtle features of our main formula~\eqref{otoc_intro} that intrigue us. If we were to guess an explicit formula for the FMI, which is implicitly defined in terms of the OTOCs, a sum-of-squares form is the most natural candidate given that the FMI is nonnegative. It is striking that the actual answer turns out to be precisely a sum of squares up to the harmonic weight. As we have remarked, the harmonic weight is subtle because it allows any power law decay of $\otoc_n$ with $n$ to sum to a finite FMI, while not allowing $\otoc_n$ that do not decay with $n$ to sum to a finite FMI. It would be interesting to better understand the physical interpretation of this weight. We expect that  this harmonic weight is likely the underlying reason why $I_{\rm free}$ decays monotonically despite the oscillations at intermediate times in individual $\otoc_n(t)$ in the spin chain and random GUE models.
\iffalse 
    \SV{[}In principle, some other weight could have appeared that was not right at the edge of convergence.\footnote{ Different ways to organize OTOCs into a cumulative quantity were recently studied in~\cite{jahnke2025free}. } Does it mean that we could actually find unitary dynamics whose higher-point OTOCs are also very important in determining the FMI? \SV{] I don't quite understand the precise point here. Do you mean that if we had a weight of say $1/n^2$, then the higher-point $\OTOC_n$ would automatically be irrelevant even if they did not decay with $n$, but in this case we actually need them to decay with $n$?} \JW{(I don't think any weight smaller than $1/n$ would make sense a priori, as we know FMI is divergent initially and all OTOCs=1.) I'm thinking what if the weights have a larger power, say $n^{-\frac12}$, $n^0$ or even $n^1$, then a finite FMI imposes more stringent decay with $n$. $1/n$ leaves some room for potential mild decay behaviours. I'd be curious if one can come up with such a model with a mild decay with $n$.}
  \fi 
  
    The harmonic weight is likely due to the fact that the FMI is an entropic quantity whose definition comes with a logarithm, and its Taylor expansion is exactly the RHS of~\eqref{otoc_intro}. Explicitly, we could rewrite the formula as\footnote{We thank Zhenbin Yang for pointing this out to us. Also, cf.~\cite{Trunin_2023,Trunin_2023_2} for a similar (but not identical) quantity called the ``logarithmic OTOC.''}
    \begin{equation}
        I_{\rm free}(A(t):B) = \sum_{n=1}^{\infty} \frac{2}{n} \otoc_n(A(t):B)^2  =-2\Tr\log [I-(A(t)B)\otimes (A(t)B)]
    \end{equation}
    where we could also replace $\Tr\log$ by $\log\det$. We do not yet have a good physical interpretation of the operator $I-(A(t)B)\otimes (A(t)B)$, which could shed further light on the interpretation of the FMI.

    \item {\bf Relation to other measures of chaos and complexity.} %Our approach of using operator-spreading to probe quantum chaos is conceptually reminiscent of
    Recall that we started with the motivation of coming up with a quantum many-body analog of ergodicity  in classical phase space. In classical chaos, the natural entropy associated with this ergodicity is  
    the Kolmogorov-Sinai (KS) entropy, which measures the rate at which one needs to refine the phase space resolution to keep track of the  trajectory of the system. A quantum version of KS entropy has previously been defined in~\cite{connes1987dynamical,alicki1994defining}. This entropy quantifies how the information gained from a measurement on the system evolves over time, and is not obviously related to our discussion in this paper of the quantum trajectory of an initial operator in operator space. However, in the mathematical literature there is some discussion of the relation between the quantum KS entropy and the free entropy~\cite{voiculescu1994alternative,voiculescu1995dynamical}. It is worth studying the  potential connection between these quantities in the physical context of quantum chaos. 
    
    %\SV{[}The physical Hilbert space is exponentially large in system size, whereas the operator space is doubly exponentially large. Hence, free mutual information and free entropy are intuitively more akin to complexities than entanglement entropies. In the same vein as how the free mutual information is defined via counting, complexity is often estimated via counting arguments~\cite{Brown_2018,oszmaniec2021epsilon,Brand_o_2021,Haferkamp_2022,chen2024incompressibility}.\SV{]}  \SV{
    Another natural question is about how the free mutual information may be related to various notions of the complexity of the time-evolved operator $A(t)=U(t)^{\dagger} AU(t)$, or the complexity of the time-evolution operator $U(t)$. One relation to complexity comes from the fact that the  free entropy captures a version of quantum Kolmogorov (descriptive) complexity of  $A(t)$~\cite{jw_talk}.  There may also be a relation or bound between some variant of the FMI and circuit complexity, as   estimates of circuit complexity often make use of the volume occupied by an ensemble of time-evolution operators such as a random unitary circuit in the space of all unitaries (see for instance~\cite{harrow, hunter_jones, jeongwan}). The FMI studied in this paper, which involves two operators, saturates at much earlier time scales than the ones associated with the circuit complexity, but it is possible that a less coarse-grained, multi-operator version could saturate at much later time scales. 
    
    %\JW{A natural candidate for comparison to the circuit complexity of the time-evolution operator $U(t)$ could be the quantity $ I_{\rm free}(\mathcal{A}(t):\mathcal{A})$ over the algebra $\mathcal A$ generated by all traceless involution operators, where the FMI is the generalized one between $\mathcal{A}$ at time $t=0$ and its time evolutions. Then the approximate freeness, indicated by an exponentially small FMI, could only be realized via a typical Haar random unitary $U(t)$ that takes exponential time to reach, just like its circuit complexity. }
    
    \item {\bf Gravity, holographic theories, and the SYK model.} 
    Out of time-ordered correlators were first used as diagnostics of chaos and scrambling in the recent literature in models such as holographic CFTs and the SYK model at low temperature~\cite{firewalls, butterfly, stringy, kitaev2014fundamental, Maldacena_2016_syk}. Recently, higher-point OTOCs have also been studied in these models~\cite{haehl1, Chandrasekaran_2023,penington2025}. The  time-dependence of the type of higher-point OTOCs that we are interested in here, $\Tr[\rho (A(t)B)^n]$, appears not to have been studied in the existing literature on holographic theories.
    In particular, while~\cite{Chandrasekaran_2023, penington2025} show that all higher-point OTOCs in the $t\to \infty$ limit are consistent with asymptotic freeness, we would like to understand the time-scale on which the decay of $\otoc_n$ first starts in these models in order to understand whether and when the FMI converges. \cite{haehl1} does consider the detailed time-dependence in $AdS_2$ and SYK, but for a  different set of ``maximally braided'' higher-point OTOCs. It would be interesting to check whether the time scale on which $I_{\rm free}$ starts to converge is identical to the scrambling time in these large $N$ models.  One case where the formulas we have  already derived  in this paper can be directly applied is the SYK model at infinite temperature.

    \item {\bf Free mutual information in mathematics.} The definition of free mutual information we have used is adapted from the orbital free entropy~\cite{HIAI_2009} in the mathematics literature. As mentioned in Footnote.~\ref{ft:FMI}, there is an alternative, and perhaps more standard, definition of free mutual information that mathematicians use, and it is based on the free Fisher information under a free liberation process~\cite{biane1997free} instead of on the counting of matrices~\cite{voiculescu1999analogues,Voiculescu_2002}. These two definitions are expected to coincide~\cite{Voiculescu_2002,biane2003large,Collins_2014}, but it is not known in full generality whether they coincide. In the companion paper~\cite{vw2025}, we will show that for the case of two traceless involutions, these two definitions indeed coincide, and they have a similar formula to the ensemble FMI $\sI$ in Sec.~\ref{sec:ensemble}. We shall prove that both formulas~\eqref{eq:integralFMI} and~\eqref{425} for $\mathcal I(a:b)$ are valid for any two traceless involutions $a$ and $b$ living in any tracial non-commutative probability space,  complementing the finite-dimensional setting we study in this paper.  In short, we have been using a consistent adaptation of free mutual information to finite dimensions. 
    
    \item {\bf Free R\'enyi mutual information.} The partial sum expression in~\eqref{424}, truncated at $n=k$, $\sum_{n=1}^k\frac2n\OTOC_n^2$, can be thought of as defining a ``R\'enyi free mutual information'' of order $k$.\footnote{\rev{The ordinary R\'enyi entropy of order $k$ is defined from the $k$th moment $\Tr\rho^k$; correspondingly, the order-$k$ R\'enyi FMI retains the joint-moment information through order $k$. The terminology is thus intended as a natural analogy, not as an assertion that this is the conventional R\'enyi mutual information of a density matrix.}} These R\'enyi's are practically useful, as in practical settings such as experiments only a finite number of OTOCs may be  available.  It is obvious from the sum-of-squares formula that they are ordered with respect to the R\'enyi index, similar to the  classical and quantum R\'enyi entropies (except in a reverse order). A principled approach would be to define R\'enyi free mutual information by modifying the definition~\eqref{all_moments} so that it only depends on the first $k$ moments. In the companion paper~\cite{vw2025}, we shall prove that this proper definition exactly matches the expected formula $\sum_{n=1}^k\frac2n\OTOC_n^2$.

\end{itemize}

\noindent \textbf{Acknowledgments.}
We would like to thank Douglas Stanford for suggesting the method in Section~\ref{sec:otocs} among many other valuable comments. We also thank Roland Farrell, Jonah Kudler-Flam, Jeongwan Haah, Patrick Hayden, Nima Lashkari, John Preskill, Tommy Schuster, and Zhenbin Yang for helpful discussions. JW acknowledges the support form DOE Q-NEXT and GeoFlow. SV acknowledges funding provided by the DOE QuantISED program (DE SC0018407) and the Institute for Quantum Information and Matter, an NSF Physics Frontiers Center (NSF Grant PHY- 2317110). 

\appendix

\section{Details on derivation of  Coulomb gas formula}

\subsection{Proof of Lemma~\ref{lem:lemma_0}}
\label{app:prob}

\lemmazero*

\begin{proof}

Recall that $\vec{y}(U A U^{\dagger}B)$ is defined by 
\be 
y_i = \cos \phi_i, \quad e^{\pm i \phi_i} \text{ are the eigenvalues of } UAU^{\dagger} B   \label{alpha_i_def}
\ee
$A$ and $B$ both have $d/2$ eigenvalues equal to $+1$ and $d/2$ eigenvalues $-1$. In a basis where $B$ is diagonal, we can write 
\be 
B = \mathbf{1} - 2 P , \quad UAU^{\dagger} = \mathbf{1} - 2Q, \quad Q =  U V P V^{\dagger} U^{\dagger} \label{pqdef}
\ee
where $P$ is the rank-$d/2$ projector $\text{diag}(1, ...,1, 0, ...,0)$, and $V$ is a unitary such that $V B V^{\dagger}=A$. 
Now for any $U$, out of the various choices of basis that diagonalize $B$ or $P$, there exists at least one choice of basis where we can write   
\begin{equation}\label{22blocks}
    P=\bigoplus_{i=1}^{d/2}\begin{pmatrix}
       1 & 0 \\
0 & 0 
    \end{pmatrix},\quad Q=\bigoplus_{i=1}^{d/2}\begin{pmatrix}
      \cos^2\theta_i & \cos\theta_i\sin\theta_i \\
\cos\theta_i\sin\theta_i & \sin^2\theta_i
    \end{pmatrix},
\end{equation}
where $\{\theta_i\}_{i=1}^{d/2}\in[0,\pi/2]$ are the principle angles between the subspaces where $P$ and $Q$ are supported. 
From this block structure, the eigenvalues of $PQP$ are given by
\begin{equation}
  x_i=\cos^2\theta_i \, , \quad i = 1, ..., d/2 \, . 
\end{equation}
The discussion so far holds for any fixed $U$. Now if $U$ is a Haar-random unitary, the probability density for the \rev{labelled} eigenvalues  $x_i$ is given by~\cite{collins2005product}
\begin{equation}
   p(x_1,\ldots,x_{d/2})\prod_{i=1}^{d/2}\dd{x_i} =\sM \prod_{1\leq i<j\leq N/2}(x_i-x_j)^2 \prod_{i=1}^{N/2}\mathbf{1}_{[0,1]}(x_i)\dd x_i \label{qdef}
\end{equation}
for some normalization constant $\sM$. For any set $S$, the indicator function $\mathbf{1}_S(x)$ is equal to 1 if $x\in S$ and zero otherwise.    

We can now use \eqref{pqdef} and \eqref{qdef} to deduce the probability density of the $y_i$ in \eqref{alpha_i_def}. From \eqref{22blocks}, the eigenvalues of $AB$ are 
\begin{equation}
    z_i= e^{\pm i2\theta_i}, \quad i = 1, ..., d/2
\end{equation}
so we simply need to make the identification $\phi_i= 2 \theta_i$. We can first change variables in \eqref{qdef} from $x_i=\cos^2\theta_i$ to $\theta_i$ (we will ignore the overall normalization constant in these intermediate steps and find it at the end):
\begin{equation}
   p(x_1,\ldots,x_{d/2})\prod_{i=1}^{d/2}\dd{x_i} \propto\prod_{1\leq i<j\leq d/2}(\cos 2\theta_i-\cos 2\theta_j)^2 \prod_{i=1}^{d/2}\mathbf{1}_{[0,\pi/2]}(\theta_i)\sin2\theta_i\dd \theta_i
\end{equation}
Then since $\phi_i = 2\theta_i$, 
\be
p(x_1,\ldots,x_{d/2})\prod_{i=1}^{d/2}\dd{x_i} \propto\prod_{1\leq i<j\leq d/2}(\cos \phi_i-\cos \phi_j)^2 \prod_{i=1}^{d/2}\mathbf{1}_{[0,\pi]}(\phi_i)\sin\phi_i\, \dd \phi_i 
\ee
and finally changing variables to $y_i = \cos \phi_i$, we have 
\be
p(x_1,\ldots,x_{d/2})\prod_{i=1}^{d/2}\dd{x_i} \propto \prod_{1\leq i<j\leq d/2}(y_i - y_j)^2 \prod_{i=1}^{d/2}\mathbf{1}_{[-1,1]}(y_i)\,\dd y_i\, . 
\ee
\rev{This gives a permutation-invariant density on the full labelled cube $[-1,1]^{d/2}$, with normalization
\be
\sN_{\rm lab}=2^{-\frac{d^2}{4}}\prod_{k=0}^{\frac{d}{2}-1}\frac{(\tfrac{d}{2}+k)!}{(k+1)!\,(k!)^2}.
\ee
Now for the ordered vector, we need to reduce the density onto $\sW_{d/2}$. This gives us an extra permutation factor of $(d/2)!$. Therefore the density of the ordered vector is normalized by $\sN=(d/2)!\,\sN_{\rm lab}$, as stated in Lemma~\ref{lem:lemma_0}.}
We note in passing that this spectral distribution is also shared by the Jacobi Unitary Ensemble (JUE$_{d/2,\alpha,\beta}$) in dimension $d/2$ with $\alpha=\beta=0$. 

\end{proof}
\subsection{Proofs of Lemmas \ref{lem:lemma_1} and \ref{lem:lemma_2} }
\label{app:volume_lemmas}

The condition \eqref{pauli_condition}, which defines $\tilde A$ in the set $\sS_{A(t)|B, \delta, N}$, can be expressed in terms of $\vec{y}(\tilde AB)$ and $\vec{x}$ as follows: 
\be 
\bigg|~\overline{T_n(\vec{y}(\tilde AB))} - \overline{T_n(\vec{x})}~\bigg| < \delta \text{ for all } 0 < n \leq N \label{33}
\ee
where 
 $T_n$ is the $n$-th Chebyshev polynomial of the first kind, and for any function $f:[-1, 1]\to \mathbb{R}$, we define
\be 
\overline{f(\vec{\gamma})} \equiv \frac{2}{d}  \sum_{i=1}^{d/2}f(\gamma_i) \, .
\ee

%We would like to show that the  value of  $\bigg|~\overline{T_n(\vec{y})} - \overline{T_n(\vec{x})}~\bigg|$  implies both upper and lower bounds on $||\alpha - \beta||_{1}$ in terms of $\delta$. This will in turn allow us to show both upper and lower bounds on the volume of $\sS_{A(t)|B, \delta, N}$ in terms of a closed formula involving the eigenvalues of $A(t) B$. To derive the relevant bounds, note that:

\L*
\iffalse 
\begin{lemma}
%\label{lemma_1}
Define $\sB_1\subset [-1, 1]^\frac{d}{2}$ as an $\epsilon$-box around $\vec{x}$ of radius 
\be 
\epsilon_1 = \frac{26}{\pi} d\,(- \delta\log\delta) \, .  
\ee
 If $N$ and $\delta$ satisfy the condition  
\be 
N > \frac{9 \pi}{\delta} \, ,   \label{310_a}
\ee
then 
for all $\tilde A \in  \sS_{A(t)|B, \delta, N}$, the corresponding vectors $\vec{y}$ are contained within $\sB_1$. 

\end{lemma}
\fi 

\begin{proof}

Consider the $L_1$ norm distance between the vectors $\vec{y}$ and $\vec{x}$: 
\be 
||\vec{y} - \vec{x}||_{1} \equiv \sum_{i=1}^{d/2} |\alpha_i - \beta_i |\, . 
\ee
From~\cite{sublinear} (Lemma 3.1), we have the following upper bound on the $L_1$ norm in terms of the moments $\bar T_n$:  for any positive integer $m$,  
\be 
||\vec{y}-\vec{x}||_{1} \leq \frac{d}{2} \le(\frac{36}{m} + \frac{4}{\pi}\sum_{k=1}^m \frac{|\overline{T_k(\vec{y})}-\overline{T_k(\vec{x})}|}{k}  \ri)\, . 
\ee
For $\bar {T_n}$ satisfying  \eqref{33}, this inequality implies that  
\be 
||\vec{y}-\vec{x}||_{1} \leq  \frac{d}{2} \le( \frac{36}{m} + \frac{4}{\pi}(\log m +1)\delta  \ri) \label{diff1}
\ee
For a given $\delta$, the RHS is minimized at $m$ closest to $9\pi/\delta$. If $N$ and $\delta$ satisfy \eqref{310}, the minimum value can be achieved and we have 
\be 
||\vec{y}- \vec{x}||_{1} \leq c \, d (-\delta \log \delta)  \label{in1}
\ee
for any $O(1)$ constant $c \geq \frac{26}{\pi}$. 
\end{proof}

\secondlemma*
\iffalse
\begin{lemma}
%\label{lemma_2}
Define $\sB_2\subset[-1,1]^\frac{d}{2}$ as an $\epsilon$-box around $\vec{x}$ of radius 
\be
\epsilon_2=\delta/N^2\, . 
\ee
$\sB_2$ is contained within the set of $\vec{y}(\tilde AB)$ corresponding to $\tilde A$ in $\sS_{A(t)|B, \delta, N}$. 
\end{lemma}
\fi 

\begin{proof}
From the Kantorovich–Rubinstein inequality, for any $\kappa$-Lipschitz function $f:[-1, 1] \to \mathbb{R}$, 
\be 
\bigg|\overline{f(\vec{y})}- \overline{f(\vec{x})}\bigg| \leq \frac{\kappa}{d}||\vec{y} -\vec{x}||_1 \ .
\ee
The Lipschitz constant of $T_n(x)$ on $[-1,1]$ is $n^2$.~\footnote{To see this, recall that $\kappa$ for any $f$ is ${\rm max}_{[-1,1]}|f'(x)|$, and note that $T'_n(x) = n U_{n-1}(x)$, where $U_n(x)$ is the $n$-th Chebyshev polynomial of the second kind. Since $\text{max}_{[-1,1]}|U_{n-1}(x)|= n$, $\kappa=n^2$.}
Hence, 
\be 
|\overline{T_n(\vec{y})}-\overline{T_n(\vec{x})}| \leq \frac{n^2}{d} ||\vec{y}-\vec{x}||_1 \ .  \label{b11}
\ee
Now for any $\alpha \in \sB_2$, we have  
\be 
||\vec{y}-\vec{x}||_1  < d\delta/N^2 
\label{nmax_condition}
\ee
which from \eqref{b11} implies 
\be |\overline{T_n(\vec{y})}-\overline{T_n(\vec{x})}|< \delta \text{ for all } n\leq N\, . 
\ee
Hence, $\sB_2$ is contained within the set of $\vec{y}(\tilde AB)$ corresponding to $\tilde A$ in $\sS_{A(t)|B, \delta, N}$. 

\end{proof}

%\SV{Note that from the earlier constraint \eqref{310}, 
%\be 
%\epsilon_2 \leq \frac{\delta^3}{(9\pi)^2} \, . 
%\ee
%} 

 %\SV{[Note that if we make this restriction $n< \sqrt{d}$ throughout, then \eqref{in1} still holds, as long as we satisfy the additional condition $\delta > \frac{9\pi}{\sqrt{d}}$. But then we would have the issue that the radius of $\sB_1$ is not small, which is a problem for the later argument around \eqref{319}... It would help if we could improve the $n^2$ to $n$ in \eqref{in2}.]}\JW{I don't see a way to improve the either bound, but now I don't think we need to impose $n < \sqrt{d}$.}

\subsection{Derivation of \eqref{eq:vol_constant}}
In this appendix, we simplify the normalization factor $\sN$ to obtain \eqref{eq:vol_constant}. 

\label{app:Stirling}

\begin{align}
\log \sN + \frac{d^2}{4}\log 2 &=\rev{\log\left(\frac d2\right)!+} \sum_{k=0}^{d/2-1} \log\frac{(d/2+k)!}{(k+1)!(k!)^2} \nn 
&=\rev{\log\left(\frac d2\right)!+} \sum_{k=0}^{d/2-1}  [\log\,(d/2+k)! -\log\,(k+1)!-2\log k!] \label{a20}
\end{align} 
We use Stirling's approximation to approximate each term.
\begin{equation}
    \log n!=n\log n - n + O(\log n)\ .
\end{equation}
Let us start with the last term of \eqref{a20}.
\begin{equation}
    \sum_{k=0}^{d/2-1}2\log k! =\sum_{k=0}^{d/2-1}2k\log k - (d/2-1)d/2 + O(d\log d) 
\end{equation}
The first term $\sum_{k=0}^{d/2-1} k \log k$ is the logarithm of Hyperfactorial $H(d/2-1)\equiv\prod_{n=0}^{d/2-1}n^n$. Using Glaisher's approximation,
\begin{equation}
    H(n) = An^{(6n^2+6n+1)/12}e^{-n^2/4}(1+O(n^{-2}))\ ,
\end{equation}
where $A\approx 1.28243$ is the Glaisher–Kinkelin constant, we have
\begin{equation}
    \sum_{k=0}^{d/2-1}2k\log k= 2\log H(d/2-1) = \frac{d^2}4\log d-\frac{d^2}{4}\log 2 - \frac{d^2}{8} + O(d\log d)\ .
\end{equation}
\begin{equation}
    \sum_{0}^{d/2-1}2\log k! =\frac{d^2}4\log d - \frac{3d^2}{8}  + O(d\log d) 
\end{equation}
Similarly, 

\begin{equation}
    \sum_{k=0}^{d/2-1}\log (k+1)! = \sum_{k=1}^{d/2}\log k! =\frac{d^2}8\log d-\frac{d^2}8\log 2 - \frac{3d^2}{16}  + O(d\log d)  
\end{equation}

Next, consider the first term of \eqref{a20}:
\begin{equation}
     \sum_{k=0}^{d/2-1}  \log\,(d/2+k)! = \sum_{k=0}^{d/2-1}(d/2+k)\log (d/2+k) - (d-1+d/2)d/4 + O(d\log d) \label{a27}
\end{equation}
We can rewrite the second term of \eqref{a27} as
\begin{equation}
    \sum_{k=0}^{d/2-1}(d/2+k)\log (d/2+k) = \sum_{k=0}^{d-1}k\log k - \sum_{k=0}^{d/2-1}k\log k=\log H(d-1) - \log H(d/2-1) 
\end{equation}
Using Glaisher's approximation, we obtain
\begin{equation}
    \sum_{k=0}^{d/2-1}(d/2+k)\log (d/2+k) = \frac{3d^2}8\log d +\frac{d^2}{8}\log 2- \frac{3d^2}{16}+ O(d\log d) 
\end{equation}
and 
\begin{equation}
     \sum_{k=0}^{d/2-1}  \log\,(d/2+k)! = \frac{3d^2}8\log d - \frac{9d^2}{16} + O(d\log d)
\end{equation}
Putting all contributions together, we obtain
\begin{align}
        \sum_{k=0}^{d/2-1} \log\frac{(d/2+k)!}{(k+1)!(k!)^2}  &= \frac{3d^2}8\log d +\frac{d^2}{8}\log 2- \frac{9d^2}{16} -\left(\frac{d^2}8\log d-\frac{d^2}8\log 2 - \frac{3d^2}{16}\right)\nn
        &-\left(\frac{d^2}4\log d -\frac{d^2}{4}\log 2- \frac{3d^2}{8}\right) + O(d\log d) \nn 
        &= \frac{d^2}2\log 2+ O(d\log d) \ .
\end{align}
\rev{Similarly, the first term  of \eqref{a20} is $\log((d/2)!)=O(d\log d)$, so that overall we find}
\be 
\log \sN = \frac{d^2}{4}\log 2\, + O(d\log d) . 
\ee

\begin{figure}[t]
    \centering
    \begin{overpic}[width=0.45\linewidth]{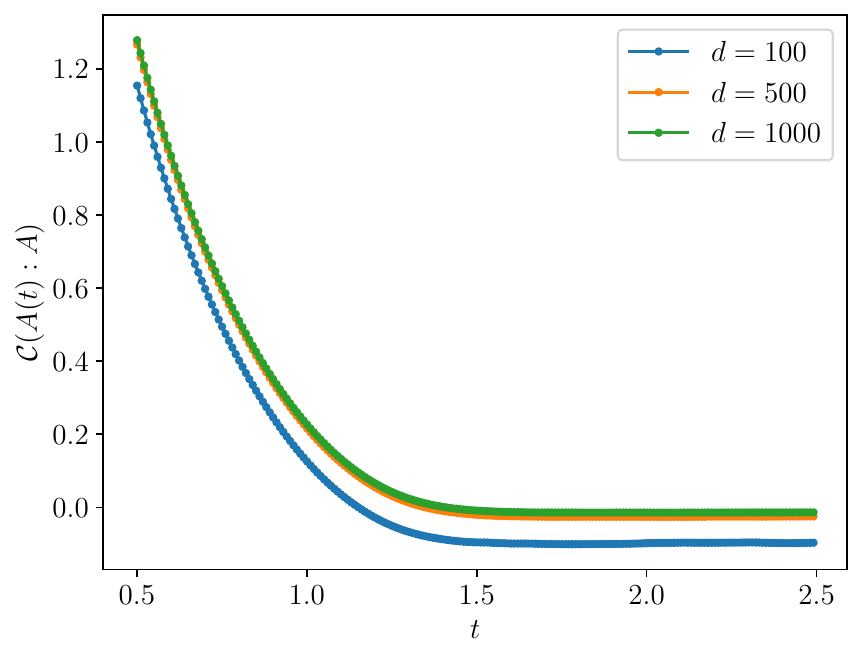}
        \put(7,80){\footnotesize(a)}
    \end{overpic}%
    \begin{overpic}[width=0.45\linewidth]{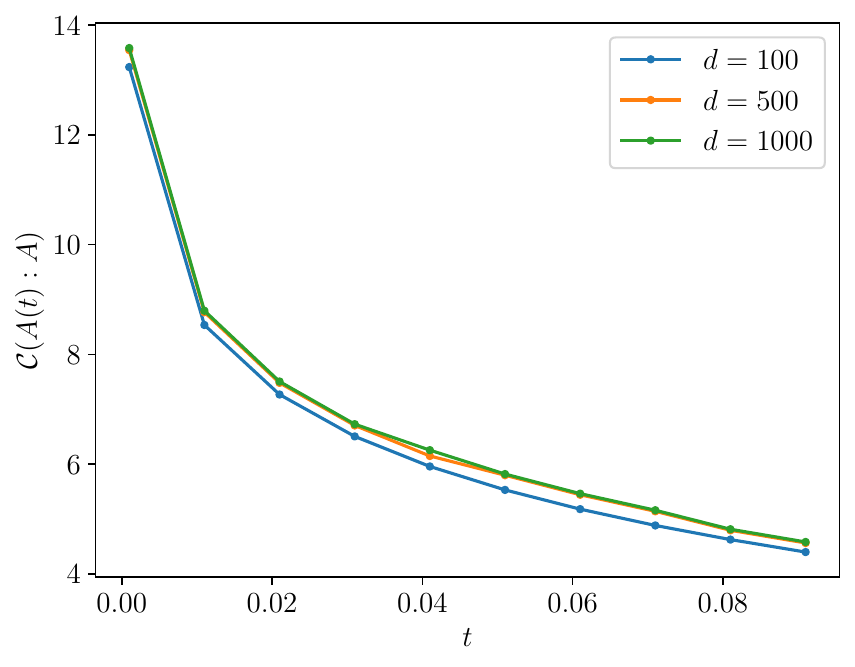}
        \put(7,80){\footnotesize(b)}
    \end{overpic}

    \begin{overpic}[width=0.45\linewidth]{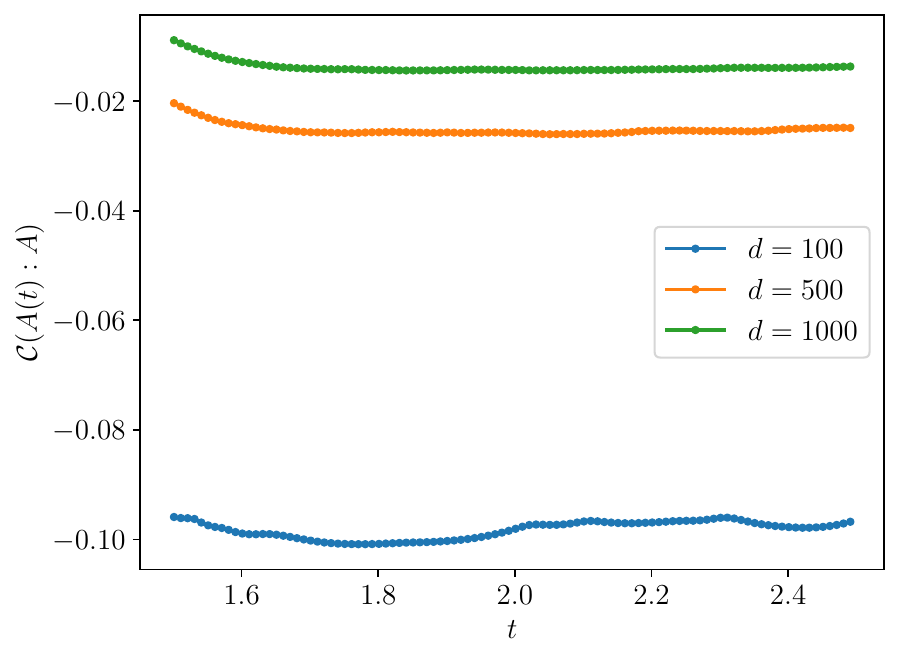}
        \put(7,70){\footnotesize(c)}
    \end{overpic}%
    \begin{overpic}[width=0.45\linewidth]{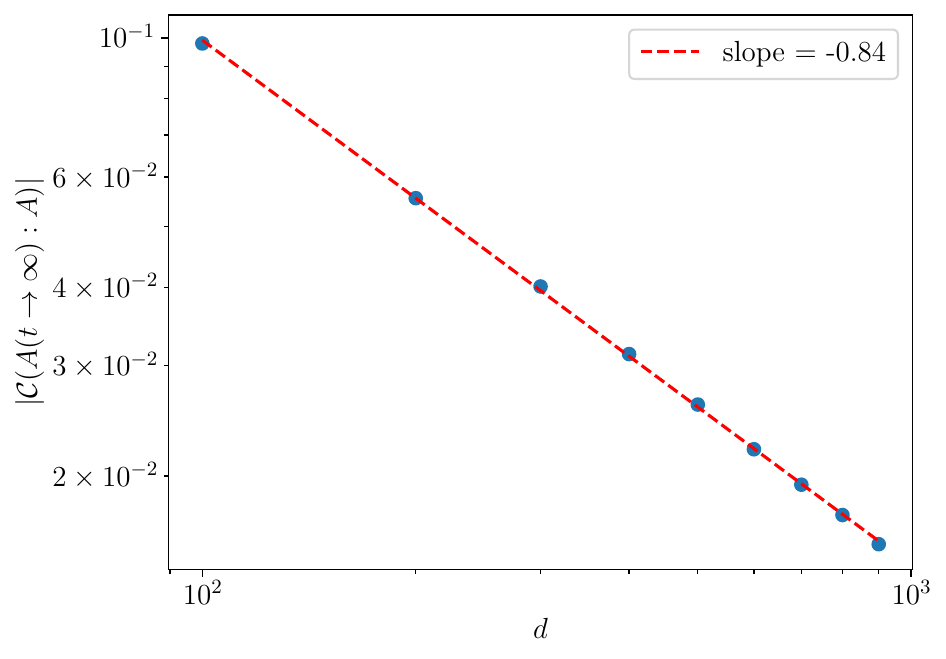}
        \put(15,70){\footnotesize (d)}
    \end{overpic}

    \caption{Dependence of the  Coulomb gas formula $\sC(A(t):B)$ on the Hilbert space dimension in the random GUE model. In (a), (b) and (c), we zoom in respectively on intermediate, early, and late times. In the late-time regime where the Coulomb gas formula is negative, its absolute value shows a $d^{-0.84}$ dependence. We show this in (d) with a log-log plot of the absolute value of the saturation value of $\sC$ for 10 values of $d$ from 100 to 1000. In this regime, we formally set $I_{\rm free}$ to zero according to point 2 at the end of Sec.~\ref{sec:formula}, with the understanding that the actual value is positive and proportional to some power of $1/d$.}
    \label{fig:d_dependence}
\end{figure}

\section{$d$-dependence of $\Delta_{\rm min}$ and $I_{\rm free}$ in physical examples}\label{app:gue_checks}

In this Appendix, we justify the assumption about $\Delta_{\rm min}$   used in the derivation of the Coulomb gas formula, and clarify the regime of validity of this formula, in explicit example of a chaotic time-evolution. 

Let us start with the GUE Hamiltonian model introduced in Sec.~\ref{sec:gue}.  
First note that  the Hamiltonian is normalized such that irrespective of the Hilbert space dimension $d$, the Coulomb gas formula decays from infinity to a small value on the same $O(1)$ time scale: this can be seen in  Fig.~\ref{fig:d_dependence} (a). Hence, it makes sense to compare $\Delta_{\rm min}$ for different $d$ over the same range of times. 

To check the assumption in \eqref{327}, we zoom in on early times of order $10^{-2}$ for $d=100, 500$, and $1000$. 
We show the case $A\neq B$ in Fig.~\ref{fig:deltamin_ab}; this case $A = B$ is similar. Note that for each $d$, we consider a single realization of the  random GUE ensemble at that dimension. From the orders of magnitude for $\Delta_{\rm min}$ that we observe in these plots, it seems likely that the dependence of $\Delta_{\rm min}$ on $d$ is consistent with the assumption \eqref{327}. Note that it does not quite make sense to directly plot the dependence of $\Delta_{\rm min}$ for a given time as a function of $d$ in this case, as we are considering single realizations of the GUE matrices, and $\Delta_{\rm min}$ is not a self-averaging quantity. Note also from the early-time regime in Fig.~\ref{fig:d_dependence} (b) that the Coulomb gas formula is independent of $d$ for large enough $d$ at these early times. %This confirms that $\Delta_{\rm min}$ does not decay so rapidly with $d$ that it dominates in the Coulomb gas formula at such times (although this is a weaker statement than the assumption~\eqref{327}).  

We further check the order of magnitude of $\Delta_{\rm min}$ in the mixed-field Ising model and the Heisenberg model in Fig.~\ref{fig:deltamin_spin}, and find a reasonably large magnitude for most of the range of times we are interested in. 

As a separate point, let us comment on the late-time regime of Fig.~\ref{fig:d_dependence} (c) and (d). We see that the Coulomb gas formula has a negative saturation value at such times, whose magnitude decays as $\sim d^{-0.84}$. Since $I_{\rm free}$ is necessarily non-negative, the Coulomb gas formula cannot be equal to $I_{\rm free}$ in these cases. This is consistent with the point discussed at the end of Sec.~\ref{sec:formula} that when the value of the Coulomb gas formula is proportional to some power of $1/d$, there are competing contributions that no longer allow us to identify it with $I_{\rm free}$. 

\iffalse
\begin{figure}[!h]
    \centering
    \includegraphics[width=0.45\linewidth]{compare_all_new.pdf}  \includegraphics[width=0.45\linewidth]{compare_all_early_new.pdf} \includegraphics[width=0.45\linewidth]{compare_all_late_new.pdf} 
    \includegraphics[width=0.45\linewidth]{c_d_dependence.pdf} 
    \caption{Dimension-dependence of the time-evolution of the Coulomb gas formula for $I_{\rm free}$ in the random GUE model. In the three plots, we zoom in respectively on intermediate, early, and late times. We label the $y$-axis $I_{\rm free}$, but it is clear from the last plot that in the late-time regime where the Coulomb gas formula is negative, it does not correspond to $I_{\rm free}$. In this regime, we formally set $I_{\rm free}$ to zero, with the understanding that the actual value is positive and proportional to some power of $1/d$, but finding the precise power is beyond the scope of our analysis.}
    \label{fig:d_dependence}
\end{figure}
\fi 

\begin{figure}[!h]
    \centering
  \includegraphics[width=0.45\linewidth]{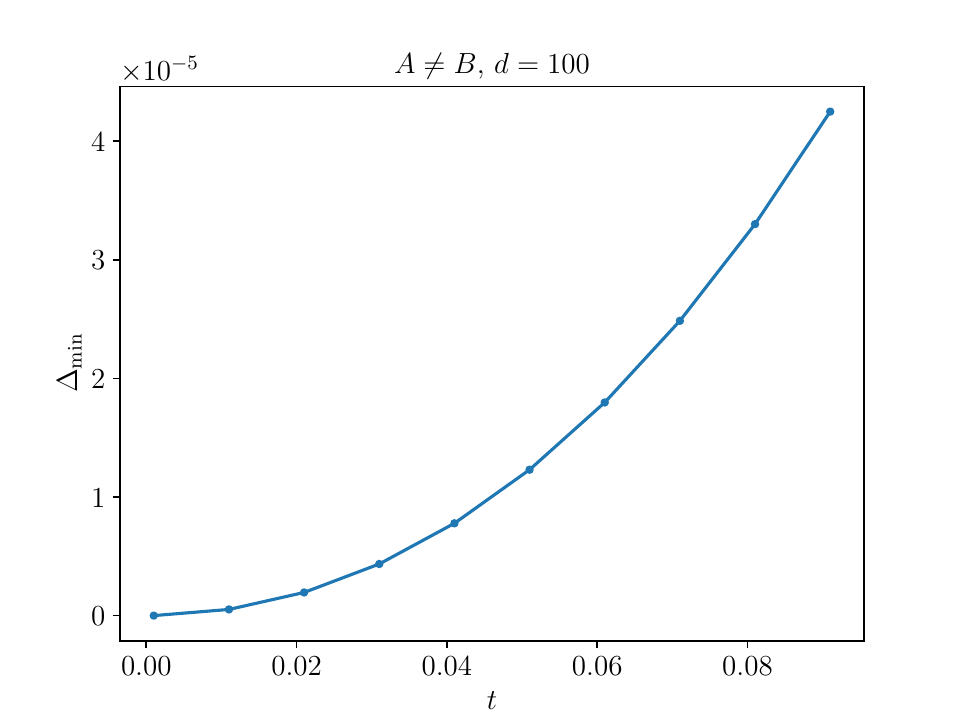} \includegraphics[width=0.45\linewidth]{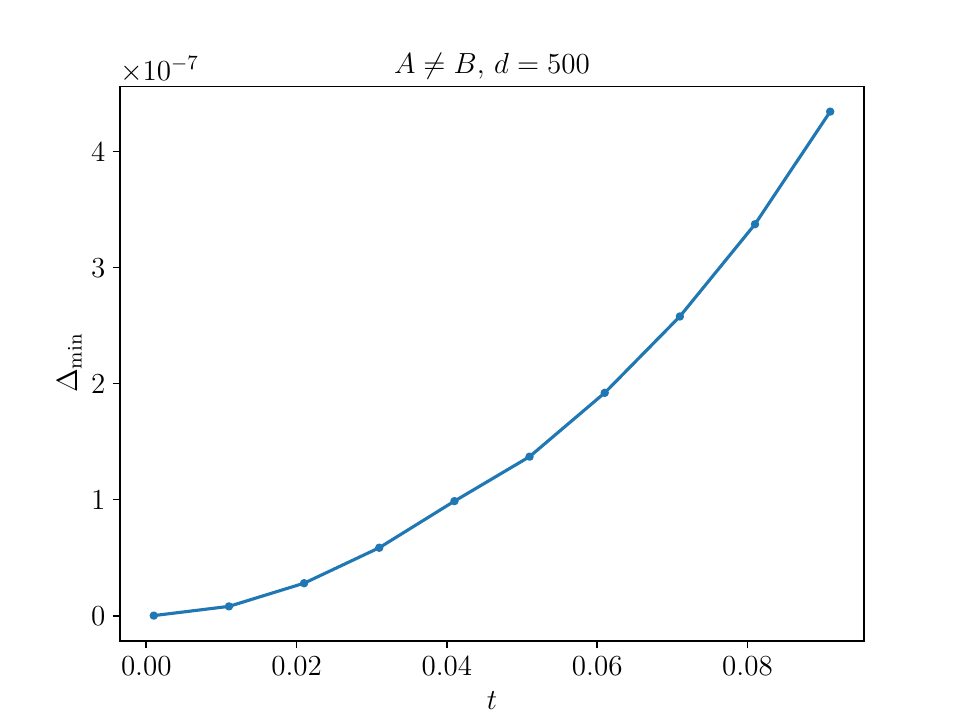} \includegraphics[width=0.45\linewidth]{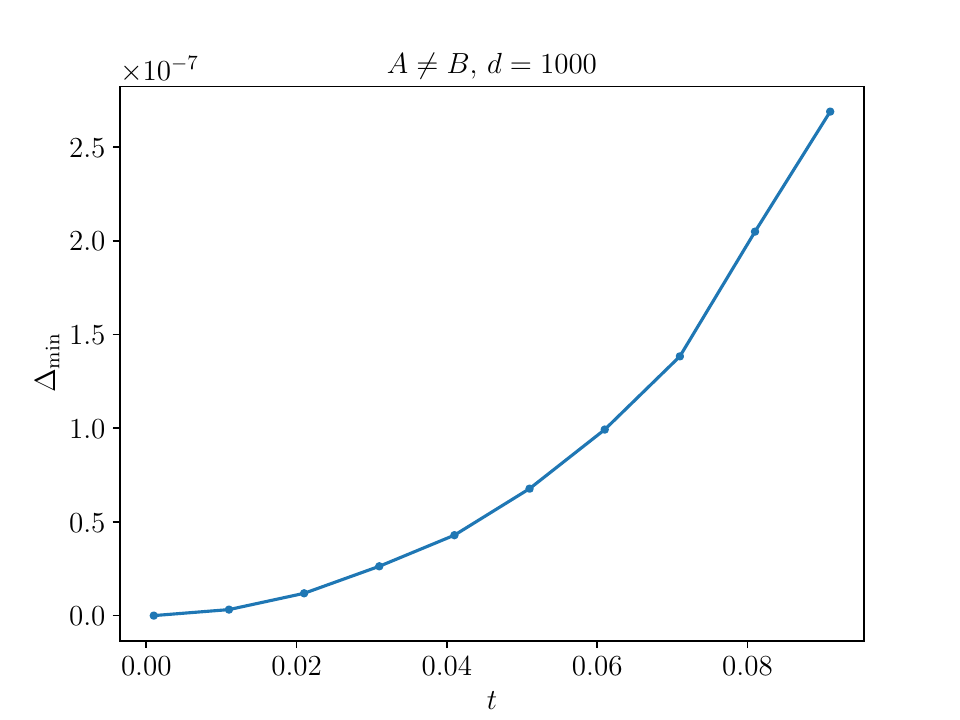}
    \caption{Dimension and time dependence of $\Delta_{\rm min}$ in the random GUE model for the case $A\neq B$.} 
    \label{fig:deltamin_ab}
\end{figure}

\begin{figure}
    \centering
       \begin{subfigure}{0.49\textwidth}
        \centering
\includegraphics[width=\linewidth]{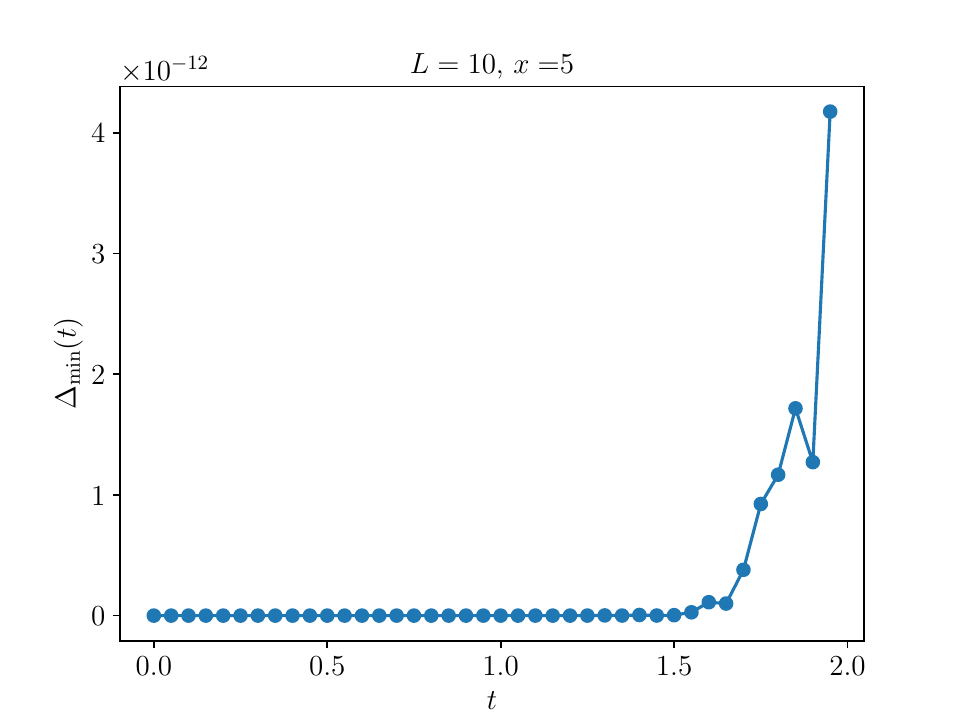}
    \end{subfigure}
    \begin{subfigure}{0.49\textwidth}
\centering
\includegraphics[width=\linewidth]{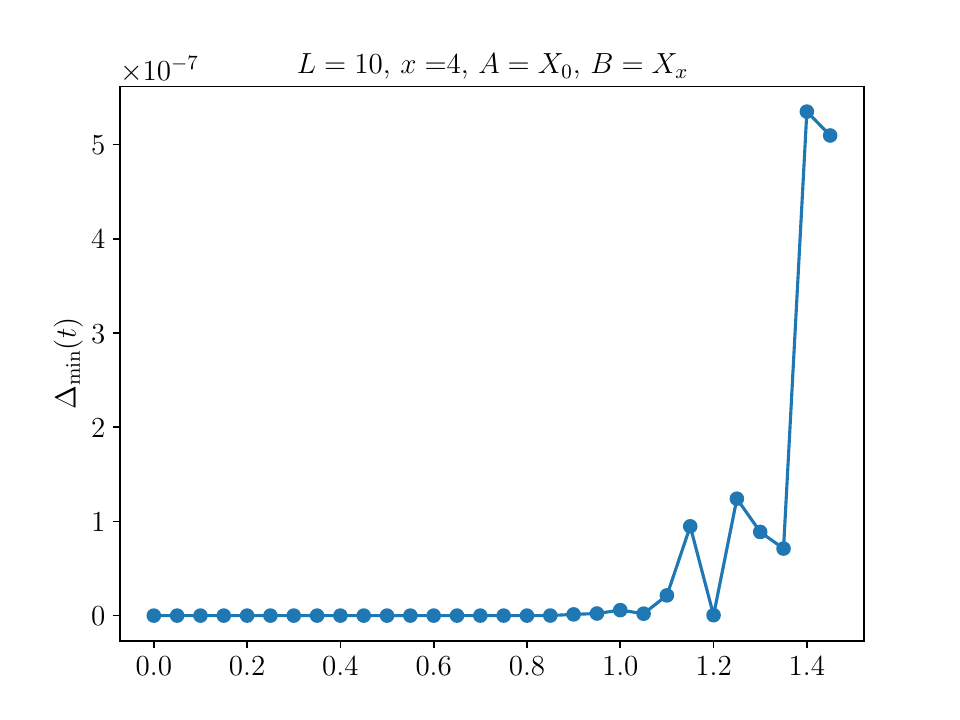}
    \end{subfigure}
    \caption{$\Delta_{\rm min}$ as a function of time for $L=10$ $(d=2^{10})$ in the chaotic spin chain model (left) and in the Heisenberg model for the case $A=X_0, B=X_x$ (right).}
    \label{fig:deltamin_spin}
\end{figure}

\section{No delta functions in $\bar \rho(x)$ for smooth ensembles}
\label{sec:smoothness}

In this Appendix, we show that under a sufficient smoothness condition for the ensemble $\nu$, the ensemble-averaged spectral density $\bar{\mu}(z)$ of $U(t)AU(t)^{\dagger}B$ for $U(t)\sim \nu$ has the following property:   
\begin{equation}\label{eq:no_atoms}
   \oint_{\mathbb{T}}\dd z'\mu(z')\mathbf{1}_{z'=z}=0, \quad\forall z\in\mathbb T\ .
\end{equation}
where $\lambda_i(UAU^\dagger B)$ denotes any eigenvalue of $UAU^\dagger B$. Recall that $\mu(z)$ was defined in \eqref{48}. In the terms of the averaged spectral density $\bar \rho(z)$ of the real part of $z$, \eqref{eq:no_atoms} implies \eqref{eq:no_atoms_mt}.  

 The smoothness condition we need is that the measure $\nu$ is \emph{absolutely continuous} with respect to the Haar measure, i.e., it has a delta function-free probability density $p(U)\equiv\dd\nu(U)/\dd U$ with respect to the Haar measure $\dd U$ on $\mathbf{U}(d)$.

To see that this smoothness condition implies \eqref{eq:no_atoms}, note first that we can interpret the LHS of \eqref{eq:no_atoms} as a probability under $U\sim \nu$ as follows: 
\be 
\oint_{\mathbb{T}}\dd z'\mu(z')\mathbf{1}_{z'=z} =  \mathrm{Pr}_{U\sim\nu}[\lambda_i(UAU^\dagger B)=z] %= \frac{\int_{\mathbf{U}(d)} dU \frac{d\nu(U)}{dU} }{ \int_{\mathbf{U}(d)} dU \frac{d\nu(U)}{dU}}
\ee
where $\lambda_i(UAU^\dagger B)$ is any eigenvalue of $UAU^\dagger B$. We then have 
\be 
\oint_{\mathbb{T}}\dd z'\mu(z')\mathbf{1}_{z'=z} =  \mathrm{Pr}_{U\sim\nu}[F(U, z)=0]  
\ee
where $F(U, z)$
is  the characteristic polynomial 
\begin{equation}
   F(U, z)\equiv\det(UAU^\dagger B-zI)\ .
\end{equation}
As we will explain below, at any $z$, $F(U)\equiv F(U,z)$ is a real-analytic function of $U$.~\footnote{A function $f:\mathbf U(d)\to\mathbb C\simeq\mathbb R^2$ being real-analytic means that for any $U\in \mathbf U(d)$ there exists a chart $\phi:\mathcal U\to\mathbb R^{d^2}$ where $\mathcal U$ is a neighbourhood of $U$, such that $f\circ\phi^{-1}:\phi(\mathcal U)\subset\mathbb R^{d^2}\to\mathbb R^2$ has a convergent Taylor series expansion in the Euclidean variables in $\mathcal U$.} According to a general theorem from~\cite{Mityagin_2020,krantz2002primer}, any non-trivial (non-zero) real-analytic function on $\mathbf{U}(d)$ can only vanish on a set of measure zero for a measure that is absolutely continuous with respect to the Haar measure on $\mathbf{U}(d)$. The statement \eqref{eq:no_atoms} then follows. 

 To see that $F(U)$ is real-analytic (considering some fixed $z$ and suppressing the $z$ argument), we first extend the domain to all complex matrices $X\in M_d(\mathbb C)$, and consider $\tilde F(X)=\det(XAX^\dagger B-zI)$. We can define a chart $\phi$ that embeds $X$ into $\mathbb R^{2d^2}$ via matrix entries $U_{ij}=x_{ij}+iy_{ij}$. Under the chart $\phi$, it is clear that $F\circ\phi^{-1}(\{x_{ij},y_{ij}\}_{i,j})$ is a polynomial in the matrix entries so $\tilde F(X)$ is real-analytic. Then we can use the fact that $F(U)=\tilde F(X)|_{\mathbf U(d)}$ is also real-analytic on the restricted domain of a smooth real-analytic submanifold $\mathbf U(d)$~\cite{krantz2002primer}.\footnote{For our purposes, we can define a real-analytic submanifold of $M_d(\mathbb C)$ as given by the matrices satisfying real-analytic constraint equations $G(X)=0$. For $\mathbf U(d)$, we have the constraints $G_1(X)=XX^\dagger-I$ and $G_2(X)=X^\dagger X-I$, which are polynomial functions in matrix components and hence real-analytic. Hence, $\mathbf U(d)$ is a smooth real-analytic submanifold.}

We now check the absolute continuity (delta function-freeness) condition of $\nu$ for two models studied in this paper: the GUE Hamiltonian and random brickwork quantum circuits.
\begin{itemize}
    \item \emph{GUE Hamiltonian.} The GUE ensemble has a continuous Gaussian density on the space of Hermitian matrices. The exponentiation map is analytic. Pushing forward a smooth density through an analytic map always yields a continuous density with respect to the Haar measure. Hence, the spectral density of the exponentiated GUE unitary ensemble contains no delta functions. It is evident that $F(U)$ is a non-trivial (non-zero) function at any $z$ as long as $t>0$, so the general argument above applies.
   \item \emph{Random quantum circuits.} Let the unitary evolution be $U_t(G)$, where $G$ is the shorthand for the collection of gates $\{G_i\}_{i=1}^M$ used to build the random circuit. The probability distribution of the gates $\{G_i\}_{i=1}^M$ is a product of Haar measures on SU$(q^2)^{\times M}$, and hence has a continuous delta-free probability density on the product group manifold. 
   
   Before the light-cone of $A(t)$ touches $B$, the function $F(U_t)$ is trivial (zero) when $z=\pm 1$, so the argument above does not apply and the spectral density clearly has two delta peaks $\delta_{\pm 1}$. 
   
   After the light-cone of $A(t)$ touches $B$, the function becomes non-trivial at any $z$. To see this, we simply need to find an instance of a RQC that is non-trivial. Let all gates be instantiated as some Clifford unitaries. The operator $A(t)$ after a Clifford evolution remains a Pauli string, so we have either $A(t)$ and $B$ commute or anti-commute with each other, and the spectrum of $A(t)B$ is either $\{\pm 1\}$ or $\{\pm i\}$. If $z=\pm 1$, we can choose an instance of Clifford gates such that $A(t)$ and $B$ anti-commute. If $z\neq \pm 1$, we can choose an instance of Clifford gates such that $A(t)$ and $B$ commute. $F(G)$ is thus a non-trivial function at any $z$. Then our general argument above applies.
\end{itemize}

%\footnote{Formally, we need to view $F(G)$ as a real-analytic function on the gate parameters, embedded in $\mathbb R^{2q^4M}$. The Haar measure is absolutely continuous with respect to the Lebesgue measure in any real‐coordinate embedding.}. Hence, for any $z\in\mathbb T$,
%Following the argument in Sec.~\ref{sec:smoothness}, it remains to show that $F(G):=\det(U_t(G)AU_t(G)^\dagger B-zI)$ is not a zero function for any $z$. 

\section{Details of derivation of the OTOC sum formula}
\label{app:yterm}

In this Appendix, we explain how the $\sY$ term in \eqref{46} is expanded in terms of the OTOCs to derive \eqref{full_otoc}.

Due to the Sokhotski–Plemelj formula, the function $D: \mathbb{T} \to \mathbb{C}$ defined by 
\be 
D(z-\alpha) \equiv  \frac{1}{2\pi i} \le[\lim_{\zeta \to z^+} \frac{1}{\zeta- \alpha} - \lim_{\zeta \to z^-} \frac{1}{\zeta- \alpha}\ri] 
\ee
behaves as an analog of the delta function $\delta(z-\alpha)$ for complex numbers $z$ integrated over the unit circle, i.e. for any function $f: \mathbb{T} \to \mathbb{C}$, 
\be 
\oint_\mathbb{T} \dd z\,  D(z-\alpha)\,  f(z) = f(\alpha)\, . 
\ee

We can therefore express $\sY$ as follows: 
\begin{align} 
\sY &= \frac2{d^2} \sum_{i=1}^{d} \log |{\rm Re}\, z_i - {\rm Re} \, z_i| \nn 
&= \frac2{d^2}  \oint_{\mathbb{T}}\dd z \oint_{\mathbb{T}}\dd \tilde z      \log | {\rm Re} \,z - {\rm Re} \, \tilde z| \le( \sum_{i=1}^{d}  D(z - z_i) D(\tilde z - z_i) \ri)\ . \label{232}
\end{align}
Let us now express  $\sum_{i=1}^{d} D(z-z_i) D(\tilde z- z_i)$ in terms of traces involving $A(t)B$:
\begin{align}
&\sum_{i=1}^{d} D(z-z_i) D(\tilde z- z_i) = \frac{1}{(2\pi i)^2} \sum_{i=1}^d \le(\frac{1}{z^+-z_i}-  \frac{1}{z^--z_i}\ri) \le(\frac{1}{{\tilde z}^+-z_i}-  \frac{1}{{\tilde z}^--z_i}\ri) \nn 
& = \frac{1}{(2\pi i)^2}  \bigg( \Tr\le[\frac{1}{(z^+- AB)({\tilde z}^+- AB)}\ri] +  \Tr\le[\frac{1}{(z^-- AB)({\tilde z}^-- AB)}\ri] \nn 
& \quad - \Tr\le[\frac{1}{(z^+- AB)({\tilde z}^-- AB)}\ri] -\Tr\le[\frac{1}{(z^-- AB)({\tilde z}^+- AB)}\ri]  \bigg)\ . \label{233}
\end{align}
Here for instance $z^+$ is shorthand for the limit of a complex variable $\zeta$ approaching $z \in \mathbb{T}$ from outside $\mathbb{T}$. 
Now depending on whether we take the limit from outside or inside $\mathbb{T}$, we can expand each of the above terms in terms of OTOCs: 
\begin{align}
&\frac{1}{d}\Tr\le[\frac{1}{(z^+- AB)({\tilde z}^+- AB)}\ri] = \frac{1}{z^+ {\tilde z}^+} \sum_{m, n=0}^{\infty} \frac{1}{({z^+})^m {({\tilde z}^+})^n} {\rm OTOC}_{m+n}\, . \\
& \frac{1}{d}\Tr\le[\frac{1}{(z^-- AB)({\tilde z}^-- AB)}\ri] = \sum_{m,n=0}^{\infty} ({z^-})^n({{\tilde z}^-})^m \OTOC_{m+n+2}\, .\\
& \frac{1}{d}\Tr\le[\frac{1}{(z^+- AB)({\tilde z}^-- AB)}\ri] = - \sum_{n, m=0}^{\infty}\frac{({\tilde z}^-)^n}{( z^+)^{m+1}} \OTOC_{m-n-1} \, .\\
&\frac{1}{d}\Tr\le[\frac{1}{(z^-- AB)({\tilde z}^+- AB)}\ri] = - \sum_{m,n=0}^{\infty} ({\tilde z}^+)^{-n-1} ({z}^{-})^m \OTOC_{n-m-1}\, .
\end{align}
Now putting these expressions back into \eqref{232}, we find that $\sY$ can be written as a sum of four terms $\sY =\sum_i \sY_i$ corresponding to the four terms in \eqref{233} (we set $z = e^{i \theta}$, $\tilde z = e^{i\phi}$, and now ignore the $\pm$ superscripts in the above expressions): 
\be
\begin{aligned}
\sY_1 &=  \frac{1}{(2\pi)^2} \frac2d \sum_{m,n=0}^{\infty} \OTOC_{m+n} \int_0^{2\pi}\dd \theta \int_0^{2\pi}\dd \phi\, e^{-im \theta} e^{-in\phi} \log|\cos\theta - \cos \phi|  \\
& = \frac{1}{(2\pi)^2} \frac2d\sum_{m,n=0}^{\infty} \OTOC_{m+n} \int_0^{2\pi}\dd \theta \int_0^{2\pi}\dd \phi\, e^{-im \theta} e^{-in\phi} \le[-\log 2 - \sum_{p=1}^{\infty} \frac{2}{p} \cos(p\theta) \cos(p \phi) \ri] \\
& = -\frac{2\log 2}{d} - \frac2d\sum_{p=1}^{\infty} \frac{1}{2p} \OTOC_{2p}\ .
\end{aligned}
\ee
Similarly, 
\be
\begin{aligned}
\sY_2 &= \frac{1}{(2\pi)^2} \frac2d \sum_{m, n=0}^{\infty} \OTOC_{n+m+2} \int_0^{2\pi}\dd \theta \int_0^{2\pi}\dd \phi\, e^{i(n+1)\theta} e^{i(m+1)\phi} \left(-\log 2 -\sum_{p=1}^{\infty} \frac{2}{p} \cos(p\theta) \cos(p \phi)\right)  \\
& = - \frac2d \sum_{p=1}^{\infty} \frac{1}{2p} \OTOC_{2p} \ .
\end{aligned} 
\ee
\be
\begin{aligned}
\sY_3 &= \frac{1}{(2\pi)^2} \frac2d \sum_{m, n=0}^{\infty} \OTOC_{m-n-1} \int_0^{2\pi}\dd \theta \int_0^{2\pi}\dd \phi \,e^{i(n+1)\phi}e^{-im\theta} \le(-\log 2 - \sum_{p=1}^{\infty} \frac{2}{p}\cos(p\theta)\cos(p\phi) \ri) \\
& = -\frac2d\sum_{p=1}^{\infty} \frac{1}{2p} \ .
\end{aligned}
\ee
and 
\be
\begin{aligned}
\sY_4 &= \frac2d\frac{1}{(2\pi)^2}\sum_{m, n=0}^{\infty} \OTOC_{n-m-1} \int_0^{2\pi}\dd \theta \int_0^{2\pi}\dd \phi\, e^{i (m+1) \theta} e^{-i n \phi} \le(-\log 2 - \sum_{p=1}^{\infty} \frac{2}{p}\cos(p\theta)\cos(p\phi) \ri) \\
& = - \frac2d  \sum_{p=1}^{\infty} \frac{1}{2p} \ .
\end{aligned}
\ee
Combining all terms, we have 
\be 
\sY =  - \frac{2\log 2}{d} - \frac{1}{d}\sum_{n=1}^{\infty} \frac2n \OTOC_{2n}-\frac{1}{d}\sum_{n=1}^{\infty} \frac2n\ \, . 
\ee

\section{Justification of assumptions about the OTOC sum formula in the GUE model}\label{app:gue_checks_otoc} 

In Section~\ref{sec:otocs}, we discussed two different partial sums for $\otoc_n$ in \eqref{full_otoc} and \eqref{nmax}. Let us more carefully compare these two partial sums to each other and to the Coulomb gas formula, using the GUE model of Sec.~\ref{sec:gue} as a concrete example. In Fig.~\ref{fig:ifree_gue} of the main text, we showed the partial sums in \eqref{nmax} over the full range of times and for early times. Here, we show partial sums of both \eqref{nmax} and \eqref{full_otoc} from intermediate to late times. We see that the partial sums of \eqref{full_otoc} are closer to $I_{\rm free}$ at such times, which is expected from the fact that the infinite sum \eqref{full_otoc} is exact. There is a small but visible difference between the converged value of the partial sums of \eqref{nmax} and the Coulomb gas formula at such times, which we expect is $O(1/d^z)$ for some $z>0$. 

\begin{figure}[!h]
    \centering
    \includegraphics[width=0.45\linewidth]{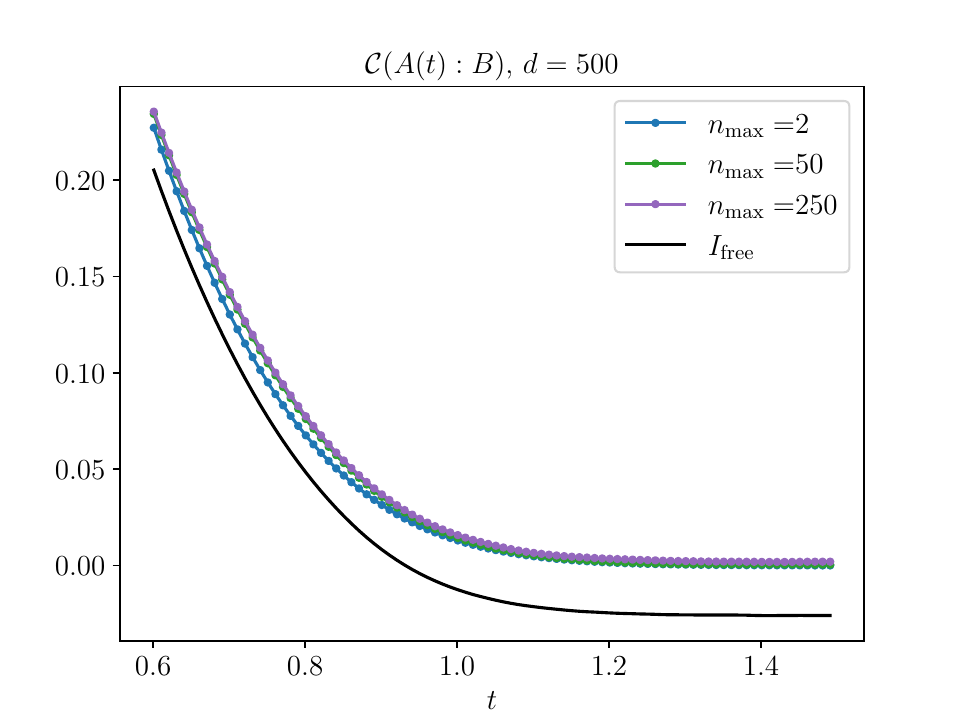}
    \includegraphics[width=0.45\linewidth]{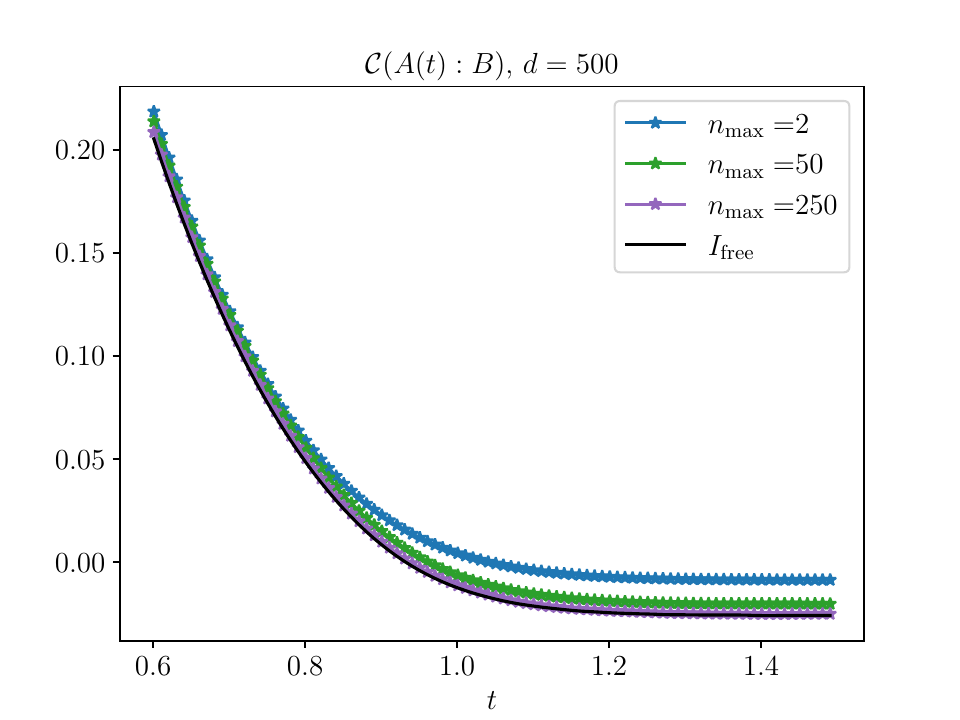}
    
    \caption{We show the partial sums of the approximate series \eqref{nmax} on the left, and those of the exact version \eqref{full_otoc} on the right. In both cases, the black curve is the Coulomb gas formula.}
    \label{fig:otoc_negative}
\end{figure}

\section{Higher-point OTOCs in the GUE model}
\label{app:gue}

In this Appendix, we explain the free probability techniques used in the derivation of the general formula \eqref{otoc_gen}  for $\otoc_n(A(t):A)$ in the random GUE model. 
As discussed in the main text, $A$ and $U(t) = e^{-iH_{\rm GUE(t)}}$ approach two freely independent variables $a$ and $u$ in the large $d$ limit. 
Let us define  
\be 
\tau(\cdot) \equiv \lim_{d\to \infty} \overline{\Tr[\cdot]} \, . 
\ee
where the overline indicates an average over the GUE ensemble. We will refer to $\tau(x^n)$  as the ``moments'' of $x$, and to $\tau(x_1...x_n)$ as a ``joint moment'' of $x_1, ..., x_n$. 

The task of evaluating 
$\lim_{d\to\infty}\otoc_n(A(t):A)$
is therefore equivalent to computing the moments of the polynomial $uau^{\dagger}a$ involving the free random variables $u$ and $a$: 
\be 
\lim_{d\to \infty} \otoc_n(A(t):A) = \lim_{d\to \infty}\frac{1}{d}\Tr[(U(t)AU(t)^{\dagger}A)]  = \tau((u a u^{\dagger} a)^n) \, . 
\ee

Now for {\it any} set of 
noncommuting random variables $\{a_i\}_{i=1}^m$ (which are not necessarily free with respect to each other), the joint moments of $x_1,\ldots,x_n\in \{a_i\}_{i=1}^m$ can be expressed as follows:~\cite{nica2006lectures}  
  \be 
  \tau(x_1\cdots x_n)=\sum_{\pi\in {\rm NC}(n)}\prod_{B\in\pi}\kappa_{|B|}(\{x_i\}_{i\in B})\, . \label{cumudef}
  \ee 
  The quantities $\kappa_n(x_1, ..., x_n)$ appearing on the RHS are called the ``free cumulants'' of $x_1, ..., x_n$. The full set of equations \eqref{cumudef} should together be seen as an implicit {\it definition} of the free cumulants, so at this stage the above equation is not yet useful if our goal is to calculate the moments $\tau$. 

  The equation \eqref{cumudef} becomes a calculation tool in the case where some of the $a_i$ are free with respect to each other, as in this case any {\it mixed cumulants}, defined as $\kappa_n(x_1, ..., x_n)$ where any pair of $x_i$ appearing in the argument are free with respect to each other, vanish. For example, in our calculation of $\tau((uau^{\dagger}a)^n)$ cumulants such as $\kappa_n(u_1, ..., u_n)$ where each $u_i$ is either $u$ or $u^{\dagger}$ are non-zero, but cumulants such as $\kappa_n(u, a, u^{\dagger}, u, ..., u)$ are zero. While this gives a significant simplification, cumulants involving both $u$ and $u^{\dagger}$, which can in principle appear in $\tau((uau^{\dagger}a)^n)$, are still somewhat complicated to evaluate explicitly. A further simplification comes from the following observation: 
  because $\kappa_n(a)=0$ for odd $n$, cf.~\eqref{cumulantsof_a}, the moments of $uau^\dagger a$ turn out to be equal to those of  $uaua$: 
\be 
\tau((u a u^{\dagger} a)^n)=  \tau((u a u a)^n)\, . \label{f4}
\ee
Let us illustrate this in the case $n=1$, where the task is to calculate $\tau(u a u^{\dagger}a)$.
\begin{align}
    \tau(ua&u^{\dagger}a)=\raisebox{-16pt}{\includegraphics[width=0.44\textwidth]{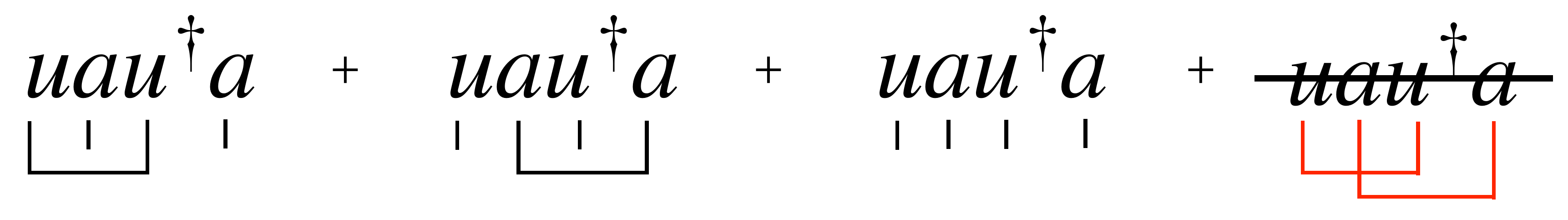}}\ , \label{all}\\ =& \kappa_2(u,u^\dagger)\kappa(a)^2+\kappa(a)^2\kappa(u)\kappa(u^\dagger)+\kappa_2(a,a)\kappa(u)\kappa(u^\dagger).\label{uu}
\end{align}
Note that the last term of \eqref{all} does not contribute because of the restriction to non-crossing partitions in the sum \eqref{cumudef}, and cases where $u$ or $u^{\dagger}$ is contracted with $a$ do not contribute due to the freeness of $u$ and $a$. Now the first two terms in \eqref{uu} vanish as $\kappa(a)=0$, so that we get 
\begin{equation}
\begin{aligned}
    \tau(uau^\dagger a) &= \kappa_2(a,a)\kappa(u)\kappa(u^\dagger)=\eta(t)^2\ ,
\end{aligned}
\end{equation}
which does not involve any mixed cumulants of $u$ and $u^{\dagger}$. It can readily be checked that we get the same result for $\tau(uaua)$. In general, for higher $n>1$, any mixed cumulants of $u$ and $u^\dagger$ are accompanied by some odd cumulant of $a$ that vanishes, so that we have \eqref{f4}. 
This is because all terms with contractions of $u$ with $u^{\dagger}$ also entail a factor of an odd free cumulant of $a$.

Our task has now been simplified to the much simpler one of calculating $\tau((ua)^{2n})$.  We can now use the free multiplicative convolution of $a$ and $u$: 
\begin{equation}
\tau((ua)^{2n}) = \sum_{\pi\in {\rm NC}(2n)} \prod_{V\in\pi}\kappa_{|V|}(a)\prod_{W\in \pi^c}m_{|W|}(u)\ , %\label{otoc_gen}
\end{equation}
where 
\be 
\kappa_{n}(a) \equiv \kappa_n(\, \underbrace{a, ...,a}_{n \text{times}}\,), \quad m_{n}(u) \equiv \tau(u^n) \, .
\ee
This formula follows from starting with the moment-cumulant relation \eqref{cumudef} for the case of two free random variables, and then repackaging the free cumulants of $u$ back into moments. The derivation can be found in~\cite{nica2006lectures}.

\begin{table}[!h]
  \centering
  % vertical bars added: ||Y|Y|
  %\begin{tabularx}{\linewidth}%{|M{1.2cm}|C|C|}
\begin{tabular}{|M{1.1cm}|M{0.5\linewidth}|M{0.4\linewidth}|}
    \hline
    \textbf{$n$} & \bf{Even NC Partitions $\pi$} & \bf{$\overline{\OTOC_n(t)}$} \\
    \hline
   1 &  $(1,2)$, $c=1$, $|\pi^c|=1^2$ & \underline{ $\eta(t)^2$}  \\[20pt]
    \hline
    2 & $(12)(34)$, $c=2$, $|\pi^c|=1^22^1$ \quad \quad \quad \quad  \quad  \quad 
    $(4321)$, $c=1$, $|\pi^c|=1^4$
    &  $\underline{2\eta(t)^2\eta(2t)}-\eta(t)^4$ \\[10pt]
    \hline 
   3 & $(12)(34)(56)$, $c=2$, $|\pi^c|=1^33^1$\quad \quad \quad \quad 
   $(12)(36)(45)$, $c=3$, $|\pi^c|=1^22^2$ \quad \quad \quad \quad $(12)(3456)$ , $c=6$, $|\pi^c|=1^42^1$\quad \quad \quad \quad   $(123456)$, $c=1$, $|\pi^c|=1^6$ & $\underline{2\eta(t)^3\eta(3t)}+\underline{3\eta(t)^2\eta(2t)^2}-6\eta(t)^4\eta(2t)+2\eta(t)^6$ \\[20pt]
    \hline
       4 & $(12)(34)(56)(78)$, $c=2$, $|\pi^c|=1^44^1$\quad \quad \quad \quad $(12)(34)(58)(67)$, $c=8$, $|\pi^c|=1^32^13^1$\quad \quad \quad \quad $(14)(23)(58)(67)$, $c=4$, $|\pi^c|=1^22^3$\quad \quad \quad \quad $(1234)(56)(78)$, $c=4$, $|\pi^c|=1^53^1$\quad \quad \quad \quad $(1234)(58)(67)$, $c=4$, $|\pi^c|=1^42^2$\quad \quad \quad \quad $(1256)(34)(78)$, $c=8$, $|\pi^c|=1^42^2$\quad \quad \quad \quad
   $(1234)(5678)$, $c=4$, $|\pi^c|=1^62^1$ \quad \quad \quad \quad $(12)(345678)$ , $c=8$, $|\pi^c|=1^62^1$\quad \quad \quad \quad   $(12345678)$, $c=1$, $|\pi^c|=1^8$ & $\underline{2\eta(t)^4\eta(4t)}+\underline{8\eta(t)^3\eta(2t)\eta(3t)}+\underline{4\eta(t)^2\eta(2t)^3}-8\eta(t)^5\eta(3t)-20\eta(t)^4\eta(2t)^2+20\eta(t)^6\eta(2t)-5\eta(t)^8$ \\[20pt]
    \hline
    5 & $(12)(34)(56)(78)(910)$, $c=2$, $|\pi^c|=1^55^1$\quad \quad \quad \quad $(12)(34)(56)(710)(89)$, $c=10$, $|\pi^c|=1^42^14^1$\quad \quad \quad \quad $(12)(34)(510)(67)(89)$, $c=5$, $|\pi^c|=1^43^2$\quad \quad \quad \quad $(12)(34)(510)(69)(78)$, $c=20$, $|\pi^c|=1^32^23^1$\quad \quad \quad \quad $(12)(310)(49)(58)(67)$, $c=5$, $|\pi^c|=1^22^4$\quad \quad \quad \quad $(1234)(510)(69)(78)$, $c=50$, $|\pi^c|=1^42^3$\quad \quad \quad \quad $(1234)(56)(710)(89)$, $c=60$, $|\pi^c|=1^52^13^1$\quad \quad \quad \quad $(1234)(56)(78)(910)$, $c=10$, $|\pi^c|=1^64^1$\quad \quad \quad \quad $(1234)(56710)(89)$, $c=35$, $|\pi^c|=1^62^2$\quad \quad \quad \quad  $(1234)(5678)(910)$, $c=10$, $|\pi^c|=1^73^1$\quad \quad \quad \quad $(123456)(710)(89)$, $c=35$, $|\pi^c|=1^62^2$\quad \quad \quad \quad $(123456)(78)(910)$, $c=10$, $|\pi^c|=1^73^1$\quad \quad \quad \quad
   $(1234)(5678910)$, $c=10$, $|\pi^c|=1^82^1$ \quad \quad \quad \quad $(12)(345678910)$ , $c=10$, $|\pi^c|=1^82^1$\quad \quad \quad \quad   $(12345678910)$, $c=1$, $|\pi^c|=1^{10}$ & $\underline{2\eta(t)^5\eta(5t)}+\underline{5\eta(t)^2\eta(2t)^4}+\underline{5\eta(t)^4\eta(3t)^2}+\underline{10\eta(t)^4\eta(2t)\eta(4t)}+\underline{20\eta(t)^3\eta(2t)^2\eta(3t)}+14\eta(t)^{10}-70\eta(t)^8\eta(2t)+105\eta(t)^6\eta(2t)^2-50\eta(t)^4\eta(2t)^3+30\eta(t)^7\eta(3t)-60\eta(t)^5\eta(2t)\eta(3t)-10\eta(t)^6\eta(4t)$ \\[20pt]
    \hline
  \end{tabular}
  \caption{For $n$ from 1 to $5$, we list the non-crossing  partitions of $2n$ elements contributing to \eqref{eq:otocGUE}. We treat partitions that are related by a cyclic shift as equivalent and only write down one element of each such equivalence class explicitly, indicating the multiplicity by $c$. $|\pi^c|$ indicates the block structure of the Kreweras complement of $\pi$:  for example, $|\pi^c|=|W_1|^2 |W_2|$ indicates that there are 2 blocks of size $|W_1|$ and 1 block of size $|W_2|$. The resulting full expressions for  $\mathrm{OTOC}_n(A(t), A)$ are shown on the right. We have underlined the terms with the smallest number of total factors of $\eta$ which come from noncrossing perfect matchings ($\pi$ with only $2$-cycles).}
  \label{gue_table}
\end{table}

Now $m_n(u)$ can be found using the GUE average: 
\begin{equation}\label{eq:momentsGUE}
    m_n(u)=  \lim_{d\to\infty}\frac{1}{d}\overline{\Tr[e^{-inHt}]}= \frac{J_1(2nt)}{nt}\ .
\end{equation}
The random variable $a$ is a traceless involution, $a^2=1$, so its moments are simple:
\begin{equation}\label{eq:momentsinvolution}
    m_{2n}(a)=1,\quad m_{2n+1}(a)=0\,
\end{equation}
Its free cumulants can be found by inverting the relation 
  \be 
  m_n(a)=\sum_{\pi\in {\rm NC}(n)}\prod_{B\in\pi}\kappa_{|B|}(a)\, . 
  \ee  
  This is a relatively straightforward special case of \eqref{cumudef}, and  it can be seen using standard techniques from~\cite{nica2006lectures} that
  \be \label{cumulantsof_a}
\kappa_n(a) =\begin{cases}
(-1)^{\frac{n}2-1}C_{\frac{n}2-1} \ ,\quad n\ \text{is even}\\
0\ ,\ \ \quad\quad\quad\quad\quad\quad n\ \text{is odd}
\end{cases}
  \ee
Putting everything together, we end up with the formula \eqref{eq:otocGUE} in the main text. We explicitly write down the formula up to $n=5$ in Table~\ref{gue_table}.

\section{Higher-point OTOCs in random unitary circuits}
\label{app:random_circuits}

In this section, we will explain  the calculation of $\otoc_n$ in random unitary circuits, and provide an argument for the conjectured result in~\eqref{otoc_formula_randomcircuit_mt}. We use the partition function method developed for studying the $n$-th R\'enyi entropy $S_n$ in random unitary circuits in~\cite{zhou_nahum}. For readers familiar with the methods used in~\cite{zhou_nahum}, we note that the partition function for the calculation of $\otoc_n$ has the same ``bulk'' rules as the partition function for $S_n$, but the boundary conditions are different. 
This allows us to make use of the ``membrane picture'' developed  for $S_n$ while calculating $\otoc_n$. The difference in the boundary conditions makes the calculation much more non-trivial than that of $S_n$, or that of $\otoc_2$ in \cite{operator_spreading_adam, operator_spreading_tibor}, and we need to take into account certain effects that are not captured by the membrane picture. %In particular, for certain configurations in the partition function, we find that assuming an intuitive formula based on the membrane picture is not self-consistent while calculating $\otoc_n$ (in a sense we will make explicit below). 

We will present the partition function calculation for $\otoc_n$ and review the relevant assumptions of the membrane picture in a mostly self-contained way in this appendix for clarity. The basic setup follows directly from that in \cite{zhou_nahum}. Our main innovation is the argument in  sections~\ref{sec:no_dw} and~\ref{sec:two_dw} showing that the membrane picture does not apply in  certain cases due to subtle cancellations among various configurations.  This argument is the underlying reason why only $k \geq n/2$ appear in the sum in \eqref{otoc_formula_randomcircuit_mt}. Recall that the appearance of only $k \geq n/2$ in this sum makes it plausible that  $\otoc_n$ decay with $n$ for $t\geq v_B$, and hence that $I_{\rm free}$ becomes finite for $t\geq v_B$.

\subsection{Definitions and properties of permutation states} \label{sec:background}

To set up the partition function calculation, it will be useful to first introduce certain states on $2n$ copies of the system associated with permutations in $\sS_n$. In this subsection, we define these states and discuss their relevant properties. We also  summarize certain lemmas about the permutation group which will later be useful.

For $\sigma \in \sS_n$ and any operator $O$ acting on one site, define  the states $\ket{O, \sigma}$ on $2n$ copies of one site: 
\be 
\braket{i_1 i_1'... i_n i_n'|O, \sigma} = \frac{1}{q^{n/2}}O_{i_1 i'_{\sigma(1)}}...O_{i_n i'_{\sigma(n)}}
\ee 
Recall that $q$ is the Hilbert space dimension of one site. In particular, we can label the states corresponding to $O = \mathbf{1}$ as follows: 
\be 
\braket{i_1 i_1'... i_n i_n'|\sigma} = \frac{1}{q^{n/2}}\delta_{i_1 i'_{\sigma(1)}}...\delta_{i_n i'_{\sigma(n)}} \, . \label{sigmadef}
\ee 
Two permutations which we will denote with a special notation are the identity permutation, denoted by $e$, and the cyclic permutation $(n \, n-1\, n-2 \, ... \, 1)$, denoted by $\eta$.

The overlaps among the permutation states are given by 
\begin{align}
\braket{\sigma|\tau} = q^{-d(\sigma, \tau)}\label{perm_overlap}
\end{align}
where $d(\sigma, \tau)$ is the Cayley distance between the permutations $\sigma$ and $\tau$, defined as the minimum number of transpositions with which one needs to multiply $\sigma$ to obtain $\tau$ (or vice versa). More explicitly, 
\be 
d(\sigma, \tau) = n - k(\sigma \tau^{-1}) \label{cayley} 
\ee
where $k(\sigma)$ denotes the total number of cycles in the permutation  $\sigma$. Another useful overlap which will play a role in the discussion below is 
\be  \braket{\sigma|O, \tau} = \frac{1}{q^n} \prod_{i=1}^{k(\sigma \tau^{-1})}\Tr[O^{|c_i|}] \label{frule}
\ee
where $c_i \in \sigma \tau^{-1}$ refers to the cycles in  $\sigma \tau^{-1}$, and $|c_i|$ is the number of elements in the $i$-th cycle. In the discussion below, $O$ will be either $A$ or $B$, both of which have the properties that 
\be 
O^2 =\mathbf{1}, \quad  \Tr[O]= 0 \, . 
\ee
For such $O$, we always have $\braket{\sigma|O, \tau}=0$ for odd $n$, as at least one of the $|c_i|$ in \eqref{frule} is odd. For even $n$, 
\be 
\braket{\sigma|O, \tau} = \begin{cases}
q^{-d(\sigma, \tau)} & \sigma \tau^{-1} \in P_{\rm even}\\ 
0 & \sigma \tau^{-1} \notin P_{\rm even}
\end{cases} \label{op_overlap}
\ee
where we have defined a subset $P_{\rm even}\subset \sS_n$ as follows: 
\be 
P_{\rm even} \equiv \{\tau \in \sS_n|~ |c_i| \text{ is even for each } c_i \in \tau \} \label{xdef} \, . 
\ee
For example, for $\sS_6$, $(12)(36)(45) \in P_{\rm even}$ while $\sS_6$, $(12)(36)(4)(5) \notin P_{\rm even}$.

Another subset of permutations that will play an important role below is 
\be 
P_{\rm NC} \equiv \{\tau \in \sS_n|\,  k(\tau) + k(\eta \tau^{-1}) = n+1\,\} .  \label{ncdef}
\ee 
Such $\tau$ are in one-to-one correspondence with the non-crossing partitions of $n$ elements, and we will sometimes refer to them as ``non-crossing permutations.'' Note  that the following inequality is obeyed by any permutation in $\tau \in \sS_n$:
\be 
k(\tau) + k(\eta \tau^{-1}) \leq n+1 \, . \label{nc_ineq} 
\ee
This is equivalent to the triangle inequality 
\be 
d(\tau, e) + d(\tau, \eta) \geq d(e, \eta) \, . 
\ee
The non-crossing permutations saturate this inequality. 

Let us now state two lemmas about  permutations in these subsets of  $\sS_n$ which we will use in the later discussion. These lemmas can be referred to when they appear in the later subsections. The proofs of both lemmas are provided in Appendix~\ref{app:perm_proofs}. To state Lemma~\ref{lemma_noncross_1}, in addition to the sets $P_{\rm even}$ and $P_{\rm NC}$ defined above, let us define a third set 
\be 
Q_{\rm singleton} \equiv  \{\tau \in \sS_n|~ |c_i| =1 \text{ for at least one }  c_i \in \eta\tau^{-1} \} \, .  \label{wdef}
\ee
By definition, its complement is 
\be 
\bar Q_{\rm singleton} \equiv  \{\tau \in \sS_n|~ |c_i| >1 \text{ for all  }  c_i \in \eta\tau^{-1} \} \, .  \label{wdefq}
\ee

\begin{restatable}{lemma}{lemmagone}\label{lemma_noncross_1} Let $\tau$ be an NC permutation of $n$ elements where $n$ is even.
 (a) For any $\tau \in P_{\rm even}$, we have $k(\eta\tau^{-1})\geq \frac{n}{2} +1$, and $\eta \tau^{-1}$ has at least one cycle of length 1. (b) For any $\tau \in \bar Q_{\rm singleton}$, we have  $k(\tau)\geq \frac{n}{2} +1$ and $ \tau \notin P_{\rm even}$.

\end{restatable}
%\noindent %Note that by definition, the complement of $Q_{\rm even}$ is 
\begin{comment}
\begin{equation}
    \bar Q_{\rm singleton} \equiv  \{\tau \in \sS_n|~ |c_i|>1 \text{ for all } \in \eta\tau^{-1} \} \ . \label{wcdef}
\end{equation}
\end{comment}
Part (a)  of Lemma~\ref{lemma_noncross_1} implies in particular that 
\be 
P_{\rm even}\cap P_{\rm NC}\subseteq  Q_{\rm singleton} \cap P_{\rm NC}\, .
\ee
Another useful statement about $ Q_{\rm singleton}$  is the following lemma: 
\begin{restatable}{lemma}{lemmagtwo}\label{lemma_noncross_2}
For $\tau \in  Q_{\rm singleton} \cap P_{\rm NC}$, any $\sigma$ that satisfies  
\be \label{eq:sigma_on_geodesic}
d(\eta, \tau) = d(\eta, \sigma) + d(\sigma, \tau) 
\ee
(i.e., any $\sigma$ that lies on the geodesic between $\eta$ and $\tau$)
is also in $ Q_{\rm singleton} \cap P_{\rm NC}$. 
\end{restatable}

\subsection{Partition function for $\OTOC_n$}\label{sec:partition_function}

We want to calculate $\otoc_n(A(t):B)$ for the setup shown in Fig.~\ref{fig:randomcircuit_diagram}. Recall that each two-site unitary $U^{(t)}_{x, x+1}$ in the circuit is drawn independently from the ensemble of Haar-random unitaries from $\mathbf{U}(q^2)$. The full time-evolution operator $U(t)$ up to time $t$ is constructed by taking products of these two-site unitaries. Below, $L$ denotes the total number of sites. We will assume that $L$ is large enough that we can ignore all edge effects. 
We will always assume that $A$ lies at location $x'=0$ and $B$ at $x'=x$, and we will often write $\otoc_n$ as a function of $t$ and the velocity $v=x/t$. 

In terms of the permutation states introduced in the previous subsection, we can rewrite $\otoc_n = \frac{1}{q^L}\Tr[\le(U(t)AU(t)^{\dagger}B\ri)^n]$ as a transition amplitude on $2n$ copies of the system as follows: 
\begin{align} 
{\rm OTOC}_n\le(t, v\ri) =  &\bra{\eta} \bra{\eta} ... \bra{\eta} \bra{B, \eta}_{x'=vt}\bra{\eta} ... \bra{\eta} \nn 
&(U(t) \otimes  U(t)^{\ast})^{\otimes n} \nn 
&\ket{e}\ket{e} ... \ket{e} \ket{A, e}_{x'=0} \ket{e} ... \ket{e} \times q^{(n-1)L} \,  . \label{g36} 
\end{align}
The above expression is simply a rewriting of $\otoc_n$ by manipulating various indices, and is true for any $U(t)$. This rewriting proves useful in random unitary circuits because the average of $(U\otimes U)^{\otimes 2n}$ for each $U=U_{x, x+1}^{(t)}$ can also be expressed in terms of permutation states: 
\be 
\overline{\le(U^{(t)}_{x, x+1} \otimes {U^{(t)}_{x, x+1}}^{\ast}\ri)^{\otimes n}} =\sum_{\sigma, \tau \in \sS_n}  g^{\sigma \tau}  \ket{\sigma}_x\ket{\sigma}_{x+1} \bra{\tau}_x\bra{\tau}_{x+1} \label{ux}
\ee
where $g^{\sigma \tau}$ is the inverse of the matrix 
\be 
g_{\sigma \tau} = (q^2)^{-d(\sigma, \tau)} 
\ee
 (In terms of the more standard notation of the Weingarten function, $g^{\sigma \tau} = (q^2)^n \,  \text{Wg}(\sigma \tau^{-1}, q^2)$.)

By putting \eqref{ux} into \eqref{g36} for each two-site random unitary appearing in $U(t)$, we get a partition function on a two-dimensional lattice, with the top and bottom boundary conditions given by the top and bottom lines of \eqref{g36}, and the bulk interactions determined by \eqref{ux}. In figure Fig.~\ref{fig:pf_derivation},  we further explain the bulk interactions, and how some of the sites can be integrated out to get a triangular lattice. Each $s_i$ corresponds to some permutation in $\sS_n$, so there are $n!$ states on each site of the lattice. 

\begin{figure}[!h]
    \centering
    \includegraphics[width=0.9\linewidth]{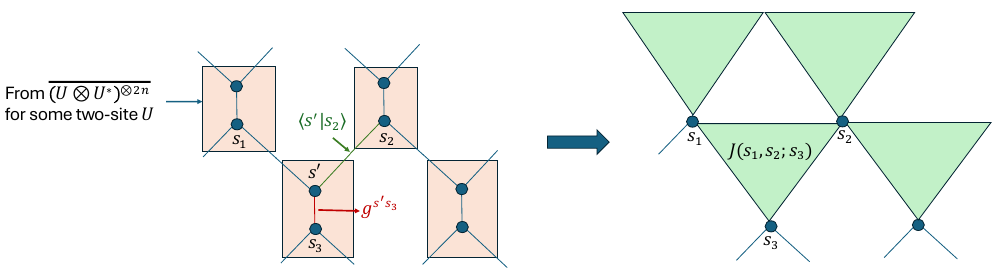}
\caption{We show a segment in space-time of the average of $(U(t)\otimes U(t)^*)^{\otimes n}$ over random unitaries, obtained using \eqref{ux} and the structure of the circuit in Fig.~\ref{fig:randomcircuit_diagram}.   The vertical lines come from \eqref{ux} and the diagonal lines correspond to taking inner products between permutation states according to the rule \eqref{perm_overlap}. By integrating out each of top states such as $s'$, we get a triangular lattice as shown on the right.} \label{fig:pf_derivation}
\end{figure}

Explicitly, the interactions $J(s_1, s_2; s_3)$  are given by 
\be 
J(s_1, s_2; s_3) = \sum_{\mu} g^{\mu \, s_3} \braket{s_1| \mu} \braket{s_2| \mu} \label{jdef}
\ee
These interactions can be either positive or negative. They take a simple form for general $n$ only in two special cases: 
\begin{enumerate}
\item[(i)] If 
$s_1=s_2$, we have $J(s_1,s_1;s_3) = \delta_{s_1 s_3}$. 
\item[(ii)] If  $d(s_1, s_2)=1$, then 
\be 
J(s_1, s_2; s_3) = \frac{q}{q^2+1} (\delta_{s_1 s_3} + \delta_{s_2 s_3}) \, . \label{elementary}
\ee
\end{enumerate}
We expect in general that 
\be
|J(s_1, s_2; s_3)|\leq 1  \label{j_cost}
\ee 
for all $q$ and $n$ and all $s_i$. We will prove this statement for the case $q=2$ as Lemma~\ref{lemma:inequality} in Appendix~\ref{app:perm_proofs}. The inequality is saturated only in case (i) listed above, i.e., when $s_1=s_2=s_3$.

By further taking into account the contributions from the top and bottom boundary conditions, 
the full  lattice and interaction rules are summarized in Fig.~\ref{fig:lattice_general}.

\begin{figure}[!h]
    \centering
    \includegraphics[width=0.8\linewidth]{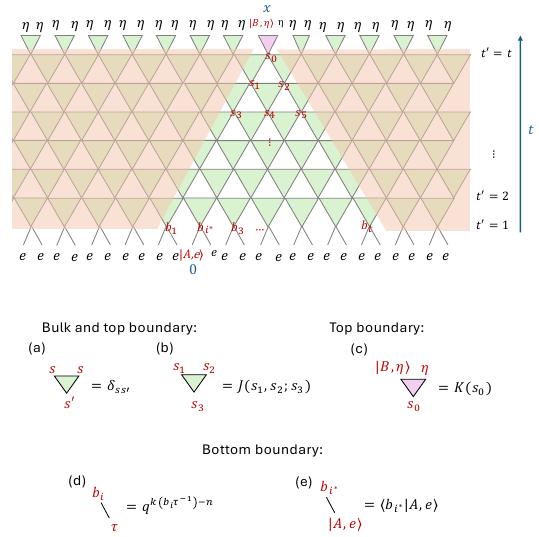}
    \caption{Lattice, interaction rules, and boundary conditions for the partition function for $\otoc_n$.}
    \label{fig:lattice_general}
\end{figure}
Rules (a) and (b) follow from the above statements about $J(s_1, s_2; s_3)$. Rule (c) can be deduced in a similar way to rule (b) using \eqref{frule} and \eqref{perm_overlap}, with \be 
K(s_0) = \sum_{ \tau} g^{\tau s_0} \braket{B, \eta|\tau} \braket{\eta|\tau} \, .   
\ee
$K(s_0)$ is generally non-zero for all $\sigma \in \sS_n$ and can take both positive and negative signs. In rules (d) and (e), we have used a special label $b_i$ for the states appearing in the lowermost label of the bulk, and used \eqref{frule} and \eqref{perm_overlap} to explain their contraction with the lower boundary conditions coming from the last line of \eqref{g36}. While only the $\tau= e$ case in rule (d) is needed so far, in the later discussion other quantities besides $\otoc_n$ will appear where $\tau$ will take other values. In the case $v< 1$, we use a distinguished label $b_{i^{\ast}}$  for the spin connected to $\ket{A, e}_x$ (for example, in the figure, $i^{\ast}=2$).

We consider all possible configurations of $\vec{s}$ and $\vec{b}$ on the triangular lattice of \eqref{fig:lattice_general}. For each configuration, we  multiply the contributions from each colored triangle, as well as the overlaps with the bottom boundary and the overall $q^{(n-1)L}$ factor in \eqref{g36}. We then sum over contributions from all $\vec{s}$, $\vec{b}$.

Let now note two  simplifications that follow immediately from rules (a)-(d): 
\begin{enumerate}
\item 
Rule (a) forces all spins in the shaded red region of the lattice to be $\eta$. All triangles in this shaded regions contribute a factor of 1. For this reason, only the undetermined $\vec s, \vec b$ spins are explicitly labelled in Fig.~\ref{fig:lattice_general}. 
\item 
${\rm OTOC}_n=0$ for all times when $n$ is odd, as in this case the overlap $\braket{b_i^{\ast}|A,e}=0$ for any choice of $b_i^{\ast}$, as discussed above \eqref{op_overlap}. 
\end{enumerate}

% We have used the label $b_i$, $i=1, ..., t$ for the undetermined spins on the bottom boundary, and $s_i$ for all remaining spins.  

We are then left with the task of summing over all possible configurations of $\vec s$, $\vec b$ for even $n$. Below, we will refer  to a triangle with $s_1\neq s_2$ ($s_1= s_2$) as a triangle with (without) a domain wall. Then the top layer at $t'=t$ has two adjacent domain walls, between $\eta$ and $s_0$ and between $s_0$ and $\eta$ respectively (unless $s_0=\eta$, in which case all other states are set to $\eta$ by rule (a)). The $n=2$ case studied in~\cite{operator_spreading_adam, operator_spreading_tibor} is simple, because there is only one possible type of domain wall, which moves downwards through the lattice according to the simple rule (ii) above. For general $n$,  $s_0$ can now be any of the $n!$ permutations in $\sS_n$. 

If $d(s_0, \eta)>1$, then unlike in the simple case \eqref{elementary}, $J(s_0, \eta; s_3)$
is generally non-zero for all $s_3\in \sS_n$ (in particular, $s_3$ is not forced to be either $s_1$ or $s_2$). This implies that the two domain walls starting at the top boundary can  in principle change into any number of further domain walls of various types as we go towards the bottom boundary. In fact, configurations where all three $s_i$ on each triangle are distinct from each other (i.e., where the number of triangles with domain walls is extensive in the volume of the unshaded region) will generically give a non-zero contribution. However, note that the general inequality \eqref{j_cost} makes it intuitive that each triangle with a domain wall comes with some cost, so that cases with an extensive number of domain walls are unlikely to dominate. We will state this intuition more formally in the next subsection. 

%However, based on studies of R\'enyi entropies in random unitary circuits~\cite{zhou_nahum}, it is understood that the dominant contributions  can often still be organized in terms of certain ``composite domain walls,'' as we will review in Sec.~\ref{sec:membrane_assumptions}. 

%\SV{It would be nice if we could say $J<1$ here to make it more intuitive that triangles with domain walls don't proliferate, but okay if we cannot.} \JW{no idea yet.}

Let us introduce some notation to better organize the sum over $\vec{s}$, $\vec{b}$ in $\otoc_n$: 
\be 
{\rm OTOC}_n = \sum_{\vec{s}, \vec{b}} \sC_{\text{top, bulk}}(\vec{s}, \vec{b})~  \sB(\vec{b})  
\ee
where 
\be 
\sC_{\text{top, bulk}}(\vec{s}, \vec{b}) = K(s_0)\prod_{\text{bulk triangles}} J(s_i, s_j ; s_k) 
\ee
and 
\be 
\sB(\vec b) = \begin{cases} q^{(n-1)2t} 
\braket{b_{i_{\ast}}|A, e} \prod_{i \neq i^{\ast}}\braket{b_i|e} & \text{if } v<1 \\
q^{(n-1)2t} \prod_{i }\braket{b_i|e} & \text{if }v>1
\end{cases}  \label{b_explicit}
\ee
 Note that $\sB$ comes from a combination of the overall factor of $q^{(n-1)L}$ in \eqref{g36} with the factors coming from overlaps on the lower boundary. (For sites in the shaded red region, the overlap $\braket{\eta|e} = q^{-(n-1)}$ cancels with $q^{(n-1)}$ from this factor.) It is useful to further organize the contributions according to $\vec{b}$: 
\be 
\otoc_n = \sum_{\vec b} \sC(\vec b) \, \sB(\vec b) , \quad \sC(\vec b) = \sum_{\vec s} \sC_{\text{top, bulk}}(\vec{s}, \vec{b})\, .  \label{cbdef}
\ee

Let us classify the configurations $\vec{b}$ near the lower boundary by the number and type of domain walls in them. Say there are $m$ domain walls in $\vec{b}$, whose locations are labelled $x_1, ..., x_m$ from left to right. Due to rule (a), we must have $x_1\geq x-t$ and $x_m \leq x+t$, with the first domain wall between $\eta$ and some $\tau_1$, and the last domain wall between some $\tau_{m-1}$ and $\eta$. It is convenient to replace the $x_i$ labels with the ``velocities'' $v_i \equiv \frac{x_i-x}{t}$.

\begin{figure}[!h]
    \centering
    \includegraphics[width=0.6\linewidth]{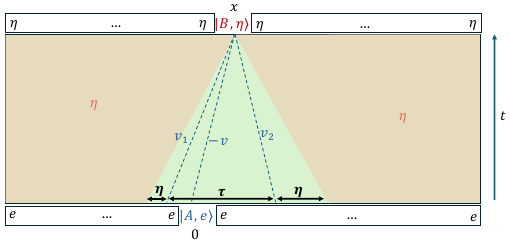}
    \caption{We schematically show the types of configurations corresponding to the second line of \eqref{exact}. We only fix the spins $b_i$ in the lowermost layer, and allow possible configurations in the green bulk region.}
    \label{fig:two_dw} 
\end{figure}

One possibility is that we have $m=0$, corresponding to no domain wall in $\vec{b}$. See Fig.~\ref{fig:no_dw} for examples of such configurations. 
There are no cases with just one domain wall, as we must have $\eta$ to the left and to the right of all domain walls. The next possibility is that we have a domain wall between $\eta$ and $\tau$ at some $v_1$ (for some arbitrary permutation $\tau \in \sS_n$), and then from $\tau$ to $\eta$ at $v_2$. We will use the shorthand $\vec b = (v_1, \tau, v_2)$ to denote this case. 
Similarly,  a configuration on the lower boundary with three domain walls can be labelled by $\vec{b}=(v_1, \tau_1, v_2, \tau_2, v_3)$. 
So the exact expression for $\otoc_n$ can be written as 
\begin{align}\label{exact}
\otoc_n(v, t) =~ &  \sC(b = \eta_i )\nn
+ & \sum_{v_1, \tau, v_2>v_1} \sC(\vec b = (v_1, \tau, v_2)) q^{(k(\tau)-1)t(v_2-v_1)} \braket{b_{i_{\ast}}|A, e}\nn
+& \sum_{v_1, \tau_1, v_2>v_1, \tau_2, v_3>v_2} \sC(\vec b = (v_1, \tau_1, v_2, \tau_2, v_3)) q^{(k(\tau_1)-1)t(v_2-v_1)} q^{(k(\tau_2)-1)t(v_3-v_2)} \braket{b_{i_{\ast}}|A, e} \nn 
 +& ...   
\end{align}
where we have used the explicit form of the $\braket{b_i|e}$ overlaps in the second and third lines. The $...$ indicates further contributions with more than three domain walls in $\vec{b}$. See  Fig.~\ref{fig:two_dw} for a schematic representation of the contributions to the second line of \eqref{exact}.

\subsection{Different regimes of $\OTOC_n$} \label{sec:regimes}

\begin{figure}[!h]
    \centering
    \includegraphics[width=0.5\linewidth]{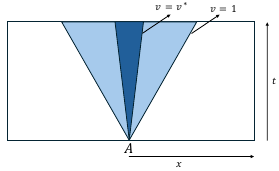}
    \caption{The unshaded, light blue, and dark blue regions respectively correspond to the first, second, and third lines of \eqref{otocn_tv}. We will focus on the leading contributions in the light blue region, where $\otoc_n(t,v)$ decreases for fixed $t$ as $v$ decreases.}%Different regimes of $\otoc_n(t,v)$. In the unshaded region, where $v>1, \otoc_n=1$. In the light blue region, $\otoc_n$ is a function of both $v$ and $t$. This is the  regime we focus on in our calculation below. $v^{\ast}$ indicates a critical velocity below which $\otoc_n$ becomes a function of only $t$ independent of $v$.  The value of $v^{\ast}$ (including whether it is non-zero), and the value of $\otoc_n$ in the dark blue region for $v<v^{\ast}$, is beyond the scope of our calculation. }
    \label{fig:otoc_regimes}
\end{figure}

Before analyzing the terms in \eqref{exact} in detail, let us make some general comments on the $v$- and $t$- dependence of $\otoc_n$ and the regime of interest. For $|v|>1$, from the sharp lightcone of the circuit, we have exactly 
\be 
\otoc_n(t, v) = 1, \quad |v| > 1\, . 
\ee
For $|v|<1$, we will see two kinds of terms in $\otoc_n$ below: terms proportional to $e^{-f(v) t}$, where $f(v)$ is some non-trivial non-negative function of $v$, and terms  proportional to $e^{-\Delta t}$, where $\Delta>0$ is independent of $v$. The functions  $f(v)$ are such that $f(v=1)=0$, and these functions grow as $v$ is decreased from 1 to 0. Hence, for any fixed large $t$ and $v$ sufficiently close to 1, the terms of the form $e^{-f(v)t}$ dominate over terms of the form $e^{-\Delta t}$. It is possible that deep within the lightcone (i.e. for $v$ small enough), some $\Delta$ becomes smaller than each of the $f(v)$, in which case $\otoc_n(t,v)$ saturates when seen as a function of $v$ for fixed $t$. So we have the following structure (see Fig.~\ref{fig:otoc_regimes}): 
\be \label{otocn_tv}
\otoc_n(t, v)= \begin{cases} 1 & |v|>1 \\
F(t, v) & v^{\ast} < |v|< 1 \\
e^{-\Delta t} & |v|< v^{\ast} 
\end{cases} 
\ee
Here $F(t,v)$ is a function with a non-trivial dependence on both $t$ and $v$ that comes from terms of the form $e^{-f(v)t}$. 
In our calculation below, we will focus on the $e^{-f(v)t}$ terms, which corresponds to being in the regime $v^{\ast}<|v|<1$. Analyzing terms of the form $e^{-\Delta t}$ (and in particular, seeing whether $v^{\ast}>0$, or whether terms of the form $e^{-f(v)t}$ dominate for all $v$) will be beyond the scope of our calculation. In all cases where we write $O(e^{-\Delta t})$ below, it is implicit that we mean $\Delta$ is independent of $v$.

\subsection{Assumptions about the structure of $\sC(\vec{b})$ for $\vec{b}$ with domain walls } \label{sec:membrane_assumptions}

So far, the full sum in \eqref{exact} is exact, and we have  organized it in a particular way according to the numbers and types of domain walls in $\vec{b}$. In order to go further, we need to identify the dominant configurations in the bulk factors $C(\vec b)$ and their total contribution. At this point, we will assume that the dominant contributions to $C(\vec b= (v_1, \tau, v_2))$ in a scaling limit where $t, x\gg 1$ come from bulk contributions with approximately localized domain walls, such as the ones shown in Fig.~\ref{fig:just_two_dw}. 
\begin{figure}[!h]
    \centering
\includegraphics[width=0.6\linewidth]{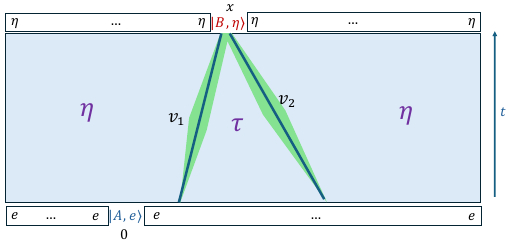}
    \caption{We assume that the dominant contribution to $\sC(\vec b = (v_1, \tau,v_2))$ comes from bulk configurations like the one shown above, where the domain wall is localized in the bulk up to some $O(1)$ width. On the sites of the lattice lying within the green regions of $O(1)$ width, we sum over all possible permutations in $\sS_n$. Outside these regions of $O(1)$ width, we have  uniform regions with only $\tau$ or  only $\eta$ states.}
    \label{fig:just_two_dw}
\end{figure}

Moreover, we assume that the sum over the many microscopic contributions in the $O(1)$ region around the domain wall leads to a simple net contribution. For example, the left domain wall in Fig.~\ref{fig:just_two_dw} contributes a factor of~\footnote{Note that this expression is not specific to $\eta$, and we expect a similar expression to apply for a domain wall between any $\sigma,\tau\in \sS_n$.} 
\be 
\sC_{\rm left}(v_1) = q^{-t \, d(\eta, \tau) \, \sE_{d(\eta, \tau)+1}(v_1) } \label{cleft} 
\ee
Here $d(\eta, \tau)$ is the Cayley distance between $\tau$ and $\eta$, defined in \eqref{cayley}. In this and all expressions below, we are ignoring some overall $O(1)$ prefactor. We have introduced a new ingredient, the membrane tension $\sE_{n}(v)$~\cite{zhou_nahum, jonay}, which is a symmetric,  positive, convex $O(1)$ function of the velocity with the general properties 
\be 
\sE_n(v) \geq |v|, \quad \sE_n'(v_B) =1, \quad \sE_n(v_B) =v_B\, .  \label{membrane_constraints}
\ee
where $0<v_B<1$ is an $n$-independent velocity called the butterfly velocity. See Fig.~\ref{fig:membrane_example}. Note in particular that $\sE_n(v)$ is always minimized at $v=0$ and grows monotonically as a function of $|v|$. The precise functional form for $\sE_n(v)$ is known for random unitary circuits only for the case $n=2$, where it can be checked to obey \eqref{membrane_constraints}. More generally, the constraints \eqref{membrane_constraints} are expected to hold based on expected physical properties for the evolution of the $n$-th R\'enyi entropies for general $n$ in random unitary circuits, as argued in~\cite{jonay}.

\begin{figure}[!h]
    \centering
    \includegraphics[width=0.3\linewidth]{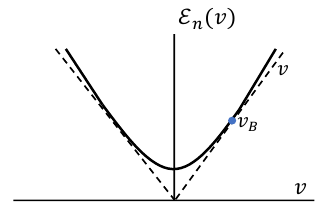}
    \caption{Example of the functional form of $\sE_n(v)$}
    \label{fig:membrane_example}
\end{figure}

  While we are not using the large $q$ limit in our calculation, the factor of $d(\eta, \tau)$ appearing in the exponent in  \eqref{cleft} is easiest to understand in this limit. In this case, the dominant contribution comes from the $\eta$-$\tau$ domain wall of minimal length from $x$ to $x+v_1t$, which passes through $t$ triangles. Each triangle with a domain wall contributes a simple weight in this limit: 
\be 
J(\eta, \tau; \eta) = J(\eta, \tau; \tau) = \frac{1}{q^{d(\eta, \tau)}}, \quad q \gg 1\, \, , 
\ee  
 so that 
 we simply get
 \be 
\sC_{\rm left}(v_1) =q^{-d(\eta, \tau) t}, \quad q \gg 1 \, . 
\ee
These expressions suggest that a domain wall between $\eta$ and $\tau$ can be seen as a product of $d(\eta , \tau)$ ``elementary'' domain walls associated with transpositions. Away from the large $q$ limit, we get a non-trivial combinatorial factor from counting the number of  paths through the lattice that can start at $x$ and end at $x-v_1t$, which leads to a non-trivial function $\sE(v_1)$ appearing next to the factor $d(\eta, \tau)$. The dependence of the function $\sE(v)$ itself on $d(\sigma, \tau)$ can be seen as a result of certain ``interactions'' among the elementary domain walls (see \cite{zhou_nahum} for more details).~\footnote{Here we have made the simplifying assumption that the the membrane tension $\sE$ for a domain wall between $\sigma$ and $\tau$ depends on $\sigma \tau^{-1}$ only through  $d(\sigma, \tau)$. It is expected that the function can depend on the more detailed information about $\sigma \tau^{-1}$, such whether the different transpositions appearing in its decomposition are commuting or non-commuting~\cite{zhou_nahum}. Such details will not significantly affect the final form of the answer that we want to emphasize, so we will ignore them to simplify the presentation.}

\begin{comment} 
$\sE_n(v)$ for any $v$ and $n$ can be set to 1 \SV{[add footnote for why]}, so \eqref{cleft} simply states that the contribution of the domain wall between $\eta$ and $\tau$ is $C_{\rm left}=(q^{-d(\eta, \tau)})^t$. This can be understood from the weights 
\be 
J(\eta, \tau; \eta) = J(\eta, \tau; \tau) = \frac{1}{q^{d(\eta, \tau)}}, \quad q \gg 1\, 
\ee
and the fact that the largest contribution to $C_{\rm left}$ in the large $q$ limit involves $t$ triangles with $\eta-\tau$ domain walls. Comparing to the large $q$ limit of \eqref{elementary}, which gives $1/q$, we see  As discussed in \cite{zhou_nahum}, this interpretation continues to hold away from the large $q$ limit, with the non-trivial interactions among the elementary domain walls captured by the function $\sE_{d(\eta, \tau)+1}(v)$. 
\end{comment}

Combining the factors from the left and right domain walls, it is natural to assume that the dominant contribution to $\sC(\vec b=(v_1, \tau, v_2))$ is given by  
\be 
C(\vec b=(v_1, \tau, v_2)) = q^{-t \, d(\eta, \tau) \, \sE_{d(\eta, \tau)+1}(v_1) } ~ q^{-t \, d(\eta, \tau) \, \sE_{d(\eta, \tau)+1}(v_2) } \, .  \label{c_simple}
\ee
One subtlety is that this approximation does not take into account a variety of additional contributions with more domain walls in the bulk, such as the ones shown in Fig.~\ref{fig:additional_contributions}.   Intuitively, these contributions can be absorbed into some $O(1)$ factor in \eqref{c_simple}. However, because these configurations can come with both positive and negative signs, there is a possibility that one can have a subtle cancellation, such that the leading contribution ends up being further exponentially suppressed in $t$ compared to   \eqref{c_simple}. We will argue below that this subtle cancellation must take place for certain choices of $\tau$. Hence, more generally, we expect to have an inequality 
\be 
|C(\vec b=(v_1, \tau, v_2))| \leq q^{-t \, d(\eta, \tau) \, \sE_{d(\eta, \tau)+1}(v_1) } ~ q^{-t \, d(\eta, \tau) \, \sE_{d(\eta, \tau)+1}(v_2) } \, .  \label{g53}
\ee
We are assuming that the dominant contribution to  $C(\vec b = (v_1, \tau, v_2))$
can never become {\it larger} than \eqref{c_simple} due to these additional configurations (apart from possibly being enhanced by some $O(1)$ prefactor). This assumption is intuitive based on the additional cost associated with each domain wall from \eqref{j_cost}, but we will not attempt the analysis of the additional combinatorial factors that is needed to rigorously show it. 

\begin{figure}
    \centering
    \includegraphics[width=\linewidth]{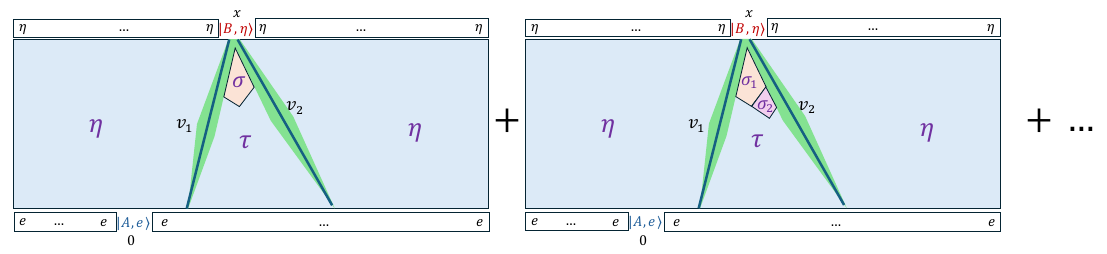}
    \caption{A series of additional contributions to $C(\vec b = (v_1, \tau, v_2))$ which are not accounted for in the formula \eqref{c_simple}.}
    \label{fig:additional_contributions}
\end{figure}

With additional domain walls, we have suppression by additional factors similar to \eqref{cleft}. For example, we will assume that 
\be 
|C(\vec b=(v_1, \tau_1, v_2, \tau_2, v_3))| \leq q^{-t \, d(\eta, \tau_1) \, \sE_{d(\eta, \tau_1)+1}(v_1) } ~ q^{-t \, d(\tau_1, \tau_2) \, \sE_{d(\tau_1, \tau_2)+1}(v_2) }~  q^{-t \, d(\tau_2, \eta) \, \sE_{d(\tau_2, \eta)+1}(v_3) } \, .  \label{three_dw}
\ee

\subsection{Analysis of the no-domain wall term in $\otoc_n$}
\label{sec:no_dw}

Let us now turn to the first term in \eqref{exact}, which has no domain walls on the lower boundary. Some of the contributions to this term are shown in Fig.~\ref{fig:no_dw}. Note that each such configuration can potentially come with either sign depending on the specific permutations that appear,  and the lengths of various domain walls.  
Since this term includes configurations with $O(1)$ domain wall length in the bulk (and even no domain walls anywhere in the case where we set $s_0=\eta$), it appears that it can be $O(1)$. However, we will show below that on summing over all contributions, this term is of the form $O(e^{-\Delta t})$ for $\Delta>0$ independent of $v$. By the general discussion of Sec.~\ref{sec:regimes}, this  term can therefore be ignored for the range of $v$ we are interested in. 

\begin{figure}
    \centering
    \includegraphics[width=\linewidth]{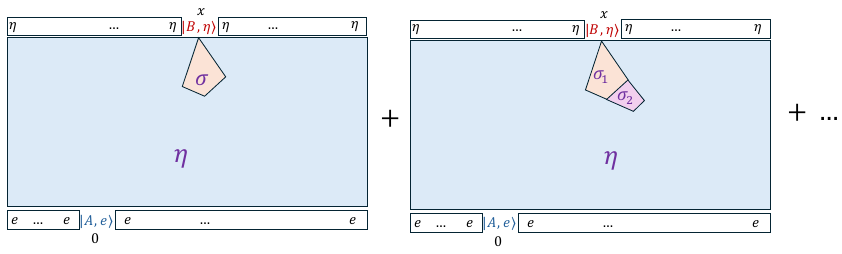}
    \caption{Contributions to the first term of \eqref{exact}, with no domain walls on the lower boundary. }
    \label{fig:no_dw}
\end{figure}

We will show this by first considering a different quantity instead of $\otoc_n$: 
\be 
M_{\eta}\equiv \bra{\eta} \bra{\eta} ... \bra{\eta} \bra{B, \eta} \bra{\eta} ... \bra{\eta} \overline{(U(t) \otimes U(t)^{\dagger})^{\otimes n}}\ket{\eta} \ket{\eta}... \ket{\eta}  \label{bn}
\ee
$M_{\eta}$ can be expressed as a partition function similar to the one in Fig.~\ref{fig:lattice_general} for $\otoc_n$, with the same top boundary conditions and bulk interactions. The boundary conditions on the lower boundary are different from $\otoc_n$, with $\ket{\eta}$ at every site instead of $\ket{e}$ or $\ket{A, e}$, and \eqref{bn} has no overall factor of $q^{(n-1)L}$. 

Using the definitions of the permutation states, we can also see that 
\be 
M_{\eta}= \overline{\Tr[U(t) B U(t)^{\dagger}]^n} = (\Tr B )^n=0
\ee
Now, using the same interaction rules as in Fig.~\ref{fig:lattice_general} with altered lower boundary conditions, and considering contributions with various numbers of domain walls in $\vec{b}$ to \eqref{bn}, we find that we can expand $M_{\eta}$ as follows:
\begin{align}
0 =  M_{\eta}  &= ~\sC(b_i = \eta) \nn 
& + \sum_{v_1, \tau,  v_2>v_1}\sC(\vec b = (v_1, \tau, v_2)) \frac{1}{q^{d(\eta, \tau)t(v_2-v_1)}} \nn 
& + \sum_{v_1, \tau_1,  v_2>v_1, \tau_2, v_3>v_2} \sC(\vec b = (v_1, \tau_1, v_2,\tau_2, v_3 )) \frac{1}{q^{d(\eta, \tau)t(v_2-v_1)}} \frac{1}{q^{d(\eta, \tau)t(v_3-v_2)}}   + ... \label{g51}
\end{align}
In particular, note from the first line that configurations with no domain walls in $\vec b$ give the same contribution in the new quantity $M_{\eta}$ as they do in $\otoc_n$: this comes from a combination of the fact that $M_{\eta}$ has $\eta$ at all sites in the lower boundary conditions, and the fact that it does not have the overall factor of $q^{(n-1)L}$ in $\otoc_n$. 
Now using \eqref{g53} and \eqref{three_dw} (and the analog of \eqref{three_dw} for more domain walls), we see that each line of \eqref{g51} starting from the second line is $O(e^{-\Delta t})$. Let us upper-bound the terms appearing in the second line more explicitly: since $v_2-v_1>0$,  
\begin{align} 
|\sC(\vec b = (v_1, \tau, v_2)) \frac{1}{q^{d(\eta, \tau)t(v_2-v_1)}}| &\leq |\sC(\vec b = (v_1, \tau, v_2))| \nn 
&\leq q^{-t \, d(\eta, \tau) \, \sE_{d(\eta, \tau)+1}(v_1) } ~ q^{-t \, d(\eta, \tau) \, \sE_{d(\eta, \tau)+1}(v_2) }\nn 
&\leq q^{-2t \, d(\eta, \tau) \, \sE_{d(\eta, \tau)+1}(0) }  = O(e^{-\Delta t})
\end{align}
where we have used the fact that $\sE_n(v)$ has a non-zero minimum value at $v=0$. The third line onwards of \eqref{g51} are further suppressed. 

Since the LHS of \eqref{g51} is zero, this implies that 
\be 
C(b_i=\eta) = O(e^{-\Delta t}) \label{ceta_small}
\ee
for $\Delta>0$ independent of $v$. 
We note again that this $O(e^{-\Delta t})$ scaling of $C(b_i=\eta)$ is counterintuitive from the diagrams in Fig.~\ref{fig:no_dw}, which have domain walls of $O(1)$ length. The exponentially small value in \eqref{ceta_small} is a consequence of subtle cancellations among configurations with different signs.

\subsection{Terms with two domain walls in $\otoc_n$}
\label{sec:two_dw}
Next, let us understand the second line of \eqref{exact}, which can be separated into a $v$-independent term and a $v$-dependent term as follows: 
\begin{align} \label{2dwdef}
\otoc_{n, \text{ 2 DW}} =&\sum_{v_1, \tau, v_2} \sC(\vec b = (v_1, \tau, v_2)) q^{(k(\tau)-1)t(v_2-v_1)} \braket{b_{i_{\ast}}|A, e} \nn 
=& \, \sZ_{\rm even}(t) + \sZ_{\rm odd}(t, v) 
\end{align} 
where 
\begin{align} 
& \sZ_{\rm even}(t) = \sum_{\tau \in P_{\rm even}, v_1, v_2>v_1} \sC(\vec b = (v_1, \tau, v_2)) q^{(k(\tau)-1)t(v_2-v_1)} \label{g58} \\
& \sZ_{\rm odd}(t, v) = \sum_{\tau \in \bar P_{\rm even},\, v_1>-v, \, v_2>v_1} \sC(\vec b = (v_1, \tau, v_2)) q^{(k(\tau)-1)t(v_2-v_1)} \label{g59}
\end{align}
In the above expressions, $P_{\rm even}$ is the subset of $\sS_n$ defined in \eqref{xdef}, $\bar P_{\rm even}$ is its complement in $\sS_n$, and we have used \eqref{perm_overlap} and \eqref{op_overlap}. In particular, the condition $v_1>-v$ in \eqref{g59} comes from the fact that $\braket{b_{i_{\ast}}|A, e}=0$ if    $b_{i_{\ast}}\in \bar P_{\rm even}$. 

\subsubsection{Restriction to non-crossing permutations}

For the dominant contributions to $\sZ_{\rm even}$ and $\sZ_{\rm odd}$, we can restrict to $\tau$ that lie respectively in the intersection of $P_{\rm even}\cap P_{\rm NC}$, and the intersection of $\bar P_{\rm even}\cap P_{\rm NC}$, where $P_{\rm NC}$ is the set of non-crossing permutations defined in \eqref{ncdef}.   

To see this, note that from \eqref{g53},  
\be
|\sC(\vec b = (v_1, \tau, v_2))| q^{(k(\tau)-1)t(v_2-v_1)} \leq f_1 f_2
\ee
where 
\begin{align} 
&f_1 = q^{-t\le((n-k(\eta \tau^{-1})) \sE_{d(\eta, \tau)}(v_1) +  (k(\tau)-1)v_1\ri)}\\
&f_2 = q^{-t\le((n-k(\eta \tau^{-1})) \sE_{d(\eta, \tau)}(v_2) -  (k(\tau)-1)v_2\ri)}\\
\end{align}
Now the first inequality of \eqref{membrane_constraints}, together with the inequality \eqref{nc_ineq}, implies that 
\be 
f_1 \leq q^{-t (k(\tau) -1)(\sE_{k(\tau)}(v_1)-|v_1|)}, \quad f_2 \leq q^{-t (k(\tau) -1)(\sE_{k(\tau)}(v_2)-|v_2|)}\ .
\ee
The inequalities are saturated in the case where $v_1 <0$, $v_2> 0$, and $\tau \in  P_{\rm NC}$. We can therefore approximate \eqref{g58} and \eqref{g59} as follows: 
\begin{align}
& \sZ_{\rm even}(t) = \sum_{\substack{\tau \in P_{\rm even}\cap P_{\rm NC}, \\ v_1,\, v_2>0}} \sC(\vec b = (v_1>0, \tau, v_2)) q^{(k(\tau)-1)t(v_2+v_1)} \label{zeven} \\
& \sZ_{\rm odd}(t, v) = \sum_{\substack{\tau \in \bar P_{\rm even} \cap P_{\rm NC} ,\\  0<v_1<v,\, v_2>0}} \sC(\vec b = (v_1, \tau, v_2)) q^{(k(\tau)-1)t(v_2+v_1)} \label{zodd}
\end{align}
Here we have made a convenient change of variables $v_1 \to -v_1$
so that the  values of both $v_1$ and $v_2$ appearing in the final expression  are positive. 

Note that if the inequality in \eqref{g53} is saturated, then for $\tau \in P_{\rm NC}$, 
\be 
\sC(\vec b = (v_1, \tau, v_2)) q^{(k(\tau)-1)t(v_2+v_1)} = q^{-(k(\tau)-1)t(\sE(v_1)-v_1 + \sE(v_2)-v_2)}\, 
\ee
If the above equality holds, then the exponent is 0 for the case $v_1 =v_2 =v_B$. However, we will show in the next section that for certain $\tau$, we always have 
\be 
\sC(\vec b = (v_1, \tau, v_2)) q^{(k(\tau)-1)t(v_2+v_1)} =  O(e^{-\Delta t})\, 
\ee
for $\Delta>0$ independent of $v_1$, $v_2$, so that \eqref{g53} cannot be saturated for such permutations. Like in the earlier discussion of $C(b_i=\eta)$, this is a consequence of cancellations among contributions with different signs.  

\vspace{0.2cm}

\subsubsection{Sub-dominance of contributions from $\tau \in  Q_{\rm singleton} \cap P_{\rm NC}$} \label{sec:membrane_failure}

In this section, we will argue that for any $\tau \in  Q_{\rm singleton} \cap P_{\rm NC}$, for $ Q_{\rm singleton}$ defined in \eqref{wdef}, 
\be
\sC(\vec b = (v_1, \tau, v_2)) q^{(k(\tau)-1)t(v_2+v_1)} =O( e^{-\Delta t})\, . 
\ee
 This argument can be seen as a more complicated version of the argument for $C(b_i=\eta)$ in Sec.~\ref{sec:no_dw}. This statement will allow us to further restrict the permutations that contribute to the leading part of $\otoc_n$, and we will discuss its consequences in the next section.

For some given $v_1$, $v_2$ (recall that we have redefined $v_1$ so that it is positive for left-moving domain walls), let us define the quantity 
\begin{align}
M_{ \tau, v_1, v_2} = &  \bra{\eta} \bra{\eta} ... \bra{\eta}\bra{B, \eta}\bra{\eta}... \bra{\eta} \overline{(U(t) \otimes U(t)^{\ast})^{\otimes n}}\nn 
&\ket{\eta}...\ket{\eta}_{y-v_1t-1} \ket{\tau}_{y-v_1t} ...\ket{\tau}_{y+v_2t} \ket{\eta}_{y+v_2t+1} ... \ket{\eta} \, \times q^{(k(\tau)-1)t(v_1+v_2)}.
 \label{tau_eta_1}
\end{align}
Again, the top boundary conditions and bulk rules for this quantity are the same as those for $\otoc_n$, but its lower boundary conditions are different. See Fig.~\ref{fig:all_cases}. Now let us refer to the region from $y-v_1t$ to $y+v_2t$
as $R$ and its complement as $\bar R$. 
By an abuse of notation, let us define states $\ket{O, \sigma}$ on $2n$ copies of the {\it full system} (as opposed to a single site) for some operator $O$ acting on the full system, as 
\be 
\braket{i_1 i'_1...i_ni'_n|O, \sigma} \equiv O_{i_1i'_{\sigma(1)}}... O_{i_1i'_{\sigma(1)}}
\ee
where each $\ket{i}$ labels a basis for the full system, 
and similarly for subsystem $R$, define 
\be 
\braket{i_1 i'_1... i_n i'_n|\sigma}_R \equiv \delta_{i_1 i'_{\sigma(1)}}...\delta_{i_n i'_{\sigma(n)}}
\ee
where each $\ket{i}$ labels a basis for subsystem $R$. 
Then we can write
\begin{align}
M_{\tau, v_1, v_2}&\propto 
 \bra{B, \eta} (U(t) \otimes U(t)^{\ast})^{\otimes n} \ket{\tau}_R \otimes  \ket{\eta}_{\bar R}\nn &=\bra{U(t)BU(t)^{\dagger}, \eta} \ket{\tau}_R \otimes  \ket{\eta}_{\bar R}  \nn 
& = \bra{U(t)BU(t)^{\dagger}, e}\ket{(\tau\eta^{-1})}_R\ket{e}_{\bar R} \label{gnew}
\end{align}
Now note that for any operator $O$ and any permutation $\sigma$, 
\be 
\bra{O, e} \ket{\sigma}_R \ket{e}_{\bar R} = \sum_{i_1,...,i_n=1}^{d_R}\sum_{j_1,...,j_n=1}^{d_{\bar R}} O_{i_1j_1, i_{\sigma(1)}j_1} O_{i_2j_2, i_{\sigma(2)}j_2}...O_{i_nj_n, i_{\sigma(n)}j_n}
\ee
where each $i$ labels a basis for $R$ and each $j$ labels a basis for $\bar R$. Hence, if $\sigma$ has any cycle of length 1, i.e. $\sigma(i)=i$ for any $i$, then the above expression is proportional to $\Tr[O]$. Returning to \eqref{gnew}, where $\sigma= \tau \eta^{-1}$
 and $O= U(t) B U(t)^{\dagger}$ which has trace zero, we find that  
\be 
M_{\tau, v_1, v_2}=0, \quad \tau \in Q_{\rm singleton}\cap P_{\rm NC} \, . 
\ee

\begin{figure}
    \centering
    \includegraphics[width=\linewidth]{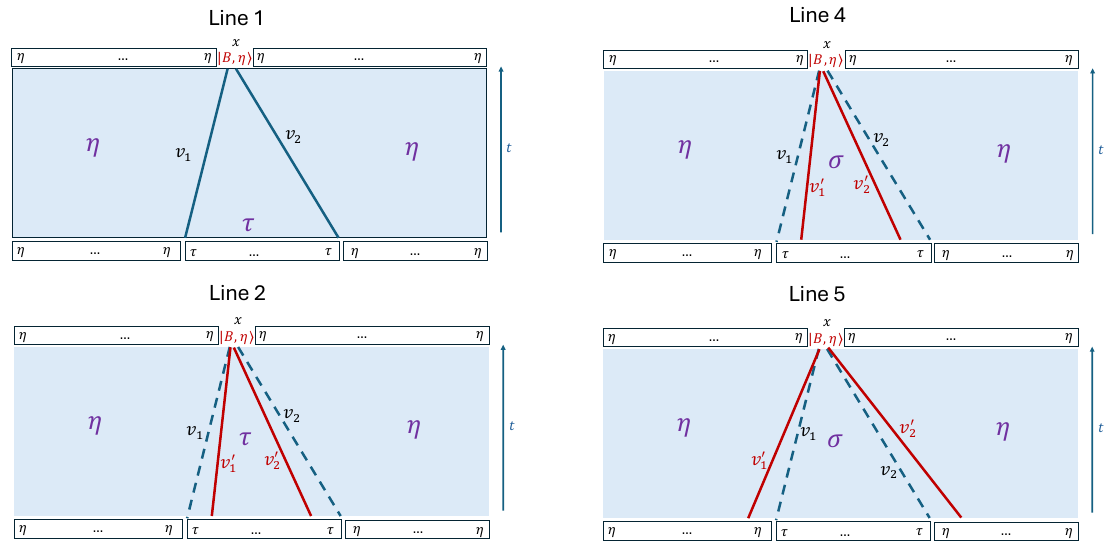}
    \caption{Domain wall configurations for various terms in \eqref{comp_zero}}
    \label{fig:all_cases}
\end{figure}

On the other hand, from the partition function interpretation of $M_{\tau, v_1, v_2}$, we have 
\begin{align}
M_{\tau, v_1,v_2} &= ~q^{(k(\tau)-1)t(v_2+v_1)}  \sC(\vec b = (v_1, \tau, v_2))   \nn 
&+  \sum_{v_1'\neq v_1, v_2'\neq v_2}q^{(k(\tau)-1)t(v_2+v_1)}  \sC(\vec b = (v_1', \tau, v_2')) \frac{1}{q^{d(\eta, \tau)t(|v_1'-v_1| + |v_2'-v_2| )}} \nn 
& + q^{(k(\tau)-1)t(v_2+v_1)}  \sC(b_i = \eta) \frac{1}{q^{d(\eta, \tau)t(v_1 +v_2)}} \nn 
& + \sum_{\sigma \in \sS_n, \sigma \neq \tau, \eta; v_1'<v_1, v_2'<v_2} q^{(k(\tau)-1)t(v_2+v_1)}  \sC(\vec b = (v_1', \sigma, v_2')) \frac{1}{q^{d(\sigma, \tau)t(v_1'+v_2')}} \frac{1}{q^{d(\tau, \eta)t(v_1-v_1'+v_2-v_2')}}  \nn
& + \sum_{\sigma \in \sS_n, \sigma \neq \tau, \eta; v_1'>v_1, v_2'>v_2} q^{(k(\tau)-1)t(v_2+v_1)}  \sC(\vec b = (v_1', \sigma, v_2')) \frac{1}{q^{d(\sigma, \tau)t(v_1+v_2)}} \frac{1}{q^{d(\sigma, \eta)t(v_1'-v_1+v_2'-v_2)}} \label{comp_zero}
\end{align} 
This expression accounts for contributions with up to two domain walls in $\vec b$. The corresponding domain wall configurations are shown in Fig.~\ref{fig:all_cases}. In the first line, the domain walls in $\vec b$ are aligned with those in the bottom boundary condition, so that we get 1 from their overlaps. In the second line, the last factor comes from non-trivial overlaps between $\vec b$ and the lower boundary conditions. Using the the fact that $\tau \in P_{\rm NC}$, the second line can be rewritten as  
\be 
 \text{Line 2} = \sum_{v_1'\neq v_1, v_2'\neq v_2}q^{(k(\tau)-1)t(v_2-|v_2'-v_2|+v_1-|v_1'-v_1|)}  \sC(\vec b = (v_1', \tau, v_2')) 
\ee
and the third line, which comes from cases where there are no domain walls in $\vec b$ (not explicitly shown in the figure), can be written as $C(b_i=\eta)$, so that from the discussion of section~\ref{sec:no_dw}, 
\be 
\text{Line 3} = O(e^{-\Delta t}) \, . 
\ee

The fourth and fifth lines, where the domain wall in $\vec b$ involves a different permutation $\sigma \neq \tau$, require some more work. First, using the fact that $d(\tau, \eta)= k(\tau)-1$, 
\be \label{line4def} 
\text{Line 4} = \sum_{\sigma \in \sS_n, \sigma \neq \tau, \eta; v_1'<v_1, v_2'<v_2} q^{(d(\tau, \eta) - d(\sigma , \tau))t(v'_2+v'_1)}  \sC(\vec b = (v_1', \sigma, v_2')) 
\ee
Now using \eqref{g53}, we have   
\begin{align}
|\text{Line 4}| &\leq  \sum_{\sigma \in \sS_n, \sigma \neq \tau; \eta, v_1'<v_1, v_2'<v_2} q^{t\le((d(\tau, \eta) - d(\sigma , \tau))v'_1 -d(\sigma, \eta) \sE_{d(\sigma, \eta)+1}(v_1')\ri)}
q^{t\le((d(\tau, \eta) - d(\sigma , \tau))v'_2 -d(\sigma, \eta) \sE_{d(\sigma, \eta)+1}(v_2')\ri)}\nn 
&
\leq \sum_{\sigma \in \sS_n, \sigma \neq \tau; \eta, v_1'<v_1, v_2'<v_2} q^{t\le((d(\tau, \eta) - d(\sigma , \tau) - d(\sigma, \eta))\sE_{d(\sigma, \eta)+1}(v_1')\ri)} q^{t\le((d(\tau, \eta) - d(\sigma , \tau) - d(\sigma, \eta))\sE_{d(\sigma, \eta)+1}(v_2')\ri)}
\end{align}
where in the second step, we have used the first inequality of \eqref{membrane_constraints} and assumed that $d(\eta, \tau)> d(\sigma, \tau)$ (as the alternative gives an exponentially suppressed contribution). Both factors are the same up to the replacement $v_1' \leftrightarrow v_2'$, so we can focus on either of them. 
Now by the triangle inequality, for any $\sigma \in \sS_n$, 
\be 
d(\tau, \eta) - d(\sigma , \tau) - d(\sigma, \eta) \leq  0  
\ee
with equality only in the case where $\sigma$ lies on the geodesic between $\tau$ and $\eta$. Hence, for any $v_1', v_2'$, Line 4 is upper-bounded by the case with the smallest value of $\sE_n(v_i')$, which is always at $v_1'=v_2'=0$. So up to an $O(1)$ prefactor, 
\be 
\text{Line 4} \leq \sum_{\sigma \neq \tau, \eta} q^{-2  \sE_{d(\sigma, \eta)+1}(0) \, \le( d(\sigma , \tau) + d(\sigma , \eta) - d(\tau, \eta) \ri)\, t} \, .  \label{line4}
\ee
Now for the fifth line, 
\begin{align}
\text{Line 5} &=  \sum_{\sigma \in \sS_n, \sigma \neq \tau, \eta;  v_1'>v_1, v_2'>v_2} C(\vec b = (v_1', \sigma, v_2')) q^{-d(\eta, \sigma)t(v_1'+v_2')} q^{(v_1+v_2)t(d(\eta, \tau)+ d(\eta, \sigma) -d(\sigma, \tau))} \label{line5def} \\ 
&\leq  \sum_{\sigma \in \sS_n, \sigma \neq \tau, \eta; v_1'>v_1, v_2'>v_2} e^{-d(\eta, \sigma)t(v_1'+v_2'+\sE_{d(\sigma, \eta)+1}(v_1')+\sE_{d(\sigma, \eta)+1}(v_2'))} q^{(v_1+v_2)t(d(\eta, \tau)+ d(\eta, \sigma) -d(\sigma, \tau))}\nn 
&\leq  \sum_{\sigma \neq \tau, \eta} e^{-d(\eta, \sigma)t(v_1+v_2+\sE_{d(\sigma, \eta)+1}(v_1)+\sE_{d(\sigma, \eta)+1}(v_2))} q^{(v_1+v_2)t(d(\eta, \tau)+ d(\eta, \sigma) -d(\sigma, \tau))}\nn 
&\leq  \sum_{\sigma \neq \tau, \eta} e^{-d(\eta, \sigma)t(\sE_{d(\sigma, \eta)+1}(v_1)+\sE_{d(\sigma, \eta)+1}(v_2))} q^{(v_1+v_2)t(d(\eta, \tau) -d(\sigma, \tau))}\nn 
&
\leq \sum_{\sigma \neq \eta, \tau} q^{-2  \sE_{d(\sigma, \eta)+1}(0) \, \le( d(\sigma , \tau) + d(\sigma , \eta) - d(\tau, \eta) \ri)\, t} \, . \label{line5}
\end{align}
In going from the second to the third line, we have used the fact that for the range of $v_1',v_2'$ allowed in the sum, the maximum value of the first factor is at $v_1'=v_1, v_2'=v_2$ (and there is an implicit $O(1)$ factor in the third line from summing over velocities). In principle, we should have written two additional cases $v_1'>v_1, v_2'<v_2$ and $v_1'<v_1, v_2'>v_2$ in \eqref{comp_zero}, but it should now be clear that these cases obey the same upper bound as Lines 4 and 5.

Now, suppose we take $\tau_{\rm min}$ to be the permutation in $Q_{\rm singleton} \cap P_{\rm NC}$ that lies closest to $\eta$ in Cayley distance, so that there are no other $\sigma \in Q_{\rm singleton} \cap P_{\rm NC}$ that can lie on the geodesic between $\eta$ and $\tau$. In fact, Lemma \ref{lemma_noncross_2} from Sec.~\ref{sec:background} implies that this case, there exists {\it no}  $\sigma \in \sS_n$ that lies on the geodesic between $\eta$ and $\tau$. Hence, \eqref{line4} and \eqref{line5} are of the form $e^{-\Delta t}$ for $\Delta>0$ independent of $v$. Putting together our conclusions about the different lines (and assuming as before that contributions with more domain walls in $\vec b$ are further suppressed), we find that for $\tau=\tau_{\rm min}\in Q_{\rm singleton}\cap P_{\rm NC}$, for any $v_1$, $v_2$, 
\begin{align}
&q^{(k(\tau)-1)t(v_2+v_1)}  \sC(\vec b = (v_1, \tau, v_2)) + \sum_{v_1'\neq v_1, v_2'\neq v_2}q^{(k(\tau)-1)t(v_2-|v_2'-v_2|+v_1-|v_1'-v_1|)}  \sC(\vec b = (v_1', \tau, v_2'))\nn   &= O(e^{-\Delta t}) 
\end{align}
This can also be written as 
\be
\sum_{v_1'\leq  v_1, v_2'\leq v_2}q^{(k(\tau)-1)t(v'_2+v'_1)}  \sC(\vec b = (v'_1, \tau, v'_2)) + O(e^{-\Delta'} t) = O(e^{-\Delta t}) \label{g71}
\ee
On the LHS, we have used the fact that for $v_1'>v_1$  (similarly for $v_2'> v_2$), 
\be 
q^{(k(\tau)-1)t(v_2-|v_2'-v_2|+v_1-|v_1'-v_1|)}  \sC(\vec b = (v_1', \tau, v_2')) \leq e^{-|v_1'-v_1|t}
\ee
and we have taken $\Delta'$ to be some ``UV cutoff'' in a discretization of the sum over $v_1', v_2'$. 
Now using \eqref{g71} and iteratively considering increasing values of $v_1, v_2$, we can see that for each $v_1, v_2$, 
\be
q^{(k(\tau)-1)t(v_2+v_1)}  \sC(\vec b = (v_1, \tau, v_2))= O(e^{-\Delta t}) \, , \quad \tau = \tau_{\rm min}. \label{taumin} 
\ee

Then consider the next $\tau_{\rm next} \in Q_{\rm singleton}  \cap P_{\rm NC}$ such that the only permutation lying on the geodesic between $\tau_{\rm next}$ and $\eta$ is $\tau_{\rm min}$. For all $\sigma \neq \tau,\tau_{\rm min}$ in Line 4 or Line 5, the same analysis leading to \eqref{line4} and \eqref{line5} applies and leads to an $O(e^{-\Delta t})$ contribution. \eqref{taumin} implies that the contribution of $\sigma = \tau_{\rm min}$ to  Line 4 and Line 5 is also $O(e^{-\Delta t})$. For example, using the expression for Line 4 in \eqref{line4def}, the contribution of $\sigma = \tau_{\rm min}$ is 
\be 
\text{Line 4} \supset q^{t \, d(\eta, \tau_{\rm min})(v_1'+v_2')} C(\vec b=(v_1', \tau_{\rm min}, v_2'))
\ee
which is $O(e^{-\Delta t})$ from \eqref{taumin}. 
%we have  shown that \eqref{line4def} is $O(e^{-\Delta t})$ in \eqref{taumin}. For Line 5, the contribution of $\sigma = \tau_{\rm min}$ is (using \eqref{line5def})

We can then iteratively apply this analysis to other $\tau\in Q_{\rm singleton}\cap P_{\rm NC}$
at increasing Cayley distance from $\eta$, to find that 
\be 
q^{(k(\tau)-1)t(v_2+v_1)}  \sC(\vec b = (v_1, \tau, v_2))= O(e^{-\Delta t})  \label{all_assumption_corrected} \text{ for all } \tau \in Q_{\rm singleton}  \cap P_{\rm NC}\,  . 
\ee

We can think of this suppressed value of $q^{(k(\tau)-1)t(v_2+v_1)}  \sC(\vec b = (v_1, \tau, v_2))$ as coming from certain subtle cancellations among the kinds of configurations shown in Fig.~\ref{fig:additional_contributions}, similar to the discussion in Sec.~\ref{sec:no_dw}.

\subsection{Final result}

Let us now return to the full sum over two-domain wall contributions $\otoc_n$, \eqref{2dwdef}. First, we note that all $\tau \in P_{\rm even} \cap P_{\rm NC}$ also lie in $Q_{\rm singleton}\cap P_{\rm NC}$. Hence, the $v$-independent contribution coming from $\sZ_{\rm even}$ in $\otoc_n$ is $O(e^{-\Delta t})$ and can be ignored. Certain permutations in ${\bar P_{\rm even}\cap \rm NC}$ also lie in $Q_{\rm singleton}\cap P_{\rm NC}$, and these contributions can also be ignored. If we assume that the membrane formula \eqref{c_simple} can be applied to other permutations (i.e., there are no subtle cancellations in these cases), we end up with 
\be \label{final}
\otoc_{n, \, \text{2 DW}} = \sum_{\tau \in \bar Q_{\rm singleton} \cap P_{\rm NC}} f_{\tau} \le(\sum_{0<v_1<v}e^{-(k(\tau)-1)(\sE_{k(\tau)}(v_1)-v_1)t}\ri) \le(\sum_{v_2>0}e^{-(k(\tau)-1)(\sE_{k(\tau)}(v_2)-v_2)t}\ri)
\ee
Here we have restored the $O(1)$ factors $f_{\tau}$, which may be positive or negative. 
Now consider the sums over $v_1$, $v_2$ for each $\tau$. Recall from \eqref{membrane_constraints} that $\sE_n(v)-v$ for any $v$ is minimized at $v=v_B$, where its value is zero. Since the sum over $v_2$ always includes $v_B$, we can set the factor in the second parenthesis in \eqref{final} to 1. The sum over $v_1$ can be set equal to $v_B$ if $v>v_B$. Otherwise, it $\sE(v_1)-v_1$ is minimized at the endpoint value $v_1 =v$.  

Now recall from Lemma~\ref{lemma_noncross_1} that for $\tau \in \bar Q_{\rm singleton}\cap P_{\rm NC}$, $k(\tau)\geq n/2+1$. Collecting contributions from different $\tau$ with the same values of $k(\tau)$, we get 
\be 
\otoc_n(t, v) = \begin{cases}
\sum_{k=n/2}^{n-1} c_k & v > v_B\\
\sum_{k=n/2}^{n-1} c_k q^{-k(\sE_{k+1}(v)-v)t}  & v < v_B 
\end{cases}
\label{otoc_final_result}
\ee
From the first line, we see that the coefficients $c_k$ must add up to 1 for continuity at $v=1$. Since we get this non-trivial result from the two-domain wall contributions, and cases with additional domain walls are subleading from \eqref{three_dw}, \eqref{otoc_final_result} is our prediction for the leading behaviour of $\otoc_n$ in the regime $v^{\ast}<v<1$ defined in Sec.~\ref{sec:regimes}.

We can further expand the second line close to $v=v_B$. The second and third equations of \eqref{membrane_constraints} then imply that 
\be 
\otoc_n(t, v) = \sum_{k=n/2}^{n-1} c_k q^{-k \frac{\sE''_{k+1}(v_B)}{2}(v -v_B)^2 t} 
\ee
It is expected that close to $v=v_B$, interactions among various domain walls vanish, so that $\sE_n(v)$ becomes independent of $n$ in this regime. Based on this, setting $\sE''_{k+1}(v_B)/\log q$ to be some $k$-independent constant $1/D$, we get the expression
\be 
\otoc_n(t, v) = \sum_{k=n/2}^{n-1} c_k e^{-k \frac{(v -v_B)^2}{2D} t} \, . \label{final_rc}
\ee
which is explicit apart from the constants $c_k$ and $D$. $D$ can be viewed as a diffusion constant for the broadening of the operator front~\cite{operator_spreading_adam, operator_spreading_tibor}. For large enough $t$, this series can be approximated with its leading $k=n/2$ term, giving the $n$-dependence in \eqref{very_near_front}.

\section{Combinatorial lemmas}

\label{app:perm_proofs}

\lemmagone*

\begin{remark}\label{remark:noncross}
    The same claim holds for any NC partition $\tau$ of $n$ elements where $n$ is even: (a) If $\tau$ has all blocks of even sizes, and then its Kreweras complement $\tau^c$ has more than $n/2+1$ blocks with at least one singleton block; and (b) if $\tau^c$ has no singletons, then $\tau$ has $n/2+1$ blocks and $\tau$ has at least one odd-size block.

\end{remark}

\begin{proof} 
 Since $\tau \in P_{\rm even}$, $|c_i|\geq 2$ for each cycle $c_i \in \tau$. We have 
\be 
n = \sum_{c_i \in \tau} |c_i| \geq 2 k(\tau) \, . 
\ee
Since we also have $\tau \in {\rm NC}$, 
    \begin{equation}
        k(\eta\tau^{-1}) = n+1 - k( \tau)\ge n/2+1\ . \label{g73}
    \end{equation}
    Now suppose that $\eta\tau^{-1}$ has no singleton (cycle of length 1), then we have
    \begin{equation}
         n=\sum_{d_i\in\eta\tau^{-1}}{|d_i|}\ge 2k(\eta \tau^{-1})\implies k(\eta\tau^{-1})\leq n/2
    \end{equation}
    which is a contradiction with \eqref{g73}. Hence, $\eta\tau^{-1}$ must have at least one singleton. This proves claim (a).

From (a), we have that 
\be
P_{\rm even}\cap P_{\rm NC}\subseteq  Q_{\rm singleton} \cap P_{\rm NC}\, .
\ee
Hence, for any $\tau\in\bar Q_{\rm singleton}$, we have $\tau\notin P_{\rm even}$.
\iffalse
Furthermore, we have 
\be 
n = \sum_{c_i \in \eta\tau^{-1}} |c_i| \geq 2 k(\eta\tau^{-1}) \, . 
\ee
and then
\be
k(\tau) = n+1- k(\eta\tau^{-1}) \ge n/2 +1 \ .
\ee
This proves claim (b).
\fi

    The proofs of (b) and of the remark is essentially identical.
\end{proof}

\lemmagtwo*

\begin{proof}
We first prove $\sigma\in \rm{NC}$ and then $\sigma\in Q_{\rm singleton}$.

A permutation is NC if and only if it lies on the geodesic between $\eta$ and $e$. $\tau\in {\rm NC}$ implies that $\tau$ is on the geodesic between $\eta$ and $e$. 
\begin{equation}
    d(\eta,e) = d(\eta,\tau)+d(\tau,e)
\end{equation}
\eqref{eq:sigma_on_geodesic} implies that $\sigma$ is on the geodesic between $\eta$ and $\tau$. We substitute $d(\eta, \tau)$ using \eqref{eq:sigma_on_geodesic},
\begin{equation}
    d(\eta,e) =  d(\eta, \sigma) + d(\sigma, \tau) +d(\tau,e)\geq d(\eta, \sigma)+d(\sigma,e)
\end{equation}
where we applied the triangle inequality to eliminate $\tau$. Together with the triangle inequality 
\begin{equation}
     d(\eta,e)\leq d(\eta,\sigma)+d(\sigma,e)\ ,
\end{equation}
we conclude that
\begin{equation}
    d(\eta,e)= d(\eta,\sigma)+d(\sigma,e)\implies \sigma\in\rm{NC}\ .
\end{equation}

We further have
\begin{equation}
    d(\eta,\sigma)=n-k(\eta\sigma^{-1}) = k(\sigma)-1,\quad d(\eta,\tau)=n-k(\eta\tau^{-1}) =k(\tau)-1\ ,
\end{equation}
which together with \eqref{eq:sigma_on_geodesic} implies that
\begin{equation}\label{eq:distance_sigma_tau}
    d(\sigma,\tau) = k(\tau)-k(\sigma)=k(\eta\sigma^{-1})-k(\eta\tau^{-1})\ .
\end{equation}
%These equalities above show that every transposition step on the geodesic from $\eta$ to $\tau$ splits one cycle into two and increases the number of cycles $k$ by one. A single transposition either splits an odd cycle into an even cycle and an odd cycle; or it splits an even cycle into two even cycles or two odd cycles. Hence, the number of odd cycles never decreases along the geodesic from $\eta$ to $\tau$. Since we start at $\tau$ and end at $\tau$, neither of which has any odd cycle, we conclude that $\sigma$ in between must also have no odd cycle. We have $\sigma\in X\cap {\rm NC}$.

As we move $\sigma$ along the geodesic $\eta\to\tau$, the Cayley distance $d(\eta,\sigma)$ measures the number of transpositions taken to reach $\sigma$ and $d(\eta,\tau)$ is the number of transpositions it takes to reach $\tau$. \eqref{eq:distance_sigma_tau} implies that $k(\eta\sigma^{-1})$ is monotonically decreasing from $k(\eta\eta^{-1})=n$ to $k(\eta\tau^{-1})=n+1-k(\tau)$, and its value drops by $1$ at every transposition from $\sigma=\pi_i\to\pi_{i+1}$, where $\pi_{i+1}$ has one more transposition than $\pi_i$. 

Note that $\pi_{i+1}=t\pi_i$ for some transposition $t$ implies that $\eta\pi_{i+1}^{-1}=\eta\pi_i^{-1}t^{-1}$, which is one more transposition than $\eta\pi_i^{-1}$. Since a transposition can either merge two cycles into one or split one cycle into two, and $k(\eta\sigma^{-1})$ decreases by 1 per transposition step, $\eta\pi_{i+1}^{-1}$ has two cycles in $\eta\pi_i^{-1}$ merged into one.

It's obvious that the number of singletons does not increase under merging. Since we assume that there is at least one singleton in $\eta\tau^{-1}$, we conclude that the number of singletons in $\eta\sigma^{-1}$ along the geodesic is no fewer. We have $\sigma\in Q_{\rm singleton}$.

\end{proof}

\begin{comment}
However, since 
\be 
|J(\sigma, \tau; \nu)| <1 \text{ unless } \sigma = \tau = \nu \, ,  \label{j_ineq}
\ee
we expect that configurations with  an extensive number of  domain walls give subleading contributions to ${\rm OTOC}_n$ compared to those where the number of triangles with domain walls is $O(1)$ or $O(t)$. Two things we need to show here: 
\begin{enumerate}
    \item Proving \eqref{j_ineq}: we have
    \begin{align}
|J|< C_n \sum_{\mu} (q^2)^{-d(\mu, \nu)} (q)^{-d(\mu, \sigma)}  (q)^{-d(\mu, \tau)} \leq C_n n! (q^2)^{-2d(\mu, \nu)- d(\mu, \sigma)-d(\mu, \tau)} \leq \frac{C_n n!}{q^2}
    \end{align}
    which is $\leq 1$ as long as $q$ is large and $n$ does not scale with $q$, but I expect it should be true for all $n$ and $q$.  
    \item While configurations with many domain walls will have a higher cost per configuration if $|J|<1$ for an extensive number of triangles in such configurations, they could in principle end up having a higher entropic/combinatorial factor due to the larger number of possibilities for the precise locations of the domain walls and the permutations that appear between them. We should argue that the entropic factor does not win by considering some examples. 
\end{enumerate}
\end{comment} 
%background on 
\begin{lemma}
\label{lemma:inequality}
    For $J(\sigma, \tau; \nu)$ defined as in \eqref{jdef}, or more explicitly, 
    \be 
    J(\sigma, \tau; \nu) = \sum_{\mu} g^{\nu \mu} \sqrt{g_{\mu \sigma} g_{\mu \tau}}
    \ee
for $g_{\mu\nu} \equiv d^{n-k(\mu \nu^{-1})}$, $g^{\mu\nu}$ defined as the inverse of $g_{\mu\nu}$ (or in terms of the Weingarten functions, $g_{\mu\nu} = d^n{\rm Wg}(\mu \nu^{-1})$), and $d=q^2$, we have
\be 
J(\sigma, \tau; \nu)  \leq 1 
\ee
for $q=2$, and for all $n$ and all $\sigma, \tau, \nu \in \sS_n$. 
The above inequality is saturated only when $\sigma = \tau = \nu $. 
\end{lemma}
\begin{proof} 
    When $\sigma=\tau$, we have
\begin{equation}
    J(\sigma, \tau; \nu) = g^{\mu\nu}g_{\sigma\mu}= \delta_{\nu\sigma}\leq 1
\end{equation}
and it is saturated when $\sigma = \tau = \nu $.

Now we discuss the case of  $\sigma\neq\tau$.
We define the following functions
\begin{equation}
    w(\pi):=d^n\mathrm{Wg}(\pi;d),\quad s(\mu):=\sqrt{g_{\sigma\mu}g_{\tau\mu}}=\sqrt{d^{k(\tau\mu^{-1})-n}d^{k(\sigma\mu^{-1})-n}}\ .
\end{equation}
We can write
\begin{equation}
    J(\sigma, \tau; \nu) =\sum_{\mu\in S_n}w(\mu\nu^{-1})s(\mu)
\end{equation}
The Cauchy-Schwarz inequality gives,
\begin{equation}
    |J(\sigma, \tau; \nu)| =|\sum_{\mu\in S_n}w(\mu\nu^{-1})s(\mu)|\leq\sqrt{\sum_{\mu\in S_n}|w(\mu\nu^{-1})|^2\sum_{\mu\in S_n}|s(\mu)|^2} \ .
\end{equation}
Since the norm is group translation-invariant, we have
\begin{equation}
   \sqrt{\sum_{\mu\in S_n}|w(\mu\nu^{-1})|^2\sum_{\mu\in S_n}|s(\mu)|^2} = \sqrt{\sum_{\pi\in S_n}|w(\pi)|^2\sum_{\mu\in S_n}|s(\mu)|^2} =||w||_2||s||_2\ .
\end{equation}
Therefore, we obtain
\begin{equation}
    |J(\sigma, \tau; \nu)|\leq ||w||_2||s||_2
\end{equation}

The Weingarten function can be written as an expansion over the irreps of $S_n$~\cite{collins2003moments}:
\begin{equation}
    \mathrm{Wg}(\pi;d) = \frac{1}{n!^2}\sum_{\lambda\vdash n:|\lambda|\leq d}\frac{\chi^\lambda(1)^2\chi^\lambda(\pi)^2}{s_{\lambda}(1)}
\end{equation}
where the sum is over partitions $\lambda$ of $n$ that has at most $d$ parts. A partition $\lambda$ ordered decreasingly can be identified with a Young diagram with $|\lambda|$ rows that label the irreps. $\chi^\lambda$ is the character of the $\lambda$-irrep of $S_n$ and $s_\lambda$ is the Schur polynomial. 

Using this formula and the orthogonal relations of the characters, one can obtain
\begin{equation}
    ||w||_2^2 = \frac{1}{n!}\sum_{\lambda\vdash n:|\lambda|\leq d} d_\lambda^2\kappa_\lambda^2,\quad \kappa_\lambda=\frac{d^n}{\prod_{(i,j)\in\lambda}(d+j-i)}
\end{equation}
where $d_\lambda$ is the dimension of the $\lambda$-irrep of $S_n$, and the product in the denominator of $\kappa_\lambda$ is over the box indices $(i,j)$ in the Young diagram corresponding to $\lambda$. 

We first proceed with the estimate:
\begin{equation}
    ||w||_2=\frac{1}{n!}\sum_{\lambda\vdash n:|\lambda|\leq d} d_\lambda^2\kappa_\lambda^2\leq\max_{\lambda\vdash n:|\lambda|\leq d}\kappa_\lambda\left(\frac{1}{n!}\sum_{\lambda\vdash n:|\lambda|\leq d} d_\lambda^2\right) \leq\max_{\lambda\vdash n:|\lambda|\leq d}\kappa_\lambda\ .
\end{equation}

To bound $\max_{\lambda\vdash n:|\lambda|\leq d}\kappa_\lambda$, we group the Young diagram column-wise, and denote $h_j$ the height of column $j$,
\begin{equation}
    \prod_{(i,j)\in\lambda}(d+j-i) = \prod_j \prod_i^{h_j} (d+j-i)\ge \prod_j j^{h_j}
\end{equation}
Let $n=pd+r$, with $r= n\mod d$, $\prod_j j^{h_j}$ is minimized by filling columns greedily up to height $d$, so we have $p$ columns of height $d$ and one column of height $r$. This gives
\begin{equation}
    \prod_j j^{h_j}\ge(p+1)^r\prod_{j=1}^p j^d=(p!)^d(p+1)^r
\end{equation}
and hence
\begin{equation}
    \max_{\lambda\vdash n:|\lambda|\leq d}\kappa_\lambda\leq\frac{d^n}{(p!)^d(p+1)^r},\quad n=pd+r\ .
\end{equation}

To bound $||s||_2$, we first rewrite $s$ as follows,
\begin{equation}
    s(\mu)=h_\sigma(\mu)h_\tau(\mu):=h(\sigma\mu^{-1})h(\tau\mu^{-1}),\quad h(\pi):=\sqrt{d^{k(\pi)-n}}
\end{equation}
Using H\"older's inequality and translational invariance of the norm, we have
\begin{equation}
    ||s||_2\leq ||h_\sigma||_4||h_\tau||_4 = ||h||_4^2\, .
\end{equation}
We have
\begin{equation}
    ||h||_4^4 = \sum_{\pi\in S_n} h(\pi)^4 =\sum_{\pi\in S_n} d^{2k(\pi)-2n}= d^{-2n}\prod_{i=0}^{n-1}(d^2+i)=n! d^{-2n} \binom{n+d^2-1}{n}
\end{equation}
where in the penultimate step, we used $\sum_\pi x^{k(\pi)}=\prod_{i=0}^{n-1}(x+i)$. 
\begin{figure}
    \centering
    \includegraphics[width=0.7\linewidth]{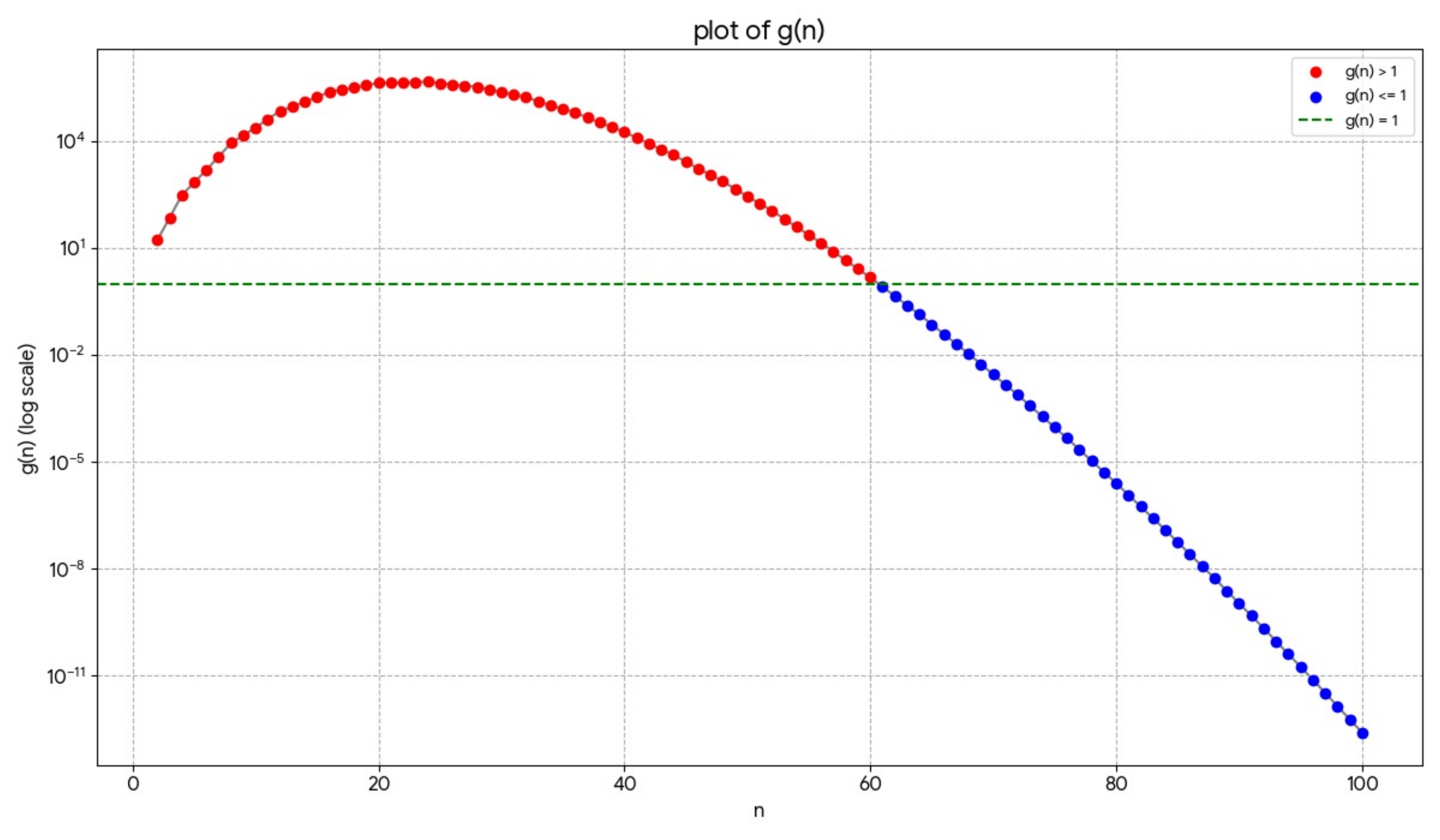}
    \caption{Values of the bound $g(n)$ on $|J|$ for $2\leq n\leq 100$.}
    \label{fig:bound_estimate}
\end{figure}

We have
\begin{equation}
    ||s||_2\leq ||h||_4^2 =\sqrt{n!} \,d^{-n} \binom{n+d^2-1}{n}^\frac12
\end{equation}

Altogether, we obtain
\begin{equation}
    |J(\sigma, \tau; \nu)|\leq||w||_2||s||_2 \leq \max_{\lambda\vdash n:|\lambda|\leq d}\kappa_\lambda\cdot ||s||_2\leq \frac{\sqrt{n!}\binom{n+d^2-1}{n}^\frac12}{(p!)^d(p+1)^r} =:g(n)
\end{equation}
This bound $g(n)$ is loose for small $n$, but it gets tighter for large $n$.

Let us focus on the case of $d=q^2=4$. A quick ratio test on the bound shows that once the bound $g(n_0)<1$ for some $n_0>26$, it remains $g(n)<1$ for all $n>n_0$. We can then evaluate the bound and explicitly verify that the bound is less than $1$ for all $n\ge 61$. This is shown in Fig.~\ref{fig:bound_estimate}.

Hence, for $d=4$, what remains is to verify that $|J|\leq 1$ for all $2\leq n\leq 60$. We can further tighten the gap without resorting to the looser bound.\footnote{We did so because $g(n)$ is elementary enough for the ratio test.}
\begin{equation}
    |J(\sigma, \tau; \nu)| \leq ||w||_2||s||_2
\end{equation}
Direct evaluation of the RHS for all $2\leq n\leq 60$ gives a less-than-one value for $20\leq n\leq 60$.

It remains to brute-force check $|J|\leq 1$ for all $2\leq n\leq 19$. We verified these cases in Python and this completes the proof of the Lemma for the case of $d=4$. 
\end{proof}

\section{Singular spectral density of $A(t)B$ in Transverse Field Ising Model}\label{app:tfim}

The physics of the transverse field Ising model (TFIM) is most transparent in terms of fermionic variables defined by the Jordan-Wigner transformation of \eqref{tfim}. We define the Majorana fermion operators as follows:
\begin{equation}
    \gamma_{2j-1}\equiv\le(\prod_{i<j} Z_i\ri)X_j,\quad\gamma_{2j} \equiv \le(\prod_{i<j} Z_i\ri)Y_j\ ,\quad j\in[L]
\end{equation}
The Majorana fermion operators are Hermitian involutions $\gamma_i^\dagger=\gamma_i,\ \gamma_i^2=I$.  The Majoranas satisfy the canonical anticommutation relation (CAR) $\{\gamma_i,\gamma_j\}=2\delta_{ij}$.

The TFIM Hamiltonian can be written as a free Majorana fermion Hamiltonian,
\begin{equation}
    H = \frac{i}4\sum_{i,j=1}^{2L}\gamma_i K_{ij}\gamma_j
\end{equation}
where $K$ is a real antisymmetric matrix. With open boundary condition, $K_{2j-1,2j}=2h$, $K_{2j,2j+1}=2J$. Heisenberg evolution is represented as a linear transformation,
\begin{equation}
    \gamma_i(t) = U(t)\gamma_i U(t)^\dagger = \sum_jO_H(t)_{ij}\gamma_j ,\quad O_H(t):=e^{Kt}\in SO(2L) 
\end{equation}
The evolved Majoranas $\gamma_i(t)$ still satisfy the CAR. 

%\SV{[ I think we shouldn't really emphasize this single-particle Majorana space, because the Pauli operators are not restricted to this. The key point below is that even if we consider all possible multi-particle states that involve just 4 majorana fermions, this is ${\sqrt{2}^4}=4$-dimensional? It looks like \text{span}... here is just the subspace spanned by single particle states of the majoranas here, in which case it isn't the right notation to use below?}\JW{Ok, all we need is this fact Cl$(\rm{span}(\gamma_1,\gamma_2,\gamma_3,\gamma_4))\simeq\mathcal B((\mathbb C^2)^{\otimes 2})$, so then we know the operators in it can be represented as operators on a four-dim Hilbert space. Usually, people introduce $\mathcal V$ to build the Fock space $\mathcal F(\mathcal V)$ and then the operator (Clifford) algebra on it, but that's a bit redundant there. I think we could omit introducing the single-particle space and directly state the fact we need. } 

The Majoranas $\{\gamma_i\}_{i=1}^{2L}$, subject to the CAR, generate the Clifford algebra Cl$_{2L}(\mathbb C)$. It is a mathematical fact that Cl$_{2L}(\mathbb C)\simeq \mathcal B((\mathbb C^2)^{\otimes L})$, which refers to the algebra of bounded operators on $(\mathbb C^2)^{\otimes L}$, the physical Hilbert space of the spin chain. This is a formal statement of the fact that two Majorana fermions can be associated with the Hilbert space of one qubit.

We want to calculate the spectrum of $A(t)B$ to calculate the FMI and OTOCs. Let us focus on the case of $A$ and $B$ being \emph{single-site} Paulis. Let's first study the case of $A=Z_i$ and $B=Z_j$. We can write $Z$ in terms of the Majoranas
\begin{equation}
    A=Z_i = -i\gamma_{2i-1}\gamma_{2i},\quad  B=Z_j = -i\gamma_{2j-1}\gamma_{2j}\ .
\end{equation}
We have
\begin{align}
    A(t)B &= - U(t)\gamma_{2i-1}\gamma_{2i}U(t)^\dagger\gamma_{2j-1}\gamma_{2j} = - U(t)\gamma_{2i-1}U(t)^\dagger U(t)\gamma_{2i}U(t)^\dagger\gamma_{2j-1}\gamma_{2j} 
    \nn &=-\gamma_{2i-1}(t)\gamma_{2i}(t)\gamma_{2j-1}\gamma_{2j}\ .
\end{align}

Therefore, $A(t)B$ belongs to the Clifford subalgebra Cl$_{4}(\mathbb C)$ generated by four Majoranas $\{\gamma_{2i-1}(t),\gamma_{2i}(t),\gamma_{2j-1},\gamma_{2j}\}$.
%Hence, we have Cl$(\{\gamma_{2i-1}(t),\gamma_{2i}(t),\gamma_{2j-1},\gamma_{2j}\})\subset$ Cl$(\mathcal V)$ and so it acts irreducibly only on a subspace of $\mathcal H$. %Explicitly, one can always find a Gaussian unitary transformation that maps an orthonormal basis of $\mathcal W$ to $\{\gamma_1,\gamma_2,\gamma_3,\gamma_4\}$. 
Then we use the mathematical fact that Cl$_{4}(\mathbb C)\simeq\mathcal B((\mathbb C^2)^{\otimes 2})$, and thus $A(t)B$ can be irreducibly represented as acting on a four-dimensional subspace. We know that generally the eigenvalues of $A(t)B$ always come in pairs $\{e^{i\pm \theta_i}\}_i$. It follows that there exist two angles $\theta_1$ and $\theta_2$ and the four eigenvalues are
\begin{equation}
    \{e^{i\pm \theta_1},e^{i\pm \theta_2}\}\ .
\end{equation}
Hence, the spectrum of $A(t)B$ represented on the full physical Hilbert space is largely degenerate. It follows that the FMI of $Z_i(t)$ and $Z_j$ remains divergent throughout the TFIM evolution.

The same analysis holds for the $A=X_i$ and $B=X_j$ case. All we need to do is to use a different convention for the Jordan-Wigner transformation:
\begin{equation}
    \tilde\gamma_{2j-1}:=\le(\prod_{i<j} X_i\ri)Z_j,\quad\tilde\gamma_{2j} = \le(\prod_{i<j} X_i\ri)Y_j\ ,\quad j\in[L]
\end{equation}
Then $X$ can be written as a product of two Majoranas: $X_j=-i\tilde\gamma_{2j-1}\tilde\gamma_{2j}$. Hence, $X_i(t)X_j$ has at most four distinct eigenvalues. The same reasoning applies to the $A=Y_i$ and $B=Y_j$ case.

\bibliographystyle{JHEP.bst}
\bibliography{biblio.bib}
   \end{document}